\documentclass[a4paper,UKenglish,cleveref, autoref, thm-restate]{lipics-v2021}

\title{The Complexity of Nested Reset Counter Systems} 

\author{A. R. Balasubramanian}{Max Planck Institute for Software Systems (MPI-SWS), Kaiserslautern, Germany \and \url{https://arbalan96.github.io/} }{bayikudi@mpi-sws.org}{https://orcid.org/0000-0002-7258-5445}{This research was sponsored in part by the Deutsche Forschungsgemeinschaft project \href{https://gepris.dfg.de/gepris/projekt/389792660}{389792660} TRR 248–CPEC}

\author{Franzisco Schmidt}{Technical University of Munich (TUM), Munich, Germany}{franzisco.schmidt@tum.de}{}{}

\authorrunning{A. R. Balasubramanian and F. Schmidt} 

\Copyright{A. R. Balasubramanian and Franzisco Schmidt} 



\category{} 

\relatedversion{} 

\nolinenumbers 

\EventEditors{Claudia Faggian and Joost-Pieter Katoen}
\EventNoEds{2}
\EventLongTitle{41st Annual Symposium on Logic in Computer Science (LICS 2026)}
\EventShortTitle{LICS 2026}
\EventAcronym{LICS}
\EventYear{2026}
\EventDate{July 20--23, 2026}
\EventLocation{Lisbon, Portugal}
\EventLogo{}
\SeriesVolume{380}
\ArticleNo{60}

\usepackage{todonotes}
\usepackage{diagbox}
\usepackage{comment}
\usepackage{amssymb}
\usepackage{tikz}
\usepackage{stmaryrd}
\usepackage{mathtools}
\usepackage{mathalpha}
\usepackage{extarrows}
\usepackage{thmtools}
\usepackage{pifont}
\usepackage{tikz}
\usepackage{tcolorbox}
\usepackage{mdframed}
\usetikzlibrary{arrows.meta, bending, positioning}
\input{macros}

\begin{document}

\maketitle

\begin{abstract}
 Nested counter systems (NCS) are a generalization of counter systems to higher-order counters. Here, a higher-order counter is allowed to have other (lower-order) counters as elements, instead of  just a number. Such systems can be viewed as working on trees, where the height of the tree naturally corresponds to the highest order counter that the system is working with. It is known that the coverability problem for NCS, which asks if a given final tree can be covered from a given initial tree, is $\mathbf{F}_{\epsilon_0}$-complete. Here $\mathbf{F}_{\epsilon_0}$ is a class in the fast-growing hierarchy of complexity classes.

In this paper, we consider an extension of NCS called nested reset counter systems (NRCS) that extends NCS with resets. We show that coverability for NRCS over order-$k$ counters is $\mathbf{F}_{\Omega_k}$-complete where $\Omega_k$ is the tower of height $k$ of the $\omega$ ordinal. This gives the first natural hierarchy of complete problems for all of these classes.
Furthermore, to prove our upper bounds, we also develop length function theorems for any fixed amount of applications of the multiset operation on finite sets.

As an application of our results, we improve existing upper bounds for various problems from XML processing, graph transformation systems, $\pi$-calculus, logic and parameterized verification. Furthermore, using our completeness results for $k$-NRCS, we also prove $\mathbf{F}_{\Omega_k}$-completeness of the considered problems from the realms of parameterized verification and logic, for all $k$.
\end{abstract}

\section{Introduction}

The last decade has seen the rise of a significant research program in theoretical computer science, namely understanding the precise complexity of problems which are \emph{non-elementary}, i.e., problems whose running times cannot be bounded by any fixed tower of exponentials of the input size. This has resulted in the development of a rich collection of mathematical techniques in order to prove both upper and lower bounds for various decision problems.
Examples of such problems include reachability for Petri nets~\cite{LerouxS19,Leroux21,CzerwinskiO21}, coverability for lossy counter machines and lossy channel systems~\cite{SchmitzS11,FigueiraFSS11,Schmitz19}, emptiness of alternating 1-register and 1-clock automata~\cite{LasotaW08,DemriL09}, coverability for unordered data nets, $\nu$-Petri nets and various other enriched nets~\cite{Rosa-Velardo17,LazicS16,HaddadSS12}, analysis problems for fragments of $\pi$-calculus~\cite{Meyer08,Balasubramanian22,AmadioM02}, broadcast networks~\cite{Balasubramanian21,DelzannoSZ10}, bisimulation equivalence of pushdown automata~\cite{JancarS19}, rewriting systems~\cite{Hague14,GenestMSZ08} 
and logics~\cite{lics/Balasubramanian21,GreatiR24,GreatiR25,LazicS15,LazicOW16}. 
We refer the reader to the excellent survey by Schmitz~\cite{Schmitz16hierarchies}
for a wide collection of such problems from automata theory, logic, formal languages and verification, which have been proven to be complete for various classes in the \emph{fast-growing hierarchy} of complexity classes. This hierarchy allows for a finer classification of the complexity of problems lying beyond the elementary realm.

When such results are viewed from a tractability perspective, they are naturally negative. However, what we get out of them are fundamental insights that allow us to transfer techniques across a wide range of domains. Furthermore, there are various non-elementary problems for which tools have been developed (such as MONA for solving WS1S~\cite{ElgaardKM98}) and considerable effort has been spent in developing approximations for well-studied and important non-elementary problems such as the Petri net reachability problem~\cite{Blondin20,BlondinH17,BlondinHO21,BlondinFHH16,FracaH15,HaaseH14,EsparzaLMMN14}. This implies that understanding the precise complexity of a problem can help us solve it by reducing it to other well-studied problems from different fields. 

In this spirit, one of the most fundamental tools that allow us to understand the complexity of problems is the notion of \emph{completeness}. Indeed, once we have proven a problem to be complete for a certain class, we can now use this problem to prove hardness results for problems of interest to us, instead of starting from scratch, i.e., from Turing machines or Minsky machines. (A similar argument applies for obtaining upper bounds as well). Completeness is such a fundamental notion of complexity classes, that once a complexity class has been introduced, almost always ``natural'' complete problems for it are sought, i.e., complete problems that are not defined explicitly by Turing machines. 
For instance, in their invited paper for CONCUR 2013~\cite{SchmitzS13},
Schmitz and Schnoebelen explicitly state the program of investigating ``master'' complete problems for classes in the fast-growing hierarchy. Various results have been proven in the context of this program, allowing the complexity classes along this hierarchy to be populated with a wide variety of problems~\cite{Schmitz16hierarchies}.

This paper is a contribution to this line of research. We consider an infinite hierarchy of complexity classes and produce the first natural hierarchy of complete problems for all of these classes. More precisely, for every fixed exponential tower of the ordinal $\omega$ of height $k$ (hereafter denoted by $\Omega_k$), we provide a complete problem for the complexity class $\fastf{\Omega_k}$. Some of the milestones along this hierarchy include $\fastf{\Omega_1} = \fastf{\omega}$, also called the Ackermann class and $\fastf{\Omega_2} = \fastf{\omega^\omega}$,
also called the Hyper-Ackermann class. Prior to our results, (natural) complete problems were known only for the first 3 levels of this hierarchy. We generalize these results and provide complete problems \emph{for all of these levels}. 

The specific hierarchy of problems that we prove to be complete are inspired by a model of computation working over \emph{higher-order counters}. Intuitively, an 1-order counter is a usual counter over the natural numbers, which we can either increment or decrement by a constant. A 2-order counter is a counter in which we can increment or decrement 1-order counters. A 3-order counter similarly then works over a 2-order counter and so on. Such systems can be naturally viewed as labelled trees, where a tree of height 1 can model a collection of 1-order counters (The root models the state of the machine
and the number of  children of the root with label $\ell$ denotes the value of the counter $\ell$). A tree of height 2 can similarly model a collection of 2-order counters and so on. Such a model is called a nested counter system, which has been considered in~\cite{DeckerT16}.

In this paper, we extend this model to also allow \emph{reset operations} to the counters. For instance, over 1-order counters, this would simply correspond to resetting the value of a counter to 0, giving rise to the equivalent model of reset Petri nets~\cite{Schnoebelen10}. The resulting systems that we get out of  adding resets to nested counter systems are called \emph{nested reset counter systems} (NRCS). The main result of this paper is to show that the coverability problem for $k$-dimensional nested reset counter systems, is $\fastf{\Omega_k}$-complete for any $k$. This generalizes the known result for $k = 1$~\cite[Section 6.1.2]{Schmitz16hierarchies}.

In order to prove our upper bound results, we develop \emph{length function} bounds for \emph{nested multiset well-quasi-orders}. Intuitively, a nested multiset well-quasi-order is obtained by starting from the empty set and then closing it under disjoint sums as well as applications of any fixed number of the multiset operation on well-quasi-orders. As one of our main contributions, we provide upper bounds on the so-called \emph{controlled bad sequences} for such well-quasi-orders.

We demonstrate the applicability of our results and the NRCS model in understanding the complexity of other problems, by providing various applications: As a first step, we use the upper bound techniques developed in this paper to improve existing upper bounds in the literature for a hierarchy of problems from various fields such as XML processing~\cite{GenestMSZ08}, graph transformation systems~\cite{KoenigS}, $\pi$-calculus~\cite{Meyer08}, parameterized verification~\cite{DelzannoSZ10} and logic~\cite{DeckerT16}. Furthermore, we use our completeness result for NRCS to show that the hierarchy of problems from parameterized verification and logic are also complete for $\fastf{\Omega_k}$ for each $k$, thereby adding  more natural problems for all of these complexity classes.
This demonstrates the efficacy of NRCS to serve as master problems for all of these classes, thereby contributing significantly to the program of understanding decision problems beyond the elementary regime. 

\subparagraph{Related work.} The closest work to ours is~\cite{DeckerT16}, where the authors proved that the coverability problem for $k$-dimensional nested counter systems (without resets!) is $\fastf{\Omega_{k-1}}$-hard and in $\fastf{\Omega_{2k}}$. Our result significantly improves upon their upper bound by bringing it down to $\fastf{\Omega_k}$. 

Length function theorems for various well-quasi-orders are known in the literature - for instance for the product ordering over the natural numbers~\cite{FigueiraFSS11}, for the subword ordering~\cite{SchmitzS11}, for powersets and multiset of natural numbers~\cite{Rosa-Velardo17}, for the priority subword ordering~\cite{PriorityChannelSystems} and for multisets of words~\cite{GreatiR25}. However, none of those results except for the priority subword ordering are applicable in our setting, because they can only be used on problems lying in classes below $\fastf{\Omega_4}$.
The priority subword ordering is a general tool that can be used to place a problem in $\fastf{\epsilon_0}$ (roughly speaking, $\epsilon_0$ is the limit of $\Omega_1,\Omega_2,\cdots$); however, since tight bounds are not known for it, we cannot use them in our setting here.

As mentioned before, many results are known regarding complete problems
for the first three levels $\fastf{\Omega_1}, \fastf{\Omega_2}$ and $ \fastf{\Omega_3}$ of the hierarchy and also for classes below $\fastf{\Omega_1}$, as well as for classes between
$\fastf{\Omega_1}$ and $\fastf{\Omega_2}$. However, prior to our results, no natural complete problems
were known for $\fastf{\Omega_k}$ for any $k > 3$. Curiously, however, we do know of complete problems for 
$\fastf{\epsilon_0}$~\cite{PriorityChannelSystems,DeckerT16,Balasubramanian22}. But a finer understanding of the classes in between $\fastf{\Omega_3}$ and $\fastf{\epsilon_0}$
was missing prior to our work.
As we demonstrate, this is not due to a lack of interesting problems - any fixed problem in the hierarchy of problems mentioned above in parameterized verification, logic and other fields fall between $\fastf{\Omega_3}$ and $\fastf{\epsilon_0}$.

\subparagraph{Our techniques.} As mentioned above, there are no known natural complete
problems for $\fastf{\Omega_k}$ for any $k > 3$. Hence, our lower bounds are from first principles, i.e., proven by reductions from bounded counter machines. To this end, we show that $k$-NRCS can be used to perform \emph{weak computations} of Hardy functions and their inverses - this allows us to build a very large but ultimately bounded counter
which we can use to simulated bounded counter machines. 

Similar techniques have been used for proving lower bounds for other problems in the literature; for instance, the ideas closest to ours are from~\cite{DeckerT16} for proving hardness of $k$-NCS (note the absence of resets). However, the hardness result was only for 
the class $\fastf{\Omega_{k-1}}$. Furthermore, their construction seems to contain a subtle flaw: To perform (weak) computations of Hardy functions, the underlying machines need to be able to represent and manipulate ordinals. One important operation that would be necessary to implement such computations would be to (lossily) identify the smallest term in the CNF of a represented ordinal. The construction given in~\cite{DeckerT16} seems to use the \emph{structural ordering} on ordinals in order to isolate the smallest term, i.e., when it compares two terms, it is only guaranteed to return a term that is smaller than both in the structural ordering. 
This indicates that there will be no reliable run of the Hardy computation, when two terms are incomparable under the structural ordering (for instance, when comparing $\omega^\omega$ and $\omega^2$, the construction in~\cite{DeckerT16} seems to always return some ordinal below $\omega$, whereas it must allow to return any ordinal below $\omega^2$). To overcome this, we develop novel gadgets that allow us to compare two different terms in the representation of an ordinal and (lossily) pick the smaller one among them according to the usual ordering. This requires various new ideas that allows to construct such ``comparator'' gadgets for ordinals below $\Omega_k$, in an inductive fashion. Furthermore, since a $k$-NRCS can be easily simulated by a $2k$-NCS, this also allows to show a $\fastf{\Omega_{k}}$-hardness result for coverability of $2k$-NCS, thereby correcting the proof in~\cite{DeckerT16} and 
recovering the $\fastf{\epsilon_0}$-completeness result of~\cite{DeckerT16}.

Turning to upper bounds, we use the framework of well-structured transition systems
to show that coverability for $k$-NRCS is decidable and bound its running time by length function theorems for normed well-quasi-orders (nwqos). To estimate the latter, we interpret trees of height $k$ as $k$-nested applications of the multiset operator to a finite set. Then, using the framework developed by Schmitz, Schnoebelen and others~\cite{FigueiraFSS11,SchmitzS11,Rosa-Velardo17,GreatiR25}, we develop length function theorems for such nested multiset nwqos, which allows us to prove that coverability of $k$-NRCS is in $\fastf{\Omega_k}$. While~\cite{Rosa-Velardo17} provides such a length function theorem for $2$-nested multiset nwqos, we show how to generalize the insights there to all $k$, by designing the right operators for defining $k$-nested multiset nwqos. 

Finally, we show that the upper bound techniques developed here are of wider applicability beyond just $k$-NRCS. We use it to reduce the current upper bounds for hierarchies of problems from XML processing, graph transformation systems and $\pi$-calculus. Combining these techniques with the lower bound of $k$-NRCS also allow us to prove other complete problems from the fields of parameterized verification and logic for all $\fastf{\Omega_k}$.
We believe this is evidence that the model of NRCS is well-suited to understand complexity classes in the fast-growing hierarchy. 

The rest of this paper is structured as follows: In Section~\ref{sec:prelims} we introduce the model of $k$-NRCS. In Section~\ref{sec:lower-bound}, we prove the $\fastf{\Omega_k}$ lower bound for coverability of $k$-NRCS. In Section~\ref{sec:upper-bounds}, we provide the corresponding tight upper bounds. In Section~\ref{sec:applications}, we present the applications of $k$-NRCS to other models from the literature. We then conclude in Section~\ref{sec:conclusion}.

\section{Preliminaries}\label{sec:prelims}

In this section, we introduce our model of nested reset counter systems (NRCS) and state our main result regarding them. We begin by defining NRCS.

\subsection{Nested Reset Counter Systems}

Nested Reset Counter Systems (NRCS) are a generalization of usual counter systems
to \emph{higher-order counters}, i.e., counters which can themselves contain other (lower-order) counters.
The essential idea is that an 1-order counter is a usual counter over the natural numbers, which can either add or subtract 1. A 2-order counter is a counter which can add or subtract 1-order counters and so on. As mentioned in the Introduction, a system with $k$-order counters can be readily interpreted as labelled trees of height $k$. Addition and subtraction operations on the counters then correspond to 
creating and destroying nodes in the respective tree.

Viewed this way, an NRCS is similar to the model of nested counter systems (NCS)~\cite[Section 3]{DeckerT16} (which itself is similar to the model of nested Petri nets~\cite{LomazovaS99}). However, in the NCS model, only addition and subtraction operations are allowed. An NRCS extends an NCS by also allowing (higher-order) counters to be reset. We now proceed to formally define the syntax and semantics of an NRCS.

A $k$-NRCS is a tuple $\net = (Q,\delta_{u},\delta_r)$ where $Q$ is a finite set of \emph{states}, $\delta_u \subseteq \bigcup_{1 \le i,j \le k+1} (Q^i \times Q^j)$ is a finite
set of \emph{update transitions} and $\delta_r \subseteq \bigcup_{1 \le i \le k} (Q^i \times Q \times Q^i)$ is a finite set of \emph{reset transitions}. The union of both $\delta_u$ and $\delta_r$ will be called the set of \emph{transitions} of $\net$.
For ease of notation,
we will depict an update transition $((p_0,\dots,p_i),(q_0,\dots,q_j))$
as $(p_0,\dots,p_i) \xdashrightarrow{u} (q_0,\dots,q_j)$.
Similarly, we will depict a reset transition $((p_0,\dots,p_i),(p),(q_0,\dots,q_i))$
as $(p_0,\dots,p_i) \xdashrightarrow{r, p} (q_0,\dots,q_i)$.

The set $\configs$ of configurations of $\mathcal{N}$ is defined as the set of all finite, rooted, unordered labelled trees of height at most $k$, with labels from the set $Q$. (Throughout this paper, we will refer to a node at level $i$ in a tree if it is at distance $i$ from the root).
The operational semantics of $\mathcal{N}$ is defined in terms of a step relation $\act{} \subseteq \configs \times \configs$ on its configurations. This relation $\act{}$ will be based on the transitions of $\mathcal{N}$, which we distinguish into three cases.

Suppose $t = (p_0,\dots,p_i) \xdashrightarrow{u} (q_0,\dots,q_j)$ is an update transition
such that $i \le j \le k$ (If $i = j$ then $t$ is a \emph{renaming transition} and if $i < j$, then $t$ is an \emph{incrementing transition}). We say that there is a step between a configuration $C$ and
a configuration $C'$ by means of the transition $t$ (denoted by $C \act{t} C'$), if there
is a path $v_0,\dots,v_i$ in $C$ starting at the root such that for each $0 \le \ell \le i$, $v_\ell$ is labelled by $p_\ell$ and $C'$ is obtained from $C$ by 1) for every $0 \le \ell \le i$, changing
the label of $v_\ell$ to $q_\ell$ and 2) for every $i+1 \le \ell \le j$, creating a new node
$v_\ell$ with label $q_\ell$ and attaching it as a child to $v_{\ell-1}$.

Suppose $t = (p_0,\dots,p_i) \xdashrightarrow{u} (q_0,\dots,q_j)$ is an update transition
such that $j < i \le k$ (We say that $t$ is a \emph{decrementing transition}). Then $C \act{t} C'$, if there
is a path $v_0,\dots,v_i$ in $C$ starting at the root such that for each $0 \le \ell \le i$, $v_\ell$ is labelled by $p_\ell$ and $C'$ is obtained from $C$ by 1) for every $0 \le \ell \le j$, changing
the label of $v_\ell$ to $q_\ell$ and 2) removing the subtree rooted at $v_{j+1}$.

Finally, suppose $t = (p_0,\dots,p_i) \xdashrightarrow{r, p} (q_0,\dots,q_i)$ is a reset transition with $0 \le i < k$. Then $C \act{t} C'$, if there
is a path $v_0,\dots,v_i$ in $C$ starting at the root such that for each $0 \le \ell \le i$, $v_\ell$ is labelled by $p_\ell$ and $C'$ is obtained from $C$ by 1) for every $0 \le \ell \le i$, changing
the label of $v_\ell$ to $q_\ell$ and 2) removing every subtree rooted at a child $v$ of $v_i$ such that $v$ is labelled by $p$.

We say that $C \act{} C'$ if $C \act{t} C'$ for some transition $t$. We say that $C \act{*} C'$ if there is a sequence of steps $C \act{t_1} C_1 \act{t_2} C_2 \dots C_{n-1} \act{t_n} C'$ for some $t_1,t_2,\dots,t_n$ and some $C_1,\dots,C_n$. If $C \act{*} C'$ holds, we say that $C$ can reach $C'$. Finally, we say that a $k$-NRCS $\net$ is a $k$-NCS, if its set of reset transitions is empty.

\begin{figure}
\begin{center}
  

\begin{tikzpicture}[
  scale=0.9,
  transform shape,
  level distance=1.6em,
  sibling distance=1.4em,
  every node/.style={inner sep=1.5pt, outer sep=0pt},
  >=stealth,
  x=0.25cm
]

\node (T1) at (0,0) {$q_0$}
  child { node {$q_1$} 
    child { node {$q_3$}}
  }
  child { node {$q_2$} }
  child { node {$q_1$}
    child { node {$q_2$} }
    child { node {$q_2$} }
  };

\node[right=4 of T1] (L1) {$\xrightarrow{t_1}$};

\node[right=4 of L1] (T2) {$q_1$}
  child { node {$q_1$} 
    child {node {$q_3$}}
  }
  child { node {$q_2$} }
  ;

\node[right=4 of T2] (L2) {$\xrightarrow{t_2}$};

\node[right=4 of L2] (T3) {$q_0$}
  child { node {$q_1$}
    child { node {$q_3$} }
  }
  child { node {$q_2$} }
  child { node {$q_1$} 
    child {node {$q_2$}}
  }
  ;

\node[right=4 of T3] (L3) {$\xrightarrow{t_3}$};

\node[right=4 of L3] (T4) {$q_3$}
  child { node {$q_2$} };

\end{tikzpicture}
\end{center}
\caption{A run of the machine from Example~\ref{example:NCS}}
\label{fig:ncs_examp}
\end{figure}
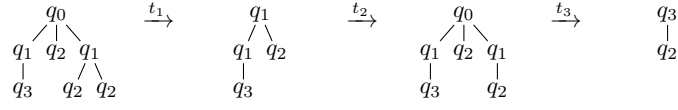

\begin{example}\label{example:NCS}
    Consider a 2-NRCS with states $\{q_0,q_1,q_2,q_3\}$ and transitions given by 
    $t_1 := (q_0,q_1) \xdashrightarrow{u} (q_1), t_2: = (q_1) \xdashrightarrow{u} (q_0,q_1,q_2)$ and $t_3 := (q_0) \xdashrightarrow{r, q_1} (q_3)$. A run
    of this machine is given in Figure~\ref{fig:ncs_examp}.
\end{example}

To state the problem of interest to us, we need to set up some notation.
We say that a configuration $C$ is smaller than another configuration $C'$,
denoted by $C \le_{is} C'$ if $C$ is an \emph{induced subgraph} of $C'$, i.e.,
there exists an injection $h$ from the nodes of $C$ to the nodes of $C'$
such that 1) $h$ maps the root of $C$ to the root of $C'$, 2) for every node $u \in C$, the labels of $u$ and $h(u)$ are the same and 3) for any two nodes $u,v \in C$,  $(u,v)$ is an edge in $C$ iff $(h(u),h(v))$ is an edge in $C'$.
It can be easily seen that $C \le_{is} C'$ iff $C$ can be obtained from $C'$
by deleting some subtrees in $C'$.

We say that a configuration $C$ can cover another configuration $D$ if $C$ can reach a configuration $C'$ that is bigger than $D$ i.e., there exists
$C'$ such that $C \act{*} C'$ and $C' \ge_{is} D$. We are now ready to state the main problem that we shall consider in this paper, namely the \emph{coverability problem for $k$-NRCS}

\begin{mdframed}
{\centering \textbf{Coverability Problem for $k$-NRCS}\\[1.0ex]}
	\noindent Given: A $k$-NRCS $\net$ and two configurations $C, C'$\\[1.0ex]
	Decide: Whether $C$ can cover $C'$
\end{mdframed}
 
The main result of this paper is to provide tight complexity results for this problem, for every $k$. To that end, we first introduce the complexity classes that we shall require, namely the \emph{fast-growing complexity classes}.

\subsection{Fast-growing Complexity Classes}

The fast-growing hierarchy of complexity classes are defined using a hierarchy of functions indexed by ordinals. To that end, we recall some basic facts about ordinals.

\subparagraph*{Ordinal Terms.} Throughout this paper, an \emph{ordinal term} (or simply ordinal) is any term that can be obtained by the following grammar 
$$  \alpha := 0 \  \ | \ \alpha + \alpha \ | \ \omega^\alpha $$
where addition, i.e., the $+$ operator, is associative with 0 as the neutral element. We often use 1 for the term $\omega^0$ and $\omega$ for $\omega^1$.
The set of ordinals can be ordered by a linear order $<$ defined as follows: For any two ordinals $\alpha, \alpha'$ we have $\alpha < \alpha'$ iff either 1) $\alpha = 0$ and $\alpha' \neq 0$ or 2) $\alpha = \omega^\beta + \gamma, \alpha' = \omega^{\beta'} + \gamma'$ and $\beta < \beta'$ or 3)  $\alpha = \omega^\beta + \gamma, \alpha' = \omega^{\beta'} + \gamma', \beta = \beta'$ and $\gamma < \gamma'$. 

An ordinal $\alpha$ is in \emph{Cantor Normal Form} (CNF) if $\alpha = \omega^{\alpha_1} + \omega^{\alpha_2} + \cdots + \omega^{\alpha_n}$ where $\alpha_1 \ge \alpha_2 \ge \cdots \ge \alpha_n$ are ordinals that are themselves in CNF. 
The terms $\omega^{\alpha_1}$ and $\omega^{\alpha_n}$ are respectively called the
biggest and smallest terms in the CNF of $\alpha$.
If $\alpha_n = 0$, then $\alpha$ is a \emph{successor ordinal} and if $\alpha_n > 0$, then it is a \emph{limit ordinal}. We will usually use $\lambda$ to denote limit ordinals throughout this paper. 
Note that, in the CNF notation, $0$ is the empty sum, and each number $c \in \mathbb{N}$ is simply $\sum_{i=1}^c \omega^0 = \sum_{i=1}^c 1$.

For $c \in \mathbb{N}$, let $\omega^{\beta} \cdot c$ denote $\sum_{i=1}^c \omega^\beta$. We sometimes write ordinals in a \emph{strict form} as $\alpha = \omega^{\beta_1} \cdot c_1 + \omega^{\beta_2} \cdot c_2 + \cdots + \omega^{\beta_m} \cdot c_m$ where $\beta_1 > \beta_2 > \cdots > \beta_m$ and the \emph{coefficients} $c_i$ are all strictly positive. We can use CNF (or its strict form) to represent any ordinal less than $\epsilon_0$ (where $\epsilon_0$ is the supremum of $\omega, \omega^\omega, \omega^{\omega^{\omega}}, \dots$). 

We end this introduction to ordinals, by defining the \emph{natural sum operator} $\oplus$ as follows:
$(\sum_{i=1}^n \omega^{\alpha_i}) \oplus (\sum_{j=1}^m \omega^{\beta_j}) = \sum_{\ell = 1}^{n+m} \omega^{\gamma_\ell}$
where $\gamma_1 \ge \gamma_2 \ge \cdots \ge \gamma_{m+n}$ is a rearrangement of $\alpha_1,\cdots \alpha_n$ and $\beta_1,\cdots,\beta_m$.

\subparagraph*{Hardy, Cicho\'n and the Fast-Growing Hierarchies.} Using ordinals, we can now describe the required hierarchy of functions and its complexity classes.
First, we introduce a definition.

A \emph{fundamental sequence for a limit ordinal} $\lambda$ is a sequence of ordinals
$(\lambda_x)_{x \in \mathbb{N}}$ whose supremum is $\lambda$. We fix this sequence as follows:
$$(\gamma + \omega^{\beta+1})_x = \gamma + \omega^{\beta} \cdot x, \qquad  (\gamma + \omega^{\lambda})_x = \gamma + \omega^{\lambda_x}$$

Using fundamental sequences, we can now define a hierarchy of functions as follows. Let $h : \mathbb{N} \to \mathbb{N}$ be a strictly increasing function. The \emph{Hardy hierarchy} for $h$ is given by a hierarchy of functions $(h^\alpha)_{\alpha < \epsilon_0}$ as 
$$h^0(x) = x, \qquad h^{\alpha+1}(x) = h^{\alpha}(h(x)), \qquad h^\lambda(x) = h^{\lambda_x}(x)$$

Similarly the Cicho\'n hierarchy for $h$ is given by 
$$h_0(x) = 0, \qquad h_{\alpha+1}(x) = 1+ h_{\alpha}(h(x)), \qquad h_\lambda(x) = h_{\lambda_x}(x)$$

Finally, the \emph{fast-growing hierarchy} for $h$ is given by
$$f_{h,0}(x) = h(x), \qquad f_{h, \alpha+1}(x) = f^{x}_{h,\alpha}(x), \qquad f_{h, \lambda}(x) = f_{h,\lambda_x}(x)$$
where $f^{x}_{h,\alpha}$ denotes the $x$-fold composition of $f_{h,\alpha}$ with itself.

We will use $H^\alpha, H_\alpha$ and $F_\alpha$ to denote the Hardy, Cicho\'n and fast-growing hierarchies that we get when $h$ is the successor function. It is known that $F_2$ grows like an exponential function, $F_3$ grows much faster than any fixed tower of exponentials
and $F_\omega$ grows roughly like the Ackermann function.

\subparagraph*{Fast-growing complexity classes} We now use the fast-growing hierarchy defined above to define our fast-growing complexity classes. Recall the functions $F_\alpha$ obtained by instantiating the fast-growing hierarchy on the successor function.
Using this, we are now going to define two hierarchies of complexity classes.

The first one is a complexity class of functions, called the extended Grzegorczyk hierarchy. More precisely, for any ordinal $\alpha > 2$, we set $\mathscr{F}_{< \alpha}$ as the collection of all functions from $\mathbb{N}$ to $\mathbb{N}$ that can be computed by a deterministic Turing machine running in time $F_\beta(n)$ for some $\beta < \alpha$. Intuitively, $\mathscr{F}_{< \alpha}$ contains
functions that grow strictly slower than $F_{\alpha}(n)$. Of special interest to us
are the functions from $\cup_{n < \omega} \ \mathscr{F}_n$, which are called the class
of \emph{primitive recursive functions}.

The second one is a complexity class of decision problems, called the \emph{fast-growing complexity classes}. More precisely,
for any ordinal $\alpha > 2$, we set $\mathbf{F}_\alpha$ as the collection
of decision problems solvable by a deterministic Turing machine running in time 
$F_\alpha(p(n))$ for some $p \in \mathscr{F}_{< \alpha}$. Intuitively, $\mathbf{F}_\alpha$
contains all decision problems that can be solved in time bounded by $F_\alpha$ composed with some `slower' function $p(n) \in \mathscr{F}_{< \alpha}$. This composition of $F_\alpha$ with a `slower' function is required
in order to get $\mathbf{F}_\alpha$-completeness through $\mathscr{F}_{< \alpha}$-reductions.

Having defined the fast-growing complexity classes, we are now in a position to state the main result of this paper. For any $k$, let $\Omega_k$ denote the ordinal obtained by a tower of exponents of $\omega$ of height $k$, i.e., $\Omega_1 = \omega$ and $\Omega_{k+1} = \omega^{\Omega_k}$. Our main result is that

\begin{theorem}
    For any $k \ge 1$, the coverability problem for $k$-NRCS is $\mathbf{F}_{\Omega_k}$-complete.
\end{theorem}

\section{Lower Bound for Coverability of $k$-NRCS}\label{sec:lower-bound}

In this section, we prove an $\fastf{\Omega_k}$ lower bound for the coverability problem
for $k$-NRCS. We do this by reducing from the coverability problem for $H^{\Omega_{k+1}}$-bounded Minsky machines. A Minsky machine is a machine with some number of counters, which it can either increment, decrement or test for zero.
Intuitively, they are just like a $1$-NRCS except instead of having resets, they have zero-tests. (See Appendix~\ref{sec:appendix-Minsky-machines} for a formal semantics). They are well-known to be an equivalent model to Turing machines. The reachability problem for $H^{\Omega_{k+1}}$-bounded Minsky machines asks to decide, if a given Minsky machine with a finite set of states $Q$ can start from a given initial configuration and cover a given final configuration using runs where its counters are bounded by $H^{\Omega_{k+1}}(|Q|)$. (Recall that $H^{\Omega_{k+1}}$ is the Hardy function obtained from the successor function at ordinal $\Omega_{k+1}$). This problem is known to be $\fastf{\Omega_k}$-hard~\cite{Schmitz16hierarchies,FischerMR68}. In order to give the reduction from this problem to coverability for $k$-NRCS, we need some notation and auxiliary results, which we shall now state.



\subsection{Encoding Ordinals by Labelled Trees}\label{subsec:ordinals-to-trees}

As a first step, recall that the Hardy function $H^{\Omega_{k+1}}$ is defined inductively
in terms of ordinals below $\Omega_{k+1}$. Intuitively, this means that, if we want to simulate counters bounded by $H^{\Omega_{k+1}}(n)$ (for some $n$) in a $k$-NRCS, we need to have the ability to 
manipulate ordinals up till $\Omega_{k+1}$ using $k$-NRCS. To this end, we provide a way of encoding ordinals less than $\Omega_{k+1}$ as labelled trees of height at most $k$ in the following way.

Fix a number $\ell$ for the rest of this section and consider all ordinals  below $(\Omega_{k+1})_\ell$. (Recall that this is the $\ell^{th}$ element in the fundamental sequence for $\Omega_{k+1}$ and is just
${\left.
\begin{array}{c}
\omega^{\cdots^{\omega^{\ell}}} \\
\end{array}
\right\}
\;k \text{ stacked } \omega's})$.
To each $\alpha \le (\Omega_{k+1})_\ell$, we first associate a forest $\text{FO}_\alpha$ of height at most $k$ as follows: If $\alpha = 0$, then $\text{FO}_\alpha$ is the empty forest; Otherwise, let $\alpha = \sum_{i=1}^n \omega^{\beta_i}$ in CNF. 
Intuitively, $\text{FO}_\alpha$ will have one tree for each $\omega^{\beta_i}$, obtained by unifying all of the trees in $\text{FO}_{\beta_i}$ by a new root.
More precisely, each tree of the forest $\text{FO}_\alpha$ is obtained by taking each $\text{FO}_{\beta_i}$ and connecting all of the trees in $\text{FO}_{\beta_i}$ by a new root labelled by $\omega$. 

Note that $\text{FO}_\alpha$ only gives us a forest of height at most $k$ and not a tree of height at most $k$. To remedy this, for every ordinal $\alpha \le (\Omega_{k+1})_\ell$, we can define a tree $T'_\alpha$ which is obtained by taking all the trees in the forest $\text{FO}_\alpha$ and connecting them with a new root labelled by $\omega$. (In this notation, $T'_0$ is then the empty tree). 

However, there is still a problem, namely that the tree $T'_\alpha$ can be of height $k+1$. Hence, in case $T'_\alpha$ has height $k+1$, we need to cut down one of its levels. To this end, note that any node at level $k$ in $T'_\alpha$ can have at most $\ell$ children; indeed, if it has more than $\ell$ children, then $\alpha > (\Omega_{k+1})_\ell$, which would be a contradiction. Therefore, we can compress a level in $T'_\alpha$, by taking any node at level $k$ in $T'_\alpha$ with $0 \le j \le \ell$ children and replacing this node and its children
with a single node labelled by $\omega^j$. In this way we get a new tree $T_\alpha$ of height at most $k$. (Note that if there are no nodes at level $k$ in $T'_\alpha$, then
$T_\alpha$ is simply $T'_\alpha$). We stress that the exact labels of the nodes in all but the $k^{th}$ level of $T_\alpha$ are irrelevant - they are all $\omega$. Hence, even if we rename all the labels
except the ones from the $k^{th}$ level, we will still have a representation of the ordinal $\alpha$, i.e., given any tree constructed by renaming $T_\alpha$ in this way, we should still be able to uniquely identify the underlying ordinal $\alpha$. This holds even if we rename the labels at the $k^{th}$ level, as long as the renaming also preserves the original label, i.e., the renaming is done as a pair with both a new label as well as the old label in $\omega^0,\dots,\omega^\ell$.

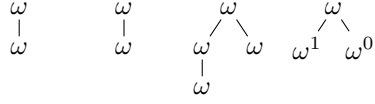
\begin{figure}
    \centering
    \begin{tikzpicture}[
  level distance=1.5em,
  sibling distance=2em,
  every node/.style={inner sep=2pt}
]

\node(T'1) {$\omega$}
  child { node {$\omega$} };

\node[right = of T'1](T1) {$\omega$}
  child { node {$\omega$} };

\node[right = of T1] (T'2) {$\omega$}
  child { node {$\omega$}
    child { node {$\omega$}}
    }
  child { node {$\omega$} };

\node[right = of T'2] (T2) {$\omega$}
  child { node {$\omega^1$}
    }
  child { node {$\omega^0$} };

\end{tikzpicture}
    \caption{Trees $T'_1, T_1, T'_{\omega+1}$ and $T_{\omega+1}$, given in that order for $k = 1$ and $\ell = 2$}    
    \label{fig:ordinal-one}
\end{figure}

\begin{example}
    Fix $k = 1$ and $\ell = 2$. Let us consider the ordinals $1 = \omega^0$
    and $\omega + 1 = \omega^{\omega^0} + \omega^0$. Its trees $T'_1, T_1, T'_{\omega+1}$
    and $T_{\omega+1}$ are  given in Figure~\ref{fig:ordinal-one}.
\end{example}

Having defined this encoding of ordinals into trees, 
we now show that we can (weakly) compute the value $H^{\Omega_{k+1}}(n)$
using $k$-NRCS for any number $n$. To this end, 
we first view 
Hardy computations as a set of rewrite rules over pairs $(\alpha,n)$ with $\alpha \le (\Omega_{k+1})_\ell$ and $n \in \mathbb{N}$:
$$(\alpha+1 ,n) \rightarrow_H (\alpha,n+1) \qquad (\lambda,n) \rightarrow_H (\lambda_n,n)$$
If we let $\xrightarrow{*}_H$ denote the reflexive and transitive closure of this relation, then, by the definition of the Hardy functions, we have that $(\alpha,n) \xrightarrow{*}_H (\alpha',n')$ implies $H^{\alpha}(n) = H^{\alpha'}(n')$. 

We will now construct $k$-NRCS that can mimic these rewrite rules.
To make this more precise, first we define a tree $C_{\alpha,n}$
for any ordinal $\alpha \le (\Omega_{k+1})_\ell$ and any number $n$ as follows: 
We take the tree $T_\alpha$ and to its root, we add $n$ children, each labelled by a symbol $\#$. The core results behind our lower bound are the following two theorems.
Intuitively, they say that it is possible for $k$-NRCS to compute the Hardy function
$H^{\alpha}(n)$ for any $\alpha \le (\Omega_{k+1})_\ell$  and any $n$.
We first describe the first main theorem.

\begin{restatable}{theorem}{Hardyforward}\label{thm:Hardy-forward}
    There exists a $k$-NRCS $\mathcal{N}^{fwd}$ such that
    for any $\alpha \le (\Omega_{k+1})_\ell$ 
     \begin{itemize}
        \item If $(\alpha,n) \act{*}_H (\alpha',n')$, then 
        there exists a run from $C_{\alpha,n}$ to $C_{\alpha',n'}$ in $\mathcal{N}^{fwd}$.
        \item Suppose there exists a run in $\mathcal{N}^{fwd}$ from $C_{\alpha,n}$ that ends in a configuration $C'$ with the root labeled by $\omega$. Then $C' = C_{\alpha',n'}$ for some $\alpha' \le (\Omega_{k+1})_\ell,n' \in \mathbb{N}$ and there exists $\alpha'',n''$ such that $T_{\alpha'} \le_{is} T_{\alpha''}$ and $H^{\alpha'}(n') \le H^{\alpha''}(n'') = H^{\alpha}(n)$.
    \end{itemize}
\end{restatable}

Note that this theorem guarantees that we can compute the Hardy function \emph{in a weak sense}. 
More precisely, the first condition guarantees there will be a run
which will start from $C_{\alpha,n}$ and will reach $C_{0,n'}$ where $H^{\alpha}(n) = H^0(n') = n'$. However, the second condition is only a weak converse to the first one, in the sense that it only guarantees that other runs of the machine will never exceed $H^{\alpha}(n)$ (but can go lower than it). 

Similar to the above theorem, we also show the following theorem, which shows that we can compute the Hardy function in a \emph{backward manner}, i.e., going from $C_{0,n'}$ to $C_{\alpha,n}$.

\begin{restatable}{theorem}{Hardybackward}\label{thm:Hardy-backward}
    There exists a $k$-NRCS $\mathcal{N}^{bwd}$ such that for any $\alpha \le (\Omega_{k+1})_\ell$ 
     \begin{itemize}
        \item If $(\alpha',n') \act{*}_H (\alpha,n)$ with $\alpha' \le (\Omega_{k+1})_\ell$, then there exists a run from $C_{\alpha,n}$ to $C_{\alpha',n'}$ in $\mathcal{N}^{bwd}$.
        \item Suppose there exists a run in $\mathcal{N}^{bwd}$ from $C_{\alpha,n}$ that ends in a configuration $C'$ with the root labeled by $\omega$. Then, $C' = C_{\alpha',n'}$ for some 
        $\alpha' \le (\Omega_{k+1})_\ell, n' \in \mathbb{N}$ and there exists $\alpha'',n''$ such that $T_{\alpha'} \le_{is} T_{\alpha''}$ and $H^{\alpha'}(n') \le H^{\alpha''}(n'') = H^{\alpha}(n)$.
    \end{itemize}
\end{restatable}

If we can prove these two theorems, then we can generalize the ideas given in~\cite{Schnoebelen10,PriorityChannelSystems} to simulate $H^{\Omega_{k+1}}$-bounded Minsky machines using $k$-NRCS. Indeed, given any such machine $\mathcal{M}$ with states $Q$,
we set $\alpha = (\Omega_{k+1})_{|Q|}$ 
and weakly compute the Hardy function $H^{\Omega_{k+1}}(|Q|) = H^{\alpha}(|Q|)$ using $\net^{fwd}$ to get some value $n$.
Then, we start simulating $\mathcal{M}$ by 1) replacing its zero-tests with resets and 2) balancing any increments/decrements in the counters of $\mathcal{M}$ with
decrements/increments in the counter containing the value $n$ (hence this counter acts as a budget for the other counters). Note that replacing zero-tests with resets can cause ``lossy'' errors, i.e., might make the simulation take a transition even when the counter value is not zero, which in turn, might make the total value of all the counters drop. Then when the simulation of $\mathcal{M}$ reaches a final state, 
we start moving all the values from the counters of $\mathcal{M}$ to the budget counter
and start running $\net^{bwd}$. By the Theorem~\ref{thm:Hardy-backward}, if we can cover $C_{\alpha,|Q|}$, then
we are guaranteed that the value in the budget counter was never lost and so all the 
resets behaved like zero-tests in the simulation of $\mathcal{M}$. 
This will then prove the required lower bound for $k$-NRCS.
The formal construction can be found in Appendix~\ref{sec:appendix-Minsky-machines}.

Hence, to prove the lower bound, it suffices to prove Theorems~\ref{thm:Hardy-forward}
and~\ref{thm:Hardy-backward}. We will present here the main ideas behind Theorem~\ref{thm:Hardy-forward}; the ideas for Theorem~\ref{thm:Hardy-backward} are similar and dual to it. 

\paragraph*{Proof structure of Theorem~\ref{thm:Hardy-forward}} Before we present the main ideas behind Theorem~\ref{thm:Hardy-forward}, we first give a global picture of how the construction shall proceed. First in subsection~\ref{subsec:compute-Hardy}, we show that constructing $\net^{fwd}$ can be reduced to constructing two types of gadgets using $k$-NRCS, that
we call the $k$-smallest-child gadget and $k$-copy gadget respectively. Intuitively, the $k$-smallest-child gadget, when applied on a tree that looks like $T_\alpha$, identifies the subtree that corresponds
to the smallest term in the CNF of $\alpha$. Similarly, a $k$-copy gadget can be used to copy a subtree of some given tree. Using these two types of gadgets, in subsection~\ref{subsec:compute-Hardy} we show that $\net^{fwd}$ can be constructed.

Then, in subsection~\ref{subsec:smallest-child}, we show how to construct these gadgets. To begin with,
we show how $k$-copy gadgets can be constructed for any $k \ge 1$, by a DFS style exploration of the tree.
To construct the $k$-smallest-child gadget, we introduce two other types of gadgets, that we call
$k$-comparator gadget and $k$-biggest-child gadget respectively.  Just like the $k$-smallest-child gadget,
both these gadgets work on trees that look like $T_\alpha$ for some ordinal $\alpha$.
The $k$-biggest-child gadget is just like the $k$-smallest-child gadget, except it marks the subtree that corresponds to the biggest term in the CNF. The $k$-comparator gadget is a weaker type of gadget compared to the $k$-smallest-child and $k$-biggest-child gadgets and can only be used to compare two different subtrees and mark which one the two terms corresponding to these two subtrees is smaller.

We then show how to construct all these three types of gadgets, by induction on $k$. For the base case,
we first show that a 1-comparator gadget can be easily constructed. Then, assuming a $k$-comparator gadget exists (for some $k$), we show how to construct a $k$-smallest-child gadget and a $k$-biggest-child gadget.
Finally, assuming that $k$-comparator, $k$-smallest-child and $k$-biggest-child gadgets exist (for some $k$), we construct a $(k+1)$-comparator gadget. Chaining all of these results, proves that
all three types of gadgets exist for any $k$. See Figure~\ref{fig:gadgets} for a pictorial representation of this inductive construction.

With this global picture in mind, we now move to subsection~\ref{subsec:compute-Hardy}, where we show how to prove construct $\net^{fwd}$, assuming $k$-smallest-child and $k$-copy gadgets exist.

\begin{figure}
    \centering
    \begin{tikzpicture}[
    node distance = 0.5cm,
  comp/.style={
    draw, rounded corners=4pt, fill=violet!15, draw=violet!60,
    minimum width=2cm, minimum height=0.9cm, align=center,
    font=\small\bfseries
  },
  sml/.style={
    draw, rounded corners=4pt, fill=orange!15, draw=orange!60,
    minimum width=2.4cm, minimum height=0.9cm, align=center,
    font=\small\bfseries
  },
  big/.style={
    draw, rounded corners=4pt, fill=yellow!20, draw=yellow!60!black,
    minimum width=2.4cm, minimum height=0.9cm, align=center,
    font=\small\bfseries
  },
  arr/.style={-{Stealth[length=5pt]}, thick, gray!70},
]
 
 
\node[comp] (c1) {1-comp.\\gadget};
 
\node[sml, above right=of c1] (s1) {1-smallest\\-child gadget};
\node[big, below right=of c1] (b1) {1-biggest\\-child gadget};
 
\draw[arr] (c1.north east) -- (s1.west);
\draw[arr] (c1.south east) -- (b1.west);
 
\node[comp, right=3.1cm of c1] (c2) {2-comp.\\gadget};
 
\draw[arr] (s1.east) -- (c2.north west);
\draw[arr] (b1.east) -- (c2.south west);
 
\draw[arr] (c1.south) to[out=-90, in=-90, looseness=0.95] (c2.south);
 
\node[sml, above right=of c2] (s2) {2-smallest\\-child gadget};
\node[big, below right=of c2] (b2) {2-biggest\\ -child gadget};
 
\draw[arr] (c2.north east) -- (s2.west);
\draw[arr] (c2.south east) -- (b2.west);
 
\node[comp, right=3.1cm of c2] (c3) {3-comp.\\gadget};
 
\draw[arr] (s2.east) -- (c3.north west);
\draw[arr] (b2.east) -- (c3.south west);
 
\draw[arr] (c2.south) to[out=-90, in=-90, looseness=0.99] (c3.south);
 
\node[right= 0.25cm of c3, yshift=0pt] {$\cdots$};
 
\end{tikzpicture}
    \caption{Inductive construction of the $k$-comparator, $k$-smallest-child and $k$-biggest-child gadgets. Arrows from a set of gadgets $S$ to a gadget $G$ means that $G$ can be constructed assuming all the gadgets in $S$ exist.}
    \label{fig:gadgets}
\end{figure}
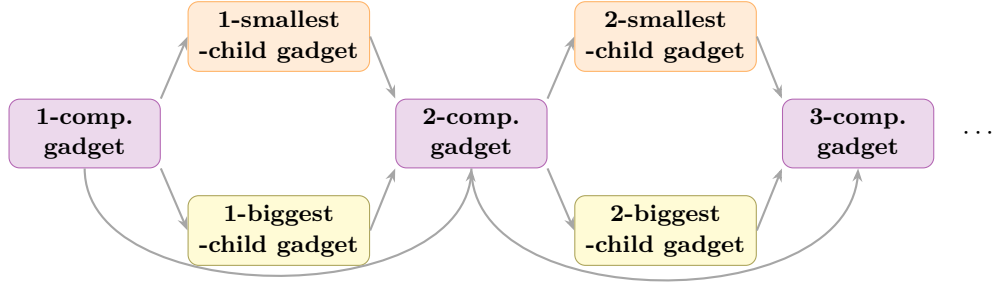

\subsection{Computing Hardy Functions}\label{subsec:compute-Hardy}

Before we define our smallest-child and copy gadgets, we set up some intuition as to why they are needed for constructing $\net^{fwd}$.
Recall that $(\alpha+1,n) \rightarrow_H (\alpha,n+1)$ and $(\lambda,n) \rightarrow_H (\lambda_n,n)$. Hence, in order to perform the Hardy computations using $k$-NRCS, we need a way
to move from $C_{\alpha+1,n}$ to $C_{\alpha,n+1}$ and from $C_{\lambda,n}$ to $C_{\lambda_n,n}$.
The former is quite easy - we need to remove one $\omega$ and in its place, add a $\#$. 
However, the latter is quite tricky - We have to compute $\lambda_n$ from $\lambda$, i.e., the $n^{th}$ element in the fundamental sequence of $\lambda$. To do this, we first need to isolate the smallest term in the CNF of $\lambda$. This poses a challenge when working over $C_{\lambda,n}$, since the tree is unordered.
To perform this successfully, we need a gadget which allows us to isolate the subtree corresponding to the smallest term in the CNF of $\lambda$.
We now formalize and construct such a gadget.

Fix some $k$ and consider some $\alpha \le (\Omega_{k+1})_\ell$. Consider the tree $T_\alpha$. 
Note that, by construction, every child subtree of the root of $T_\alpha$
corresponds to a unique term in its CNF, whose child subtrees correspond to unique terms in their CNFs and so on. Also, recall that the exact labels of the nodes in all but the $k^{th}$ level are irrelevant, i.e., even if we change the labels of $T_\alpha$ whilst retaining the labels of the nodes in the $k^{th}$ level, the resulting tree still uniquely represents $\alpha$, the resulting children of the root still uniquely represent terms in the CNF of $\alpha$ and so on. 
We will call a tree $T$ an \emph{encoder} for $\alpha$ if $T$ is obtained from $T_\alpha$ by any such renaming.

We shall now define specific types of encoders for an ordinal $\alpha$, each of which we will use for a specific task. The first is called a \emph{small-initialized encoder} and is defined as follows.
\begin{definition}
    A small-initialized encoder for $\alpha$ is an encoder $T$
    where the root is labelled by $\bsmall$.
\end{definition}

The root being labelled by $\bsmall$ signals that we now want to find the smallest term in the CNF represented by the tree $T$. Ultimately, we want a gadget that will start from an initialized encoder
and among all ordinals represented by the children of the root marked by $\omega$, it will mark the child corresponding to the smallest ordinal. To this end,
we define a \emph{lossy smallest-child encoder} as follows.

\begin{definition}
    Given a small-initialized encoder $T$ for $\alpha$, a lossy smallest-child encoder for $T$
    is a tree $T'$ that looks exactly like $T$, except for the following changes:
    \begin{itemize}
        \item Let $v_1,\dots,v_n$ be the children of the root labeled by $\omega$ in $T$. $T'$ contains all the nodes of $T$, except possibly for some nodes in the subtree of any of $v_1,\dots,v_n$.
        \item Every node in $T'$ has the same label as its corresponding node in $T$,
        except that the root is labeled by $\esmall$, and each $v_i$ (if it exists in $T'$) is labeled as follows:
        \begin{itemize}
        \item If the ordinal represented by the subtree rooted at $v_i$ in $T'$ is the smallest among all the ordinals of any of the subtrees of all the other $v_j$ in $T'$, then $v_i$ is labeled by $\smallest$.
        \item Otherwise, $v_i$ has the same label as it had in $T$.
        \end{itemize}
    \end{itemize}
    A lossy smallest-child encoder is called a \emph{perfect smallest-child encoder} if 
    $T'$ contains all the nodes of $T$.    
\end{definition}

Ideally, we want a gadget that, starting from a small-initialized encoder $T$ for some $\alpha$,  
arrives at a perfect smallest-child encoder for it. This would, among all the children of the root marked by $\omega$, correctly pick a node that represents the smallest ordinal in the CNF of $\alpha$. However, this is too strong a requirement for $k$-NRCS. Hence, we will allow the gadget to make ``lossy'' errors, i.e., we will allow the gadget to drop nodes from the  subtrees of the children of the root marked by $\omega$. In return, we are only guaranteed that the 
node marked by $smallest$ represents the smallest ordinal \emph{in $T'$} (not in $T$!).
We now formalize what it means for such a gadget to exist.

\begin{definition}[Smallest-child gadgets]
    A $k$-smallest-child gadget is a $k$-NRCS such that for any $\alpha \le (\Omega_{k+1})_{\ell}$ and any small-initialized encoder $T$ for $\alpha$:
    \begin{itemize}
        \item There exists a run from $T$ to the perfect smallest-child encoder of $T$.
        \item Suppose there exists a run from $T$ to some $T'$ such that the root of $T'$ is labeled by $\esmall$. Then $T'$ is a lossy smallest-child encoder of $T$.
    \end{itemize}
\end{definition}

Hence, a smallest-child gadget is always guaranteed to find a lossy smallest-child encoder
of a small-initialized encoder. One of the main results that we shall show in the next subsection is that
\begin{theorem}\label{thm:smallest-child-gadgets}
    A $k$-smallest-child gadget exists.
\end{theorem}

Hence, these smallest-child gadgets allow us to identify the smallest term in the CNF
of an ordinal encoded as a tree, thereby facilitating the computation of the $n^{th}$ element in the fundamental sequence of a limit ordinal $\lambda$. However, this alone is not enough to perform the Hardy computation $(\lambda,n) \act{}_H (\lambda_n,n)$, for the following reason: Suppose $\lambda = \gamma + \omega^{\beta+1}$ for some $\gamma, \beta$. 
Then $\lambda_n = \gamma + \omega^\beta \cdot n$. Hence, after marking the smallest term $\omega^{\beta+1}$ in the encoder, we need to reduce its exponent by one (which is easy) and then \emph{copy it $n-1$ times} (which is tricky). Hence, we also need a type of \emph{copy gadget} which makes multiple copies of a subtree. For this purpose, we define
\emph{marked encoders} and \emph{lossy copy encoders} as follows.

\begin{definition}
    A \emph{marked encoder} for $\alpha$ is an encoder for $\alpha$ where the root is labeled by $\bcopy$ and exactly one of its children is labeled by $\marked$.
    
    Given a marked encoder $T$ for $\alpha$, a \emph{lossy copy encoder} for $\alpha$
    is any tree $T'$ obtained as follows: Let $v$ be the child of the root labeled by $\marked$ in $T$. $T'$ is obtained from $T$, by first deleting some (possibly zero) nodes from the subtree of $v$, then duplicating this subtree as another subtree of the root
    and then deleting some (possibly zero) nodes from the subtree of $v$
    and then finally changing the label of the root to $\ecopy$, the label of $v$ to $\omega$ and the label of the copy of $v$ as $\copied$.     
    A lossy copy encoder is called a \emph{perfect copy encoder} if no nodes
    from the subtree of $v$ were deleted.
\end{definition}

Now similar to $k$-smallest-child gadgets, we can define $k$-copy gadgets. In the next subsection we show that
\begin{theorem}\label{thm:copy-gadgets}
    A $k$-copy gadget exists.
\end{theorem}

Assuming Theorems~\ref{thm:smallest-child-gadgets} and~\ref{thm:copy-gadgets}, 
we now show how to prove Theorem~\ref{thm:Hardy-forward}, i.e., how to construct
$\net^{fwd}$. To this end, we have to devise a $k$-NRCS that allows us to 
go from $C_{\alpha+1,n}$ to $C_{\alpha,n}$ and $C_{\lambda,n}$ to $C_{\lambda_n,n}$.
We have already sketched how to do the former and so we concentrate on the latter.
To go from $C_{\lambda,n}$ to $C_{\lambda_n,n}$, we first invoke the $k$-smallest child gadget to mark the child $v$ of the root corresponding to the smallest term of the CNF of $\lambda$. We non-deterministically decide whether this term is of the form $\omega^{\beta+1}$
or $\omega^{\lambda'}$ for some limit ordinal $\lambda'$. In the former case, there must be a leaf child of $v$, which we remove (to get the ordinal to $\omega^{\beta}$) and then invoke the copy gadget repeatedly for at most $n-1$ times (where $n$ is stored by means of $n$ children of the root labelled by $\#$). We can do this by transforming a $\#$ into $\#'$ every time we invoke the copy gadget, then at some point non-deterministically resetting all the $\#$ and then transforming all the labels $\#'$ back into $\#$.
In the other case of $\omega^{\lambda'}$, we now have to compute $\lambda'_n$, for which we need the smallest ordinal in the CNF of the tree of $v$. So, we now invoke 
the $(k-1)$-smallest gadget on the node $v$ and then recursively continue. 

If we never make any lossy errors, then our computation will be perfect and we will end
up exactly moving to $C_{\lambda_n,n}$. On the other hand, if we make lossy errors, then the Hardy computation value will only go smaller, i.e., we will only move to some $C_{\lambda',n'}$ such that $H^{\lambda'}(n') \le H^{\lambda}(n)$.
The intuitive reason is that, through lossy errors, we are only dropping down to lower ordinals (not just over the usual ordering, but in a precise sense over the so-called structural ordering, which is the same as the induced subgraph ordering on the corresponding trees of the ordinals) and this will only make the Hardy computation go smaller.
The construction of the $k$-NRCS $\net^{bwd}$ is symmetric and requires other dual types of gadgets. The formal details can be found in Appendices~\ref{sec:appendix-Hardy-forward} and~\ref{sec:appendix-Hardy-backward}.

Therefore, in order to prove Theorem~\ref{thm:Hardy-forward}, it suffices to construct smallest-child and copy gadgets, which we shall do next.

\subsection{Constructing Smallest-Child and Copy Gadgets}\label{subsec:smallest-child}

We first show how to construct copy gadgets, because we will also need them for our smallest-child gadgets. The main idea is the following: In order to construct a $k$-copy gadget, we need to start from a marked encoder $T$, where some child $v$ of the root is labelled by $\marked$ and we need to (lossily) copy the subtree rooted at $v$. To begin with, we make another copy of just the node $v$ (call it $v'$).
Then, the idea is simply to perform a (lossy) depth-first search (DFS) of the subtree rooted at $v$ and copy branches of this subtree at $v$ to branches of the subtree at $v'$.
More precisely, we start at $v$ and systematically explore each of its branches in a DFS manner, marking each branch
as we go along. At any point, we non-deterministically decide if either the current branch can be extended or we have hit a leaf. In the former case, we recursively continue by selecting some child at the end of this branch. 
In the latter case, we simply copy this branch as another branch of $v'$ and 
reset all the children of the last node of this branch. (If we had guessed correctly, there would be no such children to reset; if we had guessed wrongly, then we make a lossy error).
Then we go back up one level along our current branch and continue the DFS. 
In this way, we would have (lossily) copied the subtree of $v$ as the subtree of $v'$.
When we end the DFS, we remove all the marked information from all of the branches in $v$ and label the root as $\ecopy$ and $v'$ as $\copied$.
This proves Theorem~\ref{thm:copy-gadgets}. (See Appendix~\ref{subsec:appendix-copy-gadgets} for more details).

To show how to construct a $k$-smallest-child gadget, we first introduce two other types of gadgets. The first is a $k$-biggest-child gadget, which given an encoder of some ordinal $\alpha$, (lossily) marks the child of the root representing the biggest term in the CNF of $\alpha$. The second is a $k$-comparator gadget, which given an encoder of some ordinal $\alpha$ and two marked children $v_A$ and $v_B$ of the root,
(lossily) compares the ordinals represented by the subtrees of $v_A$ and $v_B$ and marks the smaller one among the two of them. We now show how to construct
all these three types of gadgets, by induction on $k$, following the structure outlined in Figure~\ref{fig:gadgets}.

For the base case of $k = 1$, we first show the following.
\begin{theorem}\label{thm:1-comparator}
    An $1$-comparator gadget exists.
\end{theorem}

Indeed a 1-comparator gadget can be easily constructed: Any marked children in the encoder for a 1-comparator gadget are leaves whose labels must have entries of the form $\omega^j$ and $\omega^{j'}$ for $j, j' \le \ell$ and so it is quite easy to compare which of them is bigger ($j > j'$ or $j < j'$) or whether both of them are equal ($j = j'$) by means of transitions of an 1-NRCS.

As the next step, we show the following.
\begin{restatable}{theorem}{ksmallest}\label{thm:k-smallest}
    Suppose, for some $k$, a $k$-comparator gadget exists. Then a $k$-smallest-child gadget exists.
\end{restatable}

The main idea behind this theorem is the following: In order to compute the child of the root corresponding to the smallest term, we first pick some child $v_A$ of the root labelled by $\omega$. Then,
we pick another child $v_B$ of the root labelled by $\omega$ and use the $k$-comparator gadget to compare the subtrees of $v_A$ and $v_B$. Once we have marked the smaller one (say $v_A$), we label the other one ($v_B$) as $\omega'$ and pick some other child $v_C$ labelled by $\omega$  to now compare with $v_A$.
We continue doing this until we non-deterministically guess that we have gone over all of the children. Once we make that guess, we mark the current smallest candidate as $\smallest$, reset all children of the root labelled by $\omega$, convert all the children of the root labelled by $\omega'$ back into $\omega$ and then mark the root as $\esmall$. This completes the description of the main idea behind Theorem~\ref{thm:k-smallest}.


As the next step, we can now show the following.
\begin{restatable}{theorem}{kbiggest}\label{thm:k-biggest}
    Suppose, for some $k$, a $k$-comparator gadget exists. Then a $k$-biggest-child gadget exists.
\end{restatable}

While the idea behind this theorem is the same as the one for Theorem~\ref{thm:k-smallest}, there is one non-trivial step. To explain this step, we first outline an incorrect idea and then show how to correct it: To compute the child of the root corresponding to the biggest term, we pick two children $v_A, v_B$ of the root labelled by $\omega$, compare them using the $k$-comparator gadget and then mark the bigger one (say $v_A$). Then, we pick another child $v_C$ of the root and compare it with $v_A$ and so on. In this way, when we non-deterministically guess that we have gone over all the children, we stop and declare our current candidate child as the biggest child. 

The problem with this approach is the following: After comparing $v_A$ with $v_B$,
suppose $v_A$ is declared the bigger child. After that, suppose $v_A$ and $v_C$ are compared and $v_A$ is once again declared the bigger child. Because this computation is lossy and some subtrees of $v_A$ might be lost, it need not be the case anymore that $v_A$ is still bigger than $v_B$. To solve this problem, we do the following: At the start of the computation, we pick some child $v_A$ of the root as a \emph{candidate biggest-child} and \emph{copy} its subtree. Then, we keep the copied version and compare the original version with each of the other children. If all the comparisons declare that the original version of $v_A$ is the bigger child,
then since the copied version is at least as big as the original version, we are guaranteed that the copied version of $v_A$ is the biggest-child. In this way, we can (lossily) compute the biggest-child. This completes the description of the main idea behind Theorem~\ref{thm:k-biggest}.

As a final step for the induction, we can then show the following theorem.

\begin{restatable}{theorem}{kplusonecomparator}\label{thm:k+1-comparator}
    Suppose $k$-comparator, $k$-smallest-child and $k$-biggest-child gadgets exist.
    Then a $(k+1)$-comparator gadget exists.
\end{restatable}

Similar to Theorem~\ref{thm:k-biggest}, we will first sketch an incorrect idea for this theorem and then show how to correct it. Suppose we are given two nodes $v_A$ and $v_B$ of the root and suppose their corresponding ordinals are $\alpha = \sum_{i=1}^n \omega^{\alpha_i}$ and $\beta = \sum_{j=1}^m \omega^{\beta_j}$. Note that $\alpha < \beta$ iff either $\alpha_1 < \beta_1$ or $\alpha_1 = \beta_1$ and $\sum_{i=2}^n \alpha_i < \sum_{j=2}^m \beta_j$. Hence to compare them, we first use the $k$-biggest-child gadget to mark the biggest children of $v_A$ and $v_B$ respectively and then use a $k$-comparator gadget on them. If one of them turns out to be  smaller than the other,
then we can immediately deduce which one among $v_A$ or $v_B$ is smaller.
On the other hand, suppose both the biggest children of $v_A$ and $v_B$ are equal.
In that case, we mark them as $\omega'$ and we move on to finding the next biggest children of $v_A$ and $v_B$ and repeat the process again.

However, there is one problem with this approach: Because of lossy computation, upon applying the $k$-comparator gadgets on the biggest children of $v_A$ and $v_B$,
they might cease to remain as the biggest children of $v_A$ and $v_B$. To fix this problem, we do the following: Before applying the $k$-comparator gadget on these two biggest children, we first copy them and then apply the $k$-comparator gadget on the copies, whilst
also performing the exact same operations of the $k$-comparator gadget on the original versions. This ensures that the original versions always stay smaller or equal to the copied versions. Then at the end of executing the $k$-comparator gadget, we check that the original versions are still the biggest children of $v_A$ and $v_B$ and if so, remove them (if not, we deadlock). Since the copied versions are at least as big as the original versions, if this check passes, the copied versions will still remain as the biggest children of $v_A$ and $v_B$. This concludes the description of the main ideas behind Theorem~\ref{thm:k+1-comparator}, thereby proving Theorem~\ref{thm:smallest-child-gadgets} and hence also
Theorem~\ref{thm:Hardy-forward}.
Full details of all the constructions of the gadgets can be found in Appendix~\ref{sec:appendix-gadgets}.

\section{Upper Bound for Coverability of $k$-NRCS}\label{sec:upper-bounds}

In this section, we prove our upper bound for the coverability problem for $k$-NRCS. Namely we shall show that for any $k \ge 1$, coverability for $k$-NRCS is in $\fastf{\Omega_k}$. We first present a global proof structure of how we prove this result. First, in subsections~\ref{subsec:wqo} and~\ref{subsec:WSTS} we show that $k$-NRCS can be viewed as a well-structured transition system (WSTS)~\cite{AbdullaCJT96,FinkelS01} by means of the induced subgraph ordering on trees of height at most $k$, which is known to be a (normed) well-quasi-ordering (nwqo). This allows us to reduce the problem of proving upper bounds for coverability of $k$-NRCS to the so-called length function theorems for controlled bad sequences over the induced subgraph nwqo. Then, in subsection~\ref{subsec:trees-nested-multisets}, we show that this nwqo can be interpreted in terms of another nwqo that we call \emph{nested multisets}. In the next three subsections, we then show how to prove length function theorems for controlled bad sequences over nested multisets. We do this, by leveraging the framework of Schmitz and Schnoebelen~\cite{arxiv-SchmitzS}: First, we compute the length function theorems by unraveling the descent equation; next, we approximate each step of the unraveling by reflections; finally, we use the machinery of order type and derivatives to bound the estimate of each reflection.

Having stated our proof strategy for the upper bound, we now proceed to subsection~\ref{subsec:wqo}, where we recall the notion of \emph{well-quasi-orders}.

\subsection{WQOs and Controlled Bad Sequences}\label{subsec:wqo}

\subparagraph*{Well-Quasi-Orders.} A \emph{quasi-order} over a set $X$ is a relation $\le_X$ which is reflexive and transitive. We write $x <_X y$ to mean that $x \le_X y$ and $y \nleq_X x$. A \emph{well-quasi-order} (wqo) over a set $X$ is a quasi-order $\le_X$ such that in any infinite sequence $x_0, x_1, x_2, \dots, $ of elements from $X$, there exists $i < j$ such that $x_i \le_X x_j$. Finally, a \emph{norm function} over $X$ is a function $\norm_X: X \to \mathbb{N}$ such that for any $n \in \mathbb{N}$, the set $X_{\le n} := \{x \in X : |x|_X \le n\}$ is finite. 

A normed well-quasi-order (nwqo)
is a tuple $(X,\le_X,\norm_X)$ such that $(X,\le_X)$ is a wqo and $\norm_X$ is a norm.
The set $X$ will be called the domain of the nwqo.
If $(X,\le_X,\norm_X)$ is a nwqo and $S \subseteq X$, then the nwqo induced by $S$ is the nwqo $(S,\le_S,\norm_S)$ where $\le_S$ and $\norm_S$ are the restrictions of $\le_X$ and $\norm_X$ to $S$ respectively.

Whenever the underlying order $\le_X$ and norm $\norm_X$ are clear from the context, we will drop those and simply refer to the nwqo by the set $X$. 
Given any two nwqos $X, Y$ we say that $X \equiv Y$ when $X$ and $Y$ are \emph{isomorphic structures}, i.e., there must be an isomorphism $f$ from $X$ to $Y$ such that $x \le_X x'$ iff $ f(x) \le_X f(x')$ and $|x|_X = |f(x)|_Y$ for any $x,x' \in X$.


\begin{example}
    Let $\Gamma_i$ denote the nwqo over any set $\{x_1,\dots,x_i\}$ such that any two distinct elements are unordered and $|x_j| = 0$ for any $j$. Clearly, $\Gamma_i$ is a nwqo. (Note that $\Gamma_0$ has empty domain and is hence called the empty nwqo).
    Similarly, the set $\mathbb{N}$ with the usual ordering and the identity function as the norm is also a nwqo. 
\end{example}

\subparagraph*{Good, bad and controlled sequences.} 

Let $X$ be a nwqo. A sequence $x_0, x_1, \dots$ over $X$ is called \emph{good} if there are $i < j$
such that $x_i \le_X x_j$. A sequence that is not good is called \emph{bad}. By definition, every bad sequence in an nwqo is finite.

For our purposes, a control function is any function $g: \mathbb{N} \to \mathbb{N}$ that is strictly increasing, primitive recursive, inflationary ($g(n) \ge n$) and superadditive ($g(x+y) \ge g(x) + g(y)$). For any $n \in \mathbb{N}$, a sequence $x_0, x_1, \dots$ is called $(g,n)$-\emph{controlled}
if for all $i$, $|x_i|_X \le g^i(n) = \underbrace{g(g(\cdots(g(n)\cdots))}_{i \text{ times }}$.
It is known that for every nwqo $X$, every control function $g$ and every $n$, there exists a finite maximum length $L$ for $(g,n)$-controlled bad sequences of $X$ ~\cite[Proposition 2.5]{SchmitzS11}. Hence,
we can define a \emph{length function} $L_{X,g}(n)$ which for every $X,g,n$ maps it to the maximum length of $(g,n)$-controlled bad sequences over $X$. 

\subparagraph*{Constructing complex nwqos.}

The final item that we shall need regarding nwqos are two constructions that allow us to construct ``complex'' nwqos from simpler ones. The first construction that we shall use is the disjoint sums construction.

\begin{definition}[Disjoint sums]
    Let $X_1$ and $X_2$ be two nwqos. Their disjoint sum is the nwqo $X_1+X_2$ whose domain is the set $\{(i,x) : x \in X_i\}$, whose ordering is given by 
    $(i,x) \le_{X_1+X_2} (j,x')$ iff $i = j$ and $x \le_{X_i} x'$
    and whose norm is given by $|x|_{X_1+X_2} = |x|_{X_i}$ if $x \in X_i$.
    It is easy to see that $X_1+X_2$ is a nwqo when both $X_1$ and $X_2$ are.
    Indeed any infinite bad sequence over $X_1+X_2$ must contain an infinite bad sequence from either $X_1$ or $X_2$ and the latter does not exist by assumption.
\end{definition}

Given nwqos $X_1,X_2,\dots,X_n$, we will use $\sum_{i=1}^n X_i$ to denote the disjoint sum obtained by $(((X_1 + X_2) + X_3) + \dots + X_n)$.  Similarly, given an nwqo $X$ and $n \in \mathbb{N}$, we will use $X \cdot n$ to denote $\sum_{i=1}^n X$.
Without loss of generality, $X \cdot n$ is the same as taking $n$ copies of $X$ (i.e., domain is $\{(i,x) : 1 \le i \le n, x \in X\}$) and comparison can only happen between elements belonging to the same copy (i.e., $(i,x) \le (j,x')$ iff $i = j$ and $x \le x'$).

The second construction that we shall use is the finite multisets construction.
To this end, we first define a finite multiset over $X$ as a function 
$m : X \to \mathbb{N}$ with finite support $sup(m) = \{x : m(x) \neq 0\}$.  We will typically denote finite multisets using a set-like notation; for instance the finite multiset $m(a) = 2, m(b) = 3, m(c) = 1$ and $m(x) = 0$ for all $x \notin \{a,b,c\}$ is denoted by $\multiset{a,a,b,b,b,c}$ or $\multiset{2 \cdot a, 3 \cdot b, 1 \cdot c}$.
Given two finite multisets $m, m'$, we use $m+m'$ to denote the finite multiset
given by $(m+m')(x) = m(x) + m'(x)$.

\begin{definition}[Finite multisets]
    Let $X$ be a nwqo. The finite multiset order over $X$ is the nwqo $\multi{X}$ defined as follows: Its domain is the collection of all finite multisets over $X$.
    Its ordering is given by $\multiset{x_1,\dots,x_n} \le_{\multi{X}} \multiset{x'_1,\dots,x'_{n'}}$ iff there is an injective function $h: \{1,\dots,n\} \to \{1,\dots,n'\}$ such that $x_i \le_X x'_{h(i)}$ for all $i$. Finally, the norm of an element $m = \multiset{x_1,\dots,x_n}$ is given as $|m|_{\multi{X}} = \sum_{1 \le i \le n} \max(|x_i|_X, 1)$ (An easier alternative would have been to just define $|m|_{\multi{X}}$ as $\sum_{1 \le i \le n} |x_i|_X$, but the problem with this definition is that if $|x|_X = 0$ for some $x$, then there will be arbitrarily many multisets $m$ such that $|m|_{\multi{X}} = 0$, which violates the norm property).
    The fact that the finite multiset order is an nwqo follows from Higman's lemma~\cite{higman1952ordering}. 
\end{definition}

Let us see some examples of these constructions.

\begin{example}
    By definition, it is easy to see that $\Gamma_1 \cdot n = \sum_{i=1}^n \Gamma_1 \equiv \Gamma_n$, since both $\Gamma_1 \cdot n$ and $\Gamma_n$ have $n$ unordered elements all of whose norms are 0. Furthermore, by definition, it is easy to see that $\Gamma_1 \equiv \multi{\Gamma_0}$, since the only multiset of the empty nwqo is the empty multiset.
    Finally, note that any multiset over $\Gamma_1$ is simply a number over $\mathbb{N}$.
    Hence, $\multi{\Gamma_1} \equiv \mathbb{N}$.
\end{example}

Equipped with these facts regarding nwqos, we now prove our upper bounds for $k$-NRCS, by leveraging the framework of \emph{well-structured transitions systems}~~\cite{FinkelS01,AbdullaCJT96}. We now explain this framework in detail.

\subsection{$k$-NRCS as Well-Structured Transition Systems}\label{subsec:WSTS}

Let $\mathcal{N} = (Q,\delta_u,\delta_r)$ be a $k$-NRCS. Recall that any configuration $C$ of $\mathcal{N}$ is a tree of height at most $k$ whose nodes are labelled by $Q$. 
Further, recall the following ordering among configurations of $\mathcal{N}$: $C \le_{is} D$ if $C$ is an induced subgraph of $D$.
It is known that $\le_{is}$ is a wqo among all labelled trees of height at most $k$~\cite[Section 8]{PriorityChannelSystems}. Additionally, if we define the norm of a configuration $|C|$ to be the number of nodes in $C$, then we get an nwqo among the set of configurations of $\mathcal{N}$. 

With this nwqo, it is easily seen that $\mathcal{N}$ satisfies the following \emph{compatibility property}. Roughly speaking, it means that any transition that can be fired from a smaller configuration $C$ can also be fired from a bigger configuration $D$.
(See Appendix~\ref{subsec:appendix-WSTS} for more details).

\begin{restatable}[Compatibility Property]{proposition}{compatibility}\label{prop:compatibility}
    Suppose $C \act{t} C'$ for some transition $t$ and suppose $C \le_{is} D$.
    Then, there exists $D'$ such that $D \act{t} D'$ and $C' \le_{is} D'$.
\end{restatable}

Because of this property, by definition, $(\mathcal{N},\le_{is})$ is \emph{well-structured}~\cite{AbdullaCJT96,FinkelS01}. We now use the theory behind well-structured systems to derive an algorithm for coverability of $\mathcal{N}$. For this, we first need some auxiliary results.

Given a configuration $C$, let $\uparrow C$ denote the set of all configurations above $C$, i.e., set of all $D$ such that $C \le_{is} D$. Let $pre(\uparrow C)$ be the set of all predecessor configurations of $\uparrow C$, i.e., 
the set of all configurations $C'$ such that $C' \act{t} D$ for some $t$ and some $D \in \uparrow C$. The following proposition states that given $C$, we can effectively compute a representation of $pre(\uparrow C)$. It follows by a simple analysis of the  transitions of $\mathcal{N}$. (See Appendix~\ref{subsec:appendix-WSTS} for more details).

\begin{restatable}[Effective Predecessor Basis]{proposition}{effectivepred}
    Given a configuration $C$, in primitive recursive time, we can compute a set of configurations 
    $C_1,\dots,C_m$ such that $\cup_{i=1}^m \uparrow C_i = pre(\uparrow C)$ and $|C_i| \le |C|+k+1$
    for every $i$. (The collection $\{C_1,\dots,C_m\}$ is called a basis of $pre(\uparrow C)$).
\end{restatable}

Because of this property and the fact that $\mathcal{N}$ is well-structured, 
the coverability problem becomes decidable~\cite{FinkelS01}. Indeed, suppose we want to check if a configuration $C_f$ of $\mathcal{N}$ can be covered from a configuration $C_{init}$. 
This is the same as checking if from $C_{init}$
we can reach a configuration in $\uparrow C_f$. 
To that end, using the effective predecessor basis property, we compute a sequence of sets
$U_0 := \uparrow C_f$ and $U_{i+1} := U_i \cup \cup_{C \in U_i} \ pre(\uparrow C)$. 
By~\cite[Lemma 2.4]{FinkelS01}, this computation must reach a fix-point, i.e.,
there must be $m$ such that $U_m = U_{m+1}$. At this point, we simply need to check that $C_{init} \in U_m$ to decide if $C_{init}$ can cover $C_f$. 

By standard arguments for WSTS, the running time of the above algorithm is dominated primarily by the number of steps $m$ needed to reach a fixpoint. This in turn, can be bounded by the length of $(g,|C_{f}|)$-controlled bad sequences over $\le_{is}$ for some primitive recursive function $g$. (See~\cite{SchmitzS13} or~\cite[Section 3.3]{Ramachandrakumar24} for a detailed argument). 
Hence, if we bound the length function for labelled trees of height $k$ with labels from the set $Q$, then we can give an upper bound for
coverability of $k$-NRCS.  

\subsection{Interpreting Trees as Nested Multisets}\label{subsec:trees-nested-multisets}

Hence, our goal is to now bound the length function for labelled trees of height $k$
under the $\le_{is}$ nwqo. For this purpose, it turns out that it is easier to interpret
trees of height $k$ as another nwqo called \emph{nested multisets} and work with
them instead. Accordingly, we begin by introducing nested multisets nwqos.

A \emph{nested multiset nwqo} is any nwqo that can be obtained by the following grammar
$$  A := \Gamma_0 \  \ | \ A + A \ | \ \multi{A} $$

Each nwqo of this grammar is equipped with the order and norm as defined in subsection~\ref{subsec:wqo} (for disjoint sums and finite multisets).
We shall now show that it is possible to view the $\le_{is}$ nwqo over labelled trees
of height $k$ as a nested multiset nwqo. 

\paragraph*{Trees of height $k$ as nested multisets}

Consider all possible trees of height at most $k$ with labels from some set $Q$ of size $n$.
Without loss of generality, we can assume that $Q = \{1,2,\cdots,n\}$. 
Let $\mathcal{M}_k$ be the nested multiset nwqo $Q \equiv \Gamma_{n} \equiv \Gamma_1 \cdot n \equiv \multi{\Gamma_0} \cdot n$. Furthermore, for any $1 \le i \le k$, let $\mathcal{M}_{i-1}$ be the nwqo $\multi{\mathcal{M}_i} \cdot n$. 

For any tree $C$ of height $i \le k$ and any number $j \le k-i$, we now assign a unique element $h_j(C) \in \mathcal{M}_j$ as follows.
 \begin{itemize}
     \item If $C$ is just one node labelled by some $q \in Q$, then $h_j(C) = (q,\emptyset)$, where $\emptyset$ denotes the empty multiset.
     \item Otherwise, let $v$ be the root of $C$ labelled by some $q \in Q$ 
     and $C_1,\dots,C_m$ be the child subtrees of $v$.
     We then set $h_j(C) = (q,\multiset{h_{j+1}(C_1), h_{j+1}(C_2), \dots, h_{j+1}(C_m)})$.
\end{itemize}

 By induction on $j$, it is easy to see that $h_j$ is a bijection from all trees of height at most $k-j$ to $\mathcal{M}_j$. Furthermore, we also have that for any two trees $C,D$ of height at most $k-j$, $C \le_{is} D$ iff $h_j(C) \le_{\mathcal{M}_j} h_j(D)$ and $|h_j(C)| \le |C|$. It then immediately follows that any $(g,n_0)$-controlled bad sequence $S$
 over trees of height $k-j$ can be converted into a $(g,n_0)$-controlled bad sequence $h_j(S)$ over $\mathcal{M}_j$, obtained by applying $h$ to each element of $S$.

 Hence, any upper bound on controlled bad sequences for the nested multiset nwqo $\mathcal{M}_0$
 is also an upper bound on controlled bad sequences for trees of height $k$ 
 with labels from $Q$. In the rest of this section, we shall show that
 \begin{theorem}\label{thm:main-upper-bound}
     For any (primitive recursive) control function $g$, we have $L_{\mathcal{M}_0,g}(x) \le h_{\Omega_{k+1}}(|Q| + |Q|xk)$ where $h(x) = x \cdot g(x)$
 \end{theorem}

In order to use this to derive upper bounds for  coverability, we first use the 
identity $h_{\Omega_{k+1}}(|Q|xk) \le f_{h,\Omega_{k}}(|Q|xk)$\cite[Lemmas C.6 and C.8]{arxiv-SchmitzS}. By~\cite[Theorem 4.2]{Schmitz16hierarchies}, it follows that
$f_{h,\Omega_k}(|Q|xk)$ is eventually upper bounded by $F_{\Omega_k}(p(|Q|xk))$ for some primitive recursive function $p$. By the argument given in the previous subsection (also see~\cite[Section 3.3]{Ramachandrakumar24}), the running time of the coverability algorithm is then bounded by $F_{\Omega_k}(p(|Q| \cdot |C_f| k))$. This proves that coverability for $k$-NRCS is in $\fastf{\Omega_k}$, thereby showing the required upper bound.

Hence, the only thing remaining is to prove Theorem~\ref{thm:main-upper-bound}, which we will do so in the rest of this section.

\subsection{Upper Bounds for Length Functions of Nested Multiset Nwqos}
\label{subsec:upper-bounds-nested-multiset-nwqos}

We now show how to compute upper bounds for length functions of nested multiset nwqos.
To this end, we will use the framework of Schmitz and Schnoebelen~\cite{SchmitzS11,Schmitz14} for establishing upper bounds for families of nwqos. 
To this end, we begin by introducing an important tool, namely \emph{residuals}.

\begin{definition}[Residuals]
    Let $A$ be a nwqo and $a \in A$. The \emph{residual} $A/a$ is the nwqo induced by the set $\{b \in A : a \nleq b\}$. 
\end{definition}

The following fact from~\cite{SchmitzS11} 
illustrates their usefulness.

\begin{proposition}[Descent Equation]
    For all control functions $g$ and all $n \in \mathbb{N}$, we have
    $$L_{A,g}(n) = \max_{a \in A_{\le n}}\{1 + L_{A/a,g}(g(n))\}$$
\end{proposition}

This descent equation allows us to bound $L_{A,g}$ in terms of $L_{A/a_0,g}$ for some $a_0 \in A_{\le n}$. In turn, this can be bounded by $L_{A/a_0/a_1,g}$ for some $a_1 \in {(A/a_0)}_{\le g(n)}$ and so on. If $A \supsetneq A/a_0 \supsetneq A/a_0/a_1 \cdots $, then $a_0, a_1, \cdots$ is a bad sequence and so in finitely many steps, this unraveling will stop and we can bound $L_{A,g}$. However, in this approach, the residuals $A, A/a_0, A/a_0/a_1, \cdots$ become too complex to work with as the sequence progresses. Hence, we shall use another tool that allows us to replace these residuals with potentially simpler nwqos. This tool is called a \emph{normed reflection}.

\begin{definition}[(Normed) Reflections]
    Let $A, B$ be two nwqos. A \emph{normed reflection} (or just a reflection) is a function $r: A \to B$ s.t.
    \begin{itemize}
        \item For all $x, y \in A, r(x) \le_B r(y)$ implies that $x \le_A y$
        \item For all $x \in A$, $|r(x)|_B \le |x|_A$
    \end{itemize}
\end{definition}

We use the notation $r: A \hookrightarrow B$ to mean that $r$ is a reflection from $A$ to $B$. If the precise mapping $r$ is not important, we will omit the symbol $r$.
It is well known that reflections are transitive, i.e., if $r: A \hookrightarrow B$ and $s: B \hookrightarrow C$, then
$s \circ r: A \hookrightarrow C$ is also a reflection.  The following proposition from~\cite{SchmitzS11} illustrates the fundamental property of reflections.

\begin{proposition}[Fundamental Property of Reflections]
    Suppose $A \hookrightarrow B$ is a reflection. Then $L_{A,g}(n) \le L_{B,g}(n)$ for all control functions $g$ and all $n \in \mathbb{N}$.
\end{proposition}

\subparagraph*{Combining residuals and reflections. } As mentioned before, working with residuals is quite complex. Hence, if we want to work with the descent equation, we need to be able to over-approximate the length function of residuals by means of reflections. To this end, for each nested multiset nwqo $A \neq \Gamma_0$ and each $a \in A_{\le n}$, we define a nested multiset nwqo $R_n(A,a)$ by structural induction as follows:
\begin{itemize}
    \item Suppose $A = \multi{\Gamma_0}$. Then $R_n(A,a) = \Gamma_{0}$.
    \item Suppose $A = \multi{\sum_{i=1}^d B_i}$. Let $b_i$ be the sub-multiset of $a$ restricted to elements from $B_i$. Note that each $b_i \in \multi{B_i}$ and 
    $a$ is the multiset sum of all $b_i$.
    Let $R_n(\multi{B_i},b_i) = \sum_{j=1}^{\ell_i} \multi{B_i^j}$.
    Set $R_n(A,a) = \sum_{i=1}^d \sum_{j=1}^{\ell_i} \multi{\sum_{k \neq i} B_i + B_i^j}$.
    \item Suppose $A = \multi{\multi{B}}$ and $a = \multiset{a_1,\dots,a_m}$. Then, $R_n(A,a)$ is given by $R_n(A,a) = \sum_{i=1}^m \multi{B \cdot (n-1) + R_n(\multi{B},a_i)}$.
    \item Suppose $A = \sum_{i=1}^d \multi{B_i}$. If $a \in \multi{B_i}$, then $R_n(A,a) = (\sum_{k \neq i} \multi{B_k}) + R_n(\multi{B_i},a)$.
\end{itemize}

We now show that the defined nwqos $R_n(A,a)$ allow us to over-approximate residuals using reflections.

\begin{restatable}{lemma}{reflections}\label{lem:reflections}
    For any nested multiset nwqo $A \neq \Gamma_0$ and any $a \in A_{\le n}$, we have that
    $A/a \hookrightarrow R_n(A,a)$.
\end{restatable}

The proof considers all of the four cases of $R_n(A,a)$ as defined above. The proofs for the first and the last cases can be obtained from their definitions. The other two cases
follow from an identity of~\cite[Proposition 1]{Rosa-Velardo17}. Because the identity in that paper is stated in different terms compared to ours,
the complete proof can be found in Appendix~\ref{subsec:appendix-reflections}.

This lemma, when combined with the fundamental property of reflections and the descent equation gives us
\begin{equation}\label{eq:descent-inequality}
  L_{A,g}(n) \le \max_{a \in A_{\le n}} \{1 + L_{R_n(A,a),g}(g(n))\}  
\end{equation}

Hence, instead of working with the descent equation, we can work with this inequality. However, there are two problems here: First, unlike the descent equation, it is not immediately obvious that unraveling this inequality will allow us to terminate after finitely many steps. 
Second, even if we prove that it terminates, our ultimate aim is to relate $L_{A,g}$ to ordinal-indexed functions in the Cicho\'n hierarchy. It is not clear how unraveling this inequality will allow us to accomplish that.

It turns out we can solve both these problems using a single solution: To each nested multiset nwqo, we will associate an 
ordinal measure called its \emph{order type}. We can use this order type to further over-approximate the terms in the above inequality to get a new inequality.
The order types will help us to prove that the unraveling of this new inequality will terminate and will also help us to relate $L_{A,g}$ to functions in the Cicho\'n hierarchy. 

\subsection{Order Type and Derivatives}\label{subsec:order-types}

To each nested multiset nwqo $A$, we can associate an ordinal $o(A)$, called its \emph{order type} as follows:
$$o(\Gamma_0) = 0 \quad o(A+B) = o(A) \oplus o(B) \quad o(\multi{A}) = \omega^{o(A)}$$

For instance, $o(\Gamma_1) = o(\multi{\Gamma_0}) = \omega^{0} = 1$. Similarly, $o(\Gamma_k) = k$ and $o(\multi{\Gamma_k}) = \omega^k$. It is easy to see by structural induction that for two distinct (i.e. not isomorphic) nwqos $A, B$, we have $o(A) \neq o(B)$. 
Hence, this mapping is injective. It turns out that it is also surjective, i.e., for any ordinal $\alpha < \epsilon_0$, we can assign a nested multiset nwqo as follows: Let $\alpha$ in CNF be $\omega^{\beta_1} + \cdots + \omega^{\beta_m}$.
Then we set $C(\alpha) = \Gamma_0$ if $\alpha = 0$ and $C(\alpha) = \sum_{i=1}^m \multi{C(\beta_i)}$ otherwise. By construction, $o$ and $C$ are bijective inverses of each other (modulo isomorphism of nwqos).

Recall that one of our two goals is to replace Equation~\ref{eq:descent-inequality} with a new inequality using order types. To that end, for each $n \in \mathbb{N}$, we  define a \emph{derivative operator} $\delta_n$ on ordinals.
\begin{definition}[Derivative Operator]
    Given $n \ge 1$, we define:
    \begin{itemize}
        \item $D_n(\omega^0) = 0$ 
        \item $D_n(\omega^{\sum_{i=1}^d \beta_i}) = \bigoplus_{i=1}^d \bigoplus_{j=1}^{\ell_i} \omega^{\oplus_{k \neq i} \beta_k \oplus \beta_i^j}$ where for each $1 \le i \le d$, $\bigoplus_{j=1}^{\ell_i} \omega^{\beta_i^j} = D_n(\omega^{\beta_i})$.
        \item $D_n(\omega^{\omega^\beta}) = \omega^{\beta \cdot (n-1) \oplus D_n(\omega^\beta)} \cdot n$
        \item $\delta_n(\sum_{i=1}^d \omega^{\beta_i}) = \{D_n(\omega^{\beta_i}) \oplus \sum_{j \neq i} \omega^{\beta_j} : i = 1, \cdots, d\}$ 
    \end{itemize}
\end{definition}

By induction, it is easy to see that if $\alpha' \in \delta_n(\alpha)$ then $\alpha' < \alpha$. 
We also note the symmetry between the definition of $R_n(A,a)$ and $\delta_n$, in the way that there is a clear correspondence between each of the four cases in both these definitions. This is not by accident; the $\delta_n$ operators are constructed exactly in such a way so that the following lemma is true. (See Appendix~\ref{subsec:appendix-order-types} for more details).

\begin{restatable}{lemma}{ordertype}\label{lem:order-type}
    Let $n \ge 1$ and $a \in A_{\le n}$ for some $A \neq \Gamma_0$. 
    Then, there exists some $\alpha' \in \delta_n(o(A))$ such that 
    $R_n(A,a) \hookrightarrow C(\alpha')$.
\end{restatable}

By combining Equation~\ref{eq:descent-inequality} and Lemma~\ref{lem:order-type}, we get the following inequality, which solves our first goal.
\begin{equation}\label{eq:AC}
 L_{A,g}(n) \le \max_{\alpha' \in \delta_n(o(A))} \{1+ L_{C(\alpha'),g}(g(n))\}   
\end{equation}

We will now move on to our second goal, namely using order types to relate $L_{A,g}$ to functions in the fast-growing hierarchy.

\subsection{Upper Bounds in the Fast-Growing Hierarchy}\label{subsec:Cichon}

For any $\alpha < \epsilon_0$, define 
$$M_{\alpha,g}(n) = \max_{\alpha' \in \delta_n(\alpha)} \{1 + M_{\alpha',g}(g(n))\}$$

Equation~\ref{eq:AC} tells us that $L_{A,g}(n) \le M_{o(A),g}(n)$. In other words, if $\alpha = o(A)$, we have $L_{C(\alpha),g}(n) \le M_{\alpha,g}(n)$. Hence, bounding $M_{\alpha,g}(n)$ will allow us to bound $L_{A,g}(n)$. We now bound $M_{\alpha,g}$ by functions from the Cicho\'n hierarchy $(h_\alpha)_{\alpha < \epsilon_0}$ where $h(x) = x \cdot g(x)$. 
First, we set up some notation and prove a couple of intermediate results.

We say that an ordinal $\alpha$ is $\ell$-lean if either $\alpha = 0$ or $\alpha = \omega^{\beta_1} \cdot c_1 + \cdots + \omega^{\beta_m} \cdot c_m$ in strict CNF
with each $c_i \le \ell$ and $\beta_i$ is $\ell$-lean for every $i$. We now have the following proposition, which can be proved by analysing the shape of the ordinals appearing in the $\delta_n$ operator. (See Appendix~\ref{subsec:appendix-Cichon} for more details).

\begin{restatable}{proposition}{lean}\label{prop:lean}
    Suppose $\Omega_{k} \le \alpha < \Omega_{k+1}$  and $\alpha$ is $\ell$-lean for some $l$.
    If $\alpha' \in \delta_n(\alpha)$ then
    $\alpha'$ is $\ell+\ell n k$-lean. Furthermore, if $\alpha = \omega^\beta$ for some $\beta$, then
    $\alpha' = D_n(\alpha)$ is $\ell n k$-lean.
\end{restatable}

We are now in a position to bound $M_{\alpha,g}$ by means of $h_\alpha$.
(See Appendix~\ref{subsec:appendix-Cichon} for more details).

\begin{restatable}{lemma}{finalbound}\label{lem:final-bound}
    Let $\alpha < \Omega_{k+1}$ be $\ell$-lean. Then $M_{\alpha,g}(n) \le h_\alpha(\ell + \ell n k)$ where $h(x) = x \cdot g(x)$.
\end{restatable}

The following theorem is now immediate from our observation that $L_{A,g} \le M_{o(A),g}$.

\begin{theorem}
    Let $A$ be a nested multiset nwqo of order type $\alpha < \Omega_{k+1}$ such that $\alpha$ is $\ell$-lean.
    Then, $L_{A,g}(n) \le h_{\alpha}(\ell + \ell n k)$.
\end{theorem}

Now, suppose we have a finite set $Q \equiv \Gamma_\ell$ and some $k \in \mathbb{N}$.
Consider the sequence of nested multiset nwqos $\mathcal{M}_k,\mathcal{M}_{k-1},\cdots,\mathcal{M}_0$ defined in Subsection~\ref{subsec:trees-nested-multisets}. 
By construction it follows that $o(\mathcal{M}_i) < \Omega_{k-i+1}$ and $o(\mathcal{M}_i)$ is $\ell$-lean. Hence, $L_{\mathcal{M}_0,g}(n) \le h_{o(\mathcal{M}_0)}(\ell + \ell n k)$. Since $o(\mathcal{M}_0) < \Omega_{k+1}$ and $o(\mathcal{M}_0)$ is $\ell$-lean, it follows that
$h_{o(\mathcal{M}_0)}(\ell + \ell n k) \le h_{\Omega_{k+1}}(\ell + \ell n k)$~\cite[Lemma B.1, Lemma C.9]{arxiv-SchmitzS}. This proves Theorem~\ref{thm:main-upper-bound} and concludes the proof of upper bound for coverability of $k$-NRCS. 

\section{Applications}\label{sec:applications}

We now use the techniques that we have developed in this paper to provide improved upper bounds
for a hierarchy of various problems from XML processing, graph transformation systems, $\pi$-calculus, parameterized verification and logic. Furthermore, we use our newly introduced model of NRCS to prove that the above hierarchy of problems from parameterized verification and logic are also complete problems for the hierarchy $\{\fastf{\Omega_k}\}_{k \in \mathbb{N}}$.
This demonstrates the efficacy of the model of NRCS as a master problem
from which reductions can be easily made.

\subsection{Upper Bounds for Systems over Bounded-Depth Graphs}\label{subsec:bounded-depth-graphs}

Our upper bound techniques developed in Section~\ref{sec:upper-bounds} are more broadly applicable than just for the model of $k$-NRCS, as we shall now demonstrate.
For this, we first set up some notation.

A $k$-depth graph $G$ is any graph that can be obtained by taking a tree $T$ of height at most $k$ and adding edges between a node and its descendants. (Note that no edges are added between two nodes that lie in different subtrees of $T$). We can extend the induced subgraph ordering from trees of height $k$ to $k$-depth graphs in a natural way.
It turns out that the induced subgraph relation $\le_{is}$ is also a nwqo on $k$-depth graphs for any $k$~\cite{ding1992subgraphs,Balasubramanian21}.

With this in mind, we can now consider classes
of transitions systems operating over labelled $k$-depth graphs. Formally, let 
$\mathcal{C}$ be any class of transition systems such that each $\mathcal{T} \in \mathcal{C}$ is a transition system $(C,\rightarrow)$ where $C$ consists of all $k$-depth labelled graphs with labels from some finite set $Q$ and $\rightarrow \subseteq C \times C$ is the step relation. Similar to $k$-NRCS, we can define the notions of $\act{*}$, reachability and coverability for such systems. 

Now, just like $k$-NRCS, if $\mathcal{C}$ satisfies the compatibility and the effective predecessor basis properties, then by the same arguments given in Subsection~\ref{subsec:WSTS}, it would follow that $\mathcal{C}$ is well-structured, coverability is decidable and the running time of the coverability algorithm  is  bounded by the length of controlled bad sequences for labelled $k$-depth graphs under the $\le_{is}$ nwqo. Using existing results, we can show that there is a reflection
from labelled $k$-depth graphs with labels from some set $Q$ to 
labelled trees of height at most $k$ with labels from the set $Q \times 2^{\{0,\dots,k-1\}}$~\cite[Section A]{Balasubramanian21}. (See also Appendix~\ref{sec:appendix-bounded-depth-graphs}).

Because of this reflection, we can conclude that upper bounds for length functions
over $k$-depth graphs can be obtained from upper bounds for length functions over trees of height at most $k$. As shown in subsection~\ref{subsec:trees-nested-multisets}, the latter is bounded by $F_{\Omega_k}(p(n))$ for some primitive recursive function $p$. This would then put the coverability problem for $\mathcal{C}$ in the class $\fastf{\Omega_k}$.

We can use this observation on various classes of systems from the literature. 
For instance, motivated by XML processing of semi-structured data, the authors of~\cite{GenestMSZ08} introduced \emph{tree rewriting pattern systems} (TRPS), which are a class of systems over trees. They showed that a subclass of TRPS called positive, $k$-depth-bounded TRPS have as configurations, labelled trees of height at most $k$ and possess the compatibility property as well as the effective predecessor basis property. Hence, by our results it follows that
\begin{theorem}
    Coverability\footnote{In~\cite{GenestMSZ08} the authors call this the pattern reachability problem} for positive $k$-depth-bounded TRPS is in $\fastf{\Omega_k}$.
\end{theorem}

Similarly, the authors of~\cite{KoenigS} consider a class of graph transformation systems working over $k$-depth graphs, called $\mathcal{G}_k$-restricted systems, for which the above argument applies and so we get,
\begin{theorem}
    Coverability for $\mathcal{G}_{k}$-restricted graph transformation systems is in $\fastf{\Omega_k}$.
\end{theorem}

Finally, in~\cite{Meyer08}, it was shown that $k$-depth-bounded $\pi$-calculus processes can be reduced to reasoning over trees of height at most $2^k - 1$, and so by the above argument we can conclude
\begin{theorem}
    Coverability for $k$-depth-bounded $\pi$-calculus processes is in $\fastf{\Omega_{2^k-1}}$.
\end{theorem}

We note that the best upper bounds that could be deduced from the literature prior to our work were $\fastf{\Omega_{2k+1}}$ for the first two problems and $\fastf{\Omega_{2^{k+1}-1}}$ for the last one~\cite{PriorityChannelSystems}. Hence, our results significantly improve upon these bounds. We now show more such applications in the next subsections, along with matching lower bounds.

\subsection{Bounded-Depth Broadcast Networks}\label{subsec:broadcast-networks}

The next application that we give is for the model of bounded-depth broadcast networks
\footnote{They are sometimes called bounded-path broadcast networks in the literature, but since they also operate on bounded-depth graphs, for the sake of uniformity, we call them bounded-depth broadcast networks here.}. 

A $k$-depth broadcast network~\cite{DelzannoSZ10} is a model of concurrent computation, in which we have finite-state agents situated on the nodes of a $k$-depth graph, called the communication topology.
All the agents execute the same finite-state protocol and initially begin at some initial state of the protocol. In each step, an agent can communicate with its neighbors by means of broadcasts, i.e., in each step, some agent sends a message which is immediately received by all of its neighbors. Therefore, configurations
of a $k$-depth broadcast network $\mathcal{P}$ are labelled $k$-depth graphs with labels from some finite set $Q$ of states for the agents. Such networks are known to satisfy the 
compatibility as well as the effective predecessor basis property.

The coverability problem for such networks is slightly different from the usual notion of coverability, because we no longer have one initial configuration, but an infinite set of initial configurations, which are all possible $k$-depth graphs with labels
from a given set $Q_{init} \subseteq Q$ of initial states. Nevertheless the usual coverability algorithm can be used for such networks (with the final check being over the set of all initial configurations)~\cite[Section 5]{DelzannoSZ10}. Just as for the other models, the running time depends primarily on the length of controlled bad sequences for $k$-depth graphs. Hence, we get that
\begin{theorem}
    Coverability for $k$-depth broadcast networks is in $\fastf{\Omega_k}$.
\end{theorem}

Similarly, if we restrict the network to only operate over communication topologies of $k$-depth trees, i.e., trees of height at most $k$, we get that
\begin{theorem}
    Coverability for $k$-depth broadcast networks over trees is in $\fastf{\Omega_k}$.
\end{theorem}

We also show that both these problems are $\fastf{\Omega_k}$-hard, by giving a reduction from coverability of $k$-NRCS to both these problems. The hardness result is obtained by taking an existing reduction from $k$-NCS, i.e., $k$-NRCS without resets to $k$-depth broadcast networks. We utilise this reduction along with the fact that broadcasts can be used to simulate
resets of counters~\cite{EsparzaFM99} to give an improved reduction from $k$-NRCS to $k$-depth broadcast networks. (See Appendix~\ref{sec:appendix-broadcast-networks} for more details).

\begin{restatable}{theorem}{bcastnets}\label{thm:broadcast-networks}
    Coverability for $k$-depth broadcast networks and $k$-depth broadcast networks over trees are $\fastf{\Omega_k}$-complete.
\end{restatable}

Previously the best known result was a $\fastf{\epsilon_0}$-completeness result for these problems, when $k$ is part of the input~\cite[Theorem 16]{Balasubramanian21}. Our new result is a refined parameterized analogue of this, giving a completeness result for each fixed $k$.

\subsection{Freeze LTL with Ordered Attributes}\label{subsec:freeze-LTL}

Finally, we also use our results to present an application related to an extension of LTL
called \emph{freeze LTL with ordered attributes}~\cite{DeckerT16}. 
This is an extension of the usual LTL to work over \emph{data words} instead of words, which we now explain. (Our presentation is slightly different, but equivalent to the one given in~\cite{DeckerT16}).

Fix some $k \in \mathbb{N}$, a finite set of atomic propositions $\Sigma$ 
and an infinite set $\mathbb{D}$.
A $k$-attributed data letter (or simply data letter) is a tuple $(A,d_1,\dots,d_k)$ where $A \subseteq \Sigma$ and each $d_i \in \mathbb{D}$. A data word is a sequence of data letters.
Intuitively, a data word is an extension of the usual notion of words to deal with
infinitely many letters. Freeze LTL with $k$-ordered attributes is an extension of LTL to talk about such data words. In addition to the usual boolean and LTL operators $\land, \lor, \lnot, X, U, F$ and $G$, it also has freeze mechanisms $\uparrow^1,\uparrow^2,\dots,\uparrow^k,\downarrow^1,\downarrow^2,\dots,\downarrow^k$.
Intuitively, the application of $\downarrow^i$ at a letter $(A,d_1,\dots,d_k)$ at some position along a word allows us to ``remember'' the data values $d_1,\dots,d_i$. 
If we have applied $\downarrow^i$ to remember data values $d_1,\dots,d_i$ at some point, we can apply $\uparrow^i$ at some later position $(B,d'_1,\dots,d'_k)$ along the word. This would allow us to ``compare'' $d'_1,\dots,d'_i$ with the already remembered $(d_1,\dots,d_i)$ and check if they are equal. The formal syntax and semantics is given in Appendix~\ref{subsec:appendix-syntax-freeze-LTL}.

Given a freeze LTL formula with $k$-ordered attributes, we say that it is satisfiable if there is some data word that satisfies it. The satisfiability problem is to decide whether a given freeze LTL formula is satisfiable. By giving reductions to and from the coverability problem for $k$-NRCS we show that
\begin{restatable}{theorem}{freezeLTL}\label{thm:freeze-LTL}
     The satisfiability problem for freeze LTL with $k$-ordered attributes is $\fastf{\Omega_k}$-complete.
\end{restatable}

To prove the upper bound, we use a result of~\cite[Theorem 12]{DeckerT16} which reduces this problem
to the coverability problem for $(k+1)$-NCS, i.e., $(k+1)$-NRCS without resets. 
Furthermore, the particular $(k+1)$-NCS that is constructed during the reduction has the property that, in any reachable configuration, there are at most 2 children of the root. We show that any such $(k+1)$-NCS can be simulated by a $k$-NRCS.
The intuition is that in the $k$-NRCS that we construct,
the root will simply absorb its 2 children without explicitly keeping them as nodes. 
Operations involving these 2 children nodes in the original $(k+1)$-NCS, can then be simulated by reset operations. This allows us to save a level, thereby giving us a $k$-NRCS.

To prove the lower bound, we use a result of~\cite[Theorem 13]{DeckerT16} which gives a reduction
from coverability for $k$-NCS to satisfiability for freeze LTL with $k$-ordered attributes.
The intuition is that the sequence of configurations in a run of a $k$-NCS can be encoded \emph{in the reverse} as a data word,
which can then be checked by a freeze LTL formula. We show that that reduction can be generalized to $k$-NRCS by a modification taking into account reset transitions.
For this, we need to ensure that in certain configurations representing the aftermath of firing a reset transition, there is no occurrence of a particular atomic proposition. This in turn can be encoded using the global operator $G$ of LTL.

Prior to our work, the best known upper and lower bounds for freeze LTL with $k$-ordered attributes were $\fastf{\Omega_{2(k+1)}}$ and $\fastf{\Omega_{k-1}}$. Our results improve upon both these fronts and give tight parameterized bounds with respect to the parameter $k$.

\section{Conclusion}\label{sec:conclusion}

We have provided the first natural collection of complete problem for the hierarchy $(\fastf{\Omega_k})_{k \in \mathbb{N}}$, namely coverability for $k$-NRCS. Furthermore, to prove our upper bounds, we have also developed length function theorems for nested applications of the multiset operator on nwqos.
Finally, we also presented various applications of our upper bound techniques and the newly introduced model of NRCS to improve upper bounds for various other problems and identify
new complete problems for $\fastf{\Omega_k}$, for all $k$.



\bibliography{refs}

@article{Schmitz16hierarchies,
  author       = {Sylvain Schmitz},
  title        = {Complexity Hierarchies beyond Elementary},
  journal      = {{ACM} Trans. Comput. Theory},
  volume       = {8},
  number       = {1},
  pages        = {3:1--3:36},
  year         = {2016},
  url          = {https://doi.org/10.1145/2858784},
  doi          = {10.1145/2858784},
  bibsource    = {dblp computer science bibliography, https://dblp.org}
}

@inproceedings{DeckerT16,
  author       = {Normann Decker and
                  Daniel Thoma},
  editor       = {Bart Jacobs and
                  Christof L{\"{o}}ding},
  title        = {On Freeze {LTL} with Ordered Attributes},
  booktitle    = {Foundations of Software Science and Computation Structures - 19th
                  International Conference, {FOSSACS} 2016, Held as Part of the European
                  Joint Conferences on Theory and Practice of Software, {ETAPS} 2016,
                  Eindhoven, The Netherlands, April 2-8, 2016, Proceedings},
  series       = {Lecture Notes in Computer Science},
  volume       = {9634},
  pages        = {269--284},
  publisher    = {Springer},
  year         = {2016},
  url          = {https://doi.org/10.1007/978-3-662-49630-5\_16},
  doi          = {10.1007/978-3-662-49630-5\_16},
  bibsource    = {dblp computer science bibliography, https://dblp.org}
}

@article{DeckerT15-fullversion,
  author       = {Normann Decker and
                  Daniel Thoma},
  title        = {On Freeze {LTL} with Ordered Attributes},
  journal      = {CoRR},
  volume       = {abs/1504.06355},
  year         = {2015},
  url          = {http://arxiv.org/abs/1504.06355},
  eprinttype    = {arXiv},
  eprint       = {1504.06355},
  bibsource    = {dblp computer science bibliography, https://dblp.org}
}

@inproceedings{LomazovaS99,
  author       = {Irina A. Lomazova and
                  Philippe Schnoebelen},
  editor       = {Dines Bj{\o}rner and
                  Manfred Broy and
                  Alexandre V. Zamulin},
  title        = {Some Decidability Results for Nested Petri Nets},
  booktitle    = {Perspectives of System Informatics, Third International Andrei Ershov
                  Memorial Conference, PSI'99, Akademgorodok, Novosibirsk, Russia, July
                  6-9, 1999, Proceedings},
  series       = {Lecture Notes in Computer Science},
  volume       = {1755},
  pages        = {208--220},
  publisher    = {Springer},
  year         = {1999},
  url          = {https://doi.org/10.1007/3-540-46562-6\_18},
  doi          = {10.1007/3-540-46562-6\_18},
  bibsource    = {dblp computer science bibliography, https://dblp.org}
}

@inproceedings{Schmitz14,
  author       = {Sylvain Schmitz},
  editor       = {Jo{\"{e}}l Ouaknine and
                  Igor Potapov and
                  James Worrell},
  title        = {Complexity Bounds for Ordinal-Based Termination - (Invited Talk)},
  booktitle    = {Reachability Problems - 8th International Workshop, {RP} 2014, Oxford,
                  UK, September 22-24, 2014. Proceedings},
  series       = {Lecture Notes in Computer Science},
  volume       = {8762},
  pages        = {1--19},
  publisher    = {Springer},
  year         = {2014},
  url          = {https://doi.org/10.1007/978-3-319-11439-2\_1},
  doi          = {10.1007/978-3-319-11439-2\_1},
  bibsource    = {dblp computer science bibliography, https://dblp.org}
}

@inproceedings{Schnoebelen10,
  author       = {Philippe Schnoebelen},
  editor       = {Petr Hlinen{\'{y}} and
                  Anton{\'{\i}}n Kucera},
  title        = {Revisiting Ackermann-Hardness for Lossy Counter Machines and Reset
                  Petri Nets},
  booktitle    = {Mathematical Foundations of Computer Science 2010, 35th International
                  Symposium, {MFCS} 2010, Brno, Czech Republic, August 23-27, 2010.
                  Proceedings},
  series       = {Lecture Notes in Computer Science},
  volume       = {6281},
  pages        = {616--628},
  publisher    = {Springer},
  year         = {2010},
  url          = {https://doi.org/10.1007/978-3-642-15155-2\_54},
  doi          = {10.1007/978-3-642-15155-2\_54},
  bibsource    = {dblp computer science bibliography, https://dblp.org}
}

@inproceedings{SchmitzS11,
  author       = {Sylvain Schmitz and
                  Philippe Schnoebelen},
  editor       = {Luca Aceto and
                  Monika Henzinger and
                  Jir{\'{\i}} Sgall},
  title        = {Multiply-Recursive Upper Bounds with Higman's Lemma},
  booktitle    = {Automata, Languages and Programming - 38th International Colloquium,
                  {ICALP} 2011, Zurich, Switzerland, July 4-8, 2011, Proceedings, Part
                  {II}},
  series       = {Lecture Notes in Computer Science},
  volume       = {6756},
  pages        = {441--452},
  publisher    = {Springer},
  year         = {2011},
  url          = {https://doi.org/10.1007/978-3-642-22012-8\_35},
  doi          = {10.1007/978-3-642-22012-8\_35},
  bibsource    = {dblp computer science bibliography, https://dblp.org}
}

@article{higman1952ordering,
  title={Ordering by divisibility in abstract algebras},
  author={Higman, Graham},
  journal={Proceedings of the London Mathematical Society},
  volume={3},
  number={1},
  pages={326--336},
  year={1952},
  publisher={Oxford University Press}
}

@article{FinkelS01,
  author       = {Alain Finkel and
                  Philippe Schnoebelen},
  title        = {Well-structured transition systems everywhere!},
  journal      = {Theor. Comput. Sci.},
  volume       = {256},
  number       = {1-2},
  pages        = {63--92},
  year         = {2001},
  url          = {https://doi.org/10.1016/S0304-3975(00)00102-X},
  doi          = {10.1016/S0304-3975(00)00102-X},
  bibsource    = {dblp computer science bibliography, https://dblp.org}
}

@inproceedings{AbdullaCJT96,
  author       = {Parosh Aziz Abdulla and
                  Karlis Cerans and
                  Bengt Jonsson and
                  Yih{-}Kuen Tsay},
  title        = {General Decidability Theorems for Infinite-State Systems},
  booktitle    = {Proceedings, 11th Annual {IEEE} Symposium on Logic in Computer Science,
                  New Brunswick, New Jersey, USA, July 27-30, 1996},
  pages        = {313--321},
  publisher    = {{IEEE} Computer Society},
  year         = {1996},
  url          = {https://doi.org/10.1109/LICS.1996.561359},
  doi          = {10.1109/LICS.1996.561359},
  bibsource    = {dblp computer science bibliography, https://dblp.org}
}

@article{PriorityChannelSystems,
  author       = {Christoph Haase and
                  Sylvain Schmitz and
                  Philippe Schnoebelen},
  title        = {The Power of Priority Channel Systems},
  journal      = {Log. Methods Comput. Sci.},
  volume       = {10},
  number       = {4},
  year         = {2014},
  url          = {https://doi.org/10.2168/LMCS-10(4:4)2014},
  doi          = {10.2168/LMCS-10(4:4)2014},
  bibsource    = {dblp computer science bibliography, https://dblp.org}
}

@phdthesis{Ramachandrakumar24,
  author       = {Balasubramanian Ayikudi Ramachandrakumar},
  title        = {Computational Complexity of Verifying Parameterized Systems},
  school       = {Technical University of Munich, Germany},
  year         = {2024},
  url          = {https://nbn-resolving.org/urn:nbn:de:bvb:91-diss-20240417-1714488-1-4},
  urn          = {urn:nbn:de:bvb:91-diss-20240417-1714488-1-4},
  bibsource    = {dblp computer science bibliography, https://dblp.org}
}

@inproceedings{SchmitzS13,
  author       = {Sylvain Schmitz and
                  Philippe Schnoebelen},
  editor       = {Pedro R. D'Argenio and
                  Hern{\'{a}}n C. Melgratti},
  title        = {The Power of Well-Structured Systems},
  booktitle    = {{CONCUR} 2013 - Concurrency Theory - 24th International Conference,
                  {CONCUR} 2013, Buenos Aires, Argentina, August 27-30, 2013. Proceedings},
  series       = {Lecture Notes in Computer Science},
  volume       = {8052},
  pages        = {5--24},
  publisher    = {Springer},
  year         = {2013},
  url          = {https://doi.org/10.1007/978-3-642-40184-8\_2},
  doi          = {10.1007/978-3-642-40184-8\_2},
  bibsource    = {dblp computer science bibliography, https://dblp.org}
}

@article{Rosa-Velardo17,
  author       = {Fernando Rosa{-}Velardo},
  title        = {Ordinal recursive complexity of Unordered Data Nets},
  journal      = {Inf. Comput.},
  volume       = {254},
  pages        = {41--58},
  year         = {2017},
  url          = {https://doi.org/10.1016/j.ic.2017.02.002},
  doi          = {10.1016/J.IC.2017.02.002},
  bibsource    = {dblp computer science bibliography, https://dblp.org}
}

@article{arxiv-SchmitzS,
  author       = {Sylvain Schmitz and
                  Philippe Schnoebelen},
  title        = {Multiply-Recursive Upper Bounds with Higman's Lemma},
  journal      = {CoRR},
  volume       = {abs/1103.4399},
  year         = {2011},
  url          = {http://arxiv.org/abs/1103.4399},
  eprinttype    = {arXiv},
  eprint       = {1103.4399},
  bibsource    = {dblp computer science bibliography, https://dblp.org}
}

@inproceedings{GenestMSZ08,
  author       = {Blaise Genest and
                  Anca Muscholl and
                  Olivier Serre and
                  Marc Zeitoun},
  editor       = {Sung Deok Cha and
                  Jin{-}Young Choi and
                  Moonzoo Kim and
                  Insup Lee and
                  Mahesh Viswanathan},
  title        = {Tree Pattern Rewriting Systems},
  booktitle    = {Automated Technology for Verification and Analysis, 6th International
                  Symposium, {ATVA} 2008, Seoul, Korea, October 20-23, 2008. Proceedings},
  series       = {Lecture Notes in Computer Science},
  volume       = {5311},
  pages        = {332--346},
  publisher    = {Springer},
  year         = {2008},
  url          = {https://doi.org/10.1007/978-3-540-88387-6\_29},
  doi          = {10.1007/978-3-540-88387-6\_29},
  bibsource    = {dblp computer science bibliography, https://dblp.org}
}

@article{ding1992subgraphs,
  title={Subgraphs and well-quasi-ordering},
  author={Ding, Guoli},
  journal={Journal of Graph Theory},
  volume={16},
  number={5},
  pages={489--502},
  year={1992},
  publisher={Wiley Online Library}
}

@inproceedings{Balasubramanian21,
  author       = {A. R. Balasubramanian},
  editor       = {Mikolaj Bojanczyk and
                  Chandra Chekuri},
  title        = {Complexity of Coverability in Bounded Path Broadcast Networks},
  booktitle    = {41st {IARCS} Annual Conference on Foundations of Software Technology
                  and Theoretical Computer Science, {FSTTCS} 2021, Virtual Conference,
                  December 15-17, 2021},
  series       = {LIPIcs},
  volume       = {213},
  pages        = {35:1--35:16},
  publisher    = {Schloss Dagstuhl - Leibniz-Zentrum f{\"{u}}r Informatik},
  year         = {2021},
  url          = {https://doi.org/10.4230/LIPIcs.FSTTCS.2021.35},
  doi          = {10.4230/LIPICS.FSTTCS.2021.35},
  bibsource    = {dblp computer science bibliography, https://dblp.org}
}

@article{KoenigS,
  author       = {Barbara K{\"{o}}nig and
                  Jan St{\"{u}}ckrath},
  title        = {Well-structured graph transformation systems},
  journal      = {Inf. Comput.},
  volume       = {252},
  pages        = {71--94},
  year         = {2017},
  url          = {https://doi.org/10.1016/j.ic.2016.03.005},
  doi          = {10.1016/J.IC.2016.03.005},
  bibsource    = {dblp computer science bibliography, https://dblp.org}
}

@inproceedings{Meyer08,
  author       = {Roland Meyer},
  editor       = {Giorgio Ausiello and
                  Juhani Karhum{\"{a}}ki and
                  Giancarlo Mauri and
                  C.{-}H. Luke Ong},
  title        = {On Boundedness in Depth in the pi-Calculus},
  booktitle    = {Fifth {IFIP} International Conference On Theoretical Computer Science
                  - {TCS} 2008, {IFIP} 20th World Computer Congress, {TC} 1, Foundations
                  of Computer Science, September 7-10, 2008, Milano, Italy},
  series       = {{IFIP}},
  volume       = {273},
  pages        = {477--489},
  publisher    = {Springer},
  year         = {2008},
  url          = {https://doi.org/10.1007/978-0-387-09680-3\_32},
  doi          = {10.1007/978-0-387-09680-3\_32},
  bibsource    = {dblp computer science bibliography, https://dblp.org}
}

@inproceedings{DelzannoSZ10,
  author       = {Giorgio Delzanno and
                  Arnaud Sangnier and
                  Gianluigi Zavattaro},
  editor       = {Paul Gastin and
                  Fran{\c{c}}ois Laroussinie},
  title        = {Parameterized Verification of Ad Hoc Networks},
  booktitle    = {{CONCUR} 2010 - Concurrency Theory, 21th International Conference,
                  {CONCUR} 2010, Paris, France, August 31-September 3, 2010. Proceedings},
  series       = {Lecture Notes in Computer Science},
  volume       = {6269},
  pages        = {313--327},
  publisher    = {Springer},
  year         = {2010},
  url          = {https://doi.org/10.1007/978-3-642-15375-4\_22},
  doi          = {10.1007/978-3-642-15375-4\_22},
  bibsource    = {dblp computer science bibliography, https://dblp.org}
}

@inproceedings{EsparzaFM99,
  author       = {Javier Esparza and
                  Alain Finkel and
                  Richard Mayr},
  title        = {On the Verification of Broadcast Protocols},
  booktitle    = {14th Annual {IEEE} Symposium on Logic in Computer Science, Trento,
                  Italy, July 2-5, 1999},
  pages        = {352--359},
  publisher    = {{IEEE} Computer Society},
  year         = {1999},
  url          = {https://doi.org/10.1109/LICS.1999.782630},
  doi          = {10.1109/LICS.1999.782630},
  bibsource    = {dblp computer science bibliography, https://dblp.org}
}

@article{FischerMR68,
  author       = {Patrick C. Fischer and
                  Albert R. Meyer and
                  Arnold L. Rosenberg},
  title        = {Counter Machines and Counter Languages},
  journal      = {Math. Syst. Theory},
  volume       = {2},
  number       = {3},
  pages        = {265--283},
  year         = {1968},
  url          = {https://doi.org/10.1007/BF01694011},
  doi          = {10.1007/BF01694011},
  bibsource    = {dblp computer science bibliography, https://dblp.org}
}

@inproceedings{LerouxS19,
  author       = {J{\'{e}}r{\^{o}}me Leroux and
                  Sylvain Schmitz},
  title        = {Reachability in Vector Addition Systems is Primitive-Recursive in
                  Fixed Dimension},
  booktitle    = {34th Annual {ACM/IEEE} Symposium on Logic in Computer Science, {LICS}
                  2019, Vancouver, BC, Canada, June 24-27, 2019},
  pages        = {1--13},
  publisher    = {{IEEE}},
  year         = {2019},
  url          = {https://doi.org/10.1109/LICS.2019.8785796},
  doi          = {10.1109/LICS.2019.8785796},
  bibsource    = {dblp computer science bibliography, https://dblp.org}
}

@inproceedings{Leroux21,
  author       = {J{\'{e}}r{\^{o}}me Leroux},
  title        = {The Reachability Problem for Petri Nets is Not Primitive Recursive},
  booktitle    = {62nd {IEEE} Annual Symposium on Foundations of Computer Science, {FOCS}
                  2021, Denver, CO, USA, February 7-10, 2022},
  pages        = {1241--1252},
  publisher    = {{IEEE}},
  year         = {2021},
  url          = {https://doi.org/10.1109/FOCS52979.2021.00121},
  doi          = {10.1109/FOCS52979.2021.00121},
  bibsource    = {dblp computer science bibliography, https://dblp.org}
}

@inproceedings{CzerwinskiO21,
  author       = {Wojciech Czerwinski and
                  Lukasz Orlikowski},
  title        = {Reachability in Vector Addition Systems is Ackermann-complete},
  booktitle    = {62nd {IEEE} Annual Symposium on Foundations of Computer Science, {FOCS}
                  2021, Denver, CO, USA, February 7-10, 2022},
  pages        = {1229--1240},
  publisher    = {{IEEE}},
  year         = {2021},
  url          = {https://doi.org/10.1109/FOCS52979.2021.00120},
  doi          = {10.1109/FOCS52979.2021.00120},
  bibsource    = {dblp computer science bibliography, https://dblp.org}
}

@inproceedings{Schmitz19,
  author       = {Sylvain Schmitz},
  editor       = {Christel Baier and
                  Ioannis Chatzigiannakis and
                  Paola Flocchini and
                  Stefano Leonardi},
  title        = {The Parametric Complexity of Lossy Counter Machines},
  booktitle    = {46th International Colloquium on Automata, Languages, and Programming,
                  {ICALP} 2019, Patras, Greece, July 9-12, 2019},
  series       = {LIPIcs},
  volume       = {132},
  pages        = {129:1--129:15},
  publisher    = {Schloss Dagstuhl - Leibniz-Zentrum f{\"{u}}r Informatik},
  year         = {2019},
  url          = {https://doi.org/10.4230/LIPIcs.ICALP.2019.129},
  doi          = {10.4230/LIPICS.ICALP.2019.129},
  bibsource    = {dblp computer science bibliography, https://dblp.org}
}

@inproceedings{FigueiraFSS11,
  author       = {Diego Figueira and
                  Santiago Figueira and
                  Sylvain Schmitz and
                  Philippe Schnoebelen},
  title        = {Ackermannian and Primitive-Recursive Bounds with Dickson's Lemma},
  booktitle    = {Proceedings of the 26th Annual {IEEE} Symposium on Logic in Computer
                  Science, {LICS} 2011, June 21-24, 2011, Toronto, Ontario, Canada},
  pages        = {269--278},
  publisher    = {{IEEE} Computer Society},
  year         = {2011},
  url          = {https://doi.org/10.1109/LICS.2011.39},
  doi          = {10.1109/LICS.2011.39},
  bibsource    = {dblp computer science bibliography, https://dblp.org}
}

@article{LasotaW08,
  author       = {Slawomir Lasota and
                  Igor Walukiewicz},
  title        = {Alternating timed automata},
  journal      = {{ACM} Trans. Comput. Log.},
  volume       = {9},
  number       = {2},
  pages        = {10:1--10:27},
  year         = {2008},
  url          = {https://doi.org/10.1145/1342991.1342994},
  doi          = {10.1145/1342991.1342994},
  bibsource    = {dblp computer science bibliography, https://dblp.org}
}

@article{DemriL09,
  author       = {St{\'{e}}phane Demri and
                  Ranko Lazic},
  title        = {{LTL} with the freeze quantifier and register automata},
  journal      = {{ACM} Trans. Comput. Log.},
  volume       = {10},
  number       = {3},
  pages        = {16:1--16:30},
  year         = {2009},
  url          = {https://doi.org/10.1145/1507244.1507246},
  doi          = {10.1145/1507244.1507246},
  bibsource    = {dblp computer science bibliography, https://dblp.org}
}

@inproceedings{LazicS16,
  author       = {Ranko Lazic and
                  Sylvain Schmitz},
  editor       = {Martin Grohe and
                  Eric Koskinen and
                  Natarajan Shankar},
  title        = {The Complexity of Coverability in {\(\nu\)}-Petri Nets},
  booktitle    = {Proceedings of the 31st Annual {ACM/IEEE} Symposium on Logic in Computer
                  Science, {LICS} '16, New York, NY, USA, July 5-8, 2016},
  pages        = {467--476},
  publisher    = {{ACM}},
  year         = {2016},
  url          = {https://doi.org/10.1145/2933575.2933593},
  doi          = {10.1145/2933575.2933593},
  bibsource    = {dblp computer science bibliography, https://dblp.org}
}

@inproceedings{HaddadSS12,
  author       = {Serge Haddad and
                  Sylvain Schmitz and
                  Philippe Schnoebelen},
  title        = {The Ordinal-Recursive Complexity of Timed-arc Petri Nets, Data Nets,
                  and Other Enriched Nets},
  booktitle    = {Proceedings of the 27th Annual {IEEE} Symposium on Logic in Computer
                  Science, {LICS} 2012, Dubrovnik, Croatia, June 25-28, 2012},
  pages        = {355--364},
  publisher    = {{IEEE} Computer Society},
  year         = {2012},
  url          = {https://doi.org/10.1109/LICS.2012.46},
  doi          = {10.1109/LICS.2012.46},
  bibsource    = {dblp computer science bibliography, https://dblp.org}
}

@inproceedings{Balasubramanian22,
  author       = {A. R. Balasubramanian},
  editor       = {Bartek Klin and
                  Slawomir Lasota and
                  Anca Muscholl},
  title        = {Complexity of Coverability in Depth-Bounded Processes},
  booktitle    = {33rd International Conference on Concurrency Theory, {CONCUR} 2022,
                  Warsaw, Poland, September 12-16, 2022},
  series       = {LIPIcs},
  volume       = {243},
  pages        = {17:1--17:19},
  publisher    = {Schloss Dagstuhl - Leibniz-Zentrum f{\"{u}}r Informatik},
  year         = {2022},
  url          = {https://doi.org/10.4230/LIPIcs.CONCUR.2022.17},
  doi          = {10.4230/LIPICS.CONCUR.2022.17},
  bibsource    = {dblp computer science bibliography, https://dblp.org}
}

@article{AmadioM02,
  author       = {Roberto M. Amadio and
                  Charles Meyssonnier},
  title        = {On Decidability of the Control Reachability Problem in the Asynchronous
                  pi-Calculus},
  journal      = {Nord. J. Comput.},
  volume       = {9},
  number       = {1},
  pages        = {70--101},
  year         = {2002},
  bibsource    = {dblp computer science bibliography, https://dblp.org}
}

@inproceedings{JancarS19,
  author       = {Petr Jancar and
                  Sylvain Schmitz},
  title        = {Bisimulation Equivalence of First-Order Grammars is ACKERMANN-Complete},
  booktitle    = {34th Annual {ACM/IEEE} Symposium on Logic in Computer Science, {LICS}
                  2019, Vancouver, BC, Canada, June 24-27, 2019},
  pages        = {1--12},
  publisher    = {{IEEE}},
  year         = {2019},
  url          = {https://doi.org/10.1109/LICS.2019.8785848},
  doi          = {10.1109/LICS.2019.8785848},
  bibsource    = {dblp computer science bibliography, https://dblp.org}
}

@inproceedings{Hague14,
  author       = {Matthew Hague},
  editor       = {Thomas A. Henzinger and
                  Dale Miller},
  title        = {Senescent ground tree rewrite systems},
  booktitle    = {Joint Meeting of the Twenty-Third {EACSL} Annual Conference on Computer
                  Science Logic {(CSL)} and the Twenty-Ninth Annual {ACM/IEEE} Symposium
                  on Logic in Computer Science (LICS), {CSL-LICS} 2014, Vienna, Austria,
                  July 14 - 18, 2014},
  pages        = {48:1--48:10},
  publisher    = {{ACM}},
  year         = {2014},
  url          = {https://doi.org/10.1145/2603088.2603112},
  doi          = {10.1145/2603088.2603112},
  bibsource    = {dblp computer science bibliography, https://dblp.org}
}

@inproceedings{lics/Balasubramanian21,
  author       = {A. R. Balasubramanian and
                  Timo Lang and
                  Revantha Ramanayake},
  title        = {Decidability and Complexity in Weakening and Contraction Hypersequent
                  Substructural Logics},
  booktitle    = {36th Annual {ACM/IEEE} Symposium on Logic in Computer Science, {LICS}
                  2021, Rome, Italy, June 29 - July 2, 2021},
  pages        = {1--13},
  publisher    = {{IEEE}},
  year         = {2021},
  url          = {https://doi.org/10.1109/LICS52264.2021.9470733},
  doi          = {10.1109/LICS52264.2021.9470733},
  bibsource    = {dblp computer science bibliography, https://dblp.org}
}

@article{GreatiR25,
  author       = {Vitor Greati and
                  Revantha Ramanayake},
  title        = {Tight length theorems for multiset extensions of Higman's lemma},
  journal      = {Theor. Comput. Sci.},
  volume       = {1057},
  pages        = {115546},
  year         = {2025},
  url          = {https://doi.org/10.1016/j.tcs.2025.115546},
  doi          = {10.1016/J.TCS.2025.115546},
  bibsource    = {dblp computer science bibliography, https://dblp.org}
}

@inproceedings{GreatiR24,
  author       = {Vitor Greati and
                  Revantha Ramanayake},
  editor       = {Agata Ciabattoni and
                  David Gabelaia and
                  Igor Sedl{\'{a}}r},
  title        = {Deducibility in the Full Lambek Calculus with Weakening Is HAck-Complete},
  booktitle    = {Advances in Modal Logic, AiML 2024, Prague, Czech Republic, August
                  19-23, 2024},
  pages        = {401--422},
  publisher    = {College Publications},
  year         = {2024},
  url          = {http://www.aiml.net/volumes/volume15/19-GreatiRamanayake.pdf},
  bibsource    = {dblp computer science bibliography, https://dblp.org}
}

@article{LazicS15,
  author       = {Ranko Lazic and
                  Sylvain Schmitz},
  title        = {Nonelementary Complexities for Branching VASS, MELL, and Extensions},
  journal      = {{ACM} Trans. Comput. Log.},
  volume       = {16},
  number       = {3},
  pages        = {20:1--20:30},
  year         = {2015},
  url          = {https://doi.org/10.1145/2733375},
  doi          = {10.1145/2733375},
  bibsource    = {dblp computer science bibliography, https://dblp.org}
}

@article{LazicOW16,
  author       = {Ranko Lazic and
                  Jo{\"{e}}l Ouaknine and
                  James Worrell},
  title        = {Zeno, Hercules, and the Hydra: Safety Metric Temporal Logic is Ackermann-Complete},
  journal      = {{ACM} Trans. Comput. Log.},
  volume       = {17},
  number       = {3},
  pages        = {16},
  year         = {2016},
  url          = {https://doi.org/10.1145/2874774},
  doi          = {10.1145/2874774},
  bibsource    = {dblp computer science bibliography, https://dblp.org}
}

@inproceedings{ElgaardKM98,
  author       = {Jacob Elgaard and
                  Nils Klarlund and
                  Anders M{\o}ller},
  editor       = {Alan J. Hu and
                  Moshe Y. Vardi},
  title        = {{MONA} 1.x: New Techniques for {WS1S} and {WS2S}},
  booktitle    = {Computer Aided Verification, 10th International Conference, {CAV}
                  '98, Vancouver, BC, Canada, June 28 - July 2, 1998, Proceedings},
  series       = {Lecture Notes in Computer Science},
  volume       = {1427},
  pages        = {516--520},
  publisher    = {Springer},
  year         = {1998},
  url          = {https://doi.org/10.1007/BFb0028773},
  doi          = {10.1007/BFB0028773},
  bibsource    = {dblp computer science bibliography, https://dblp.org}
}

@inproceedings{BlondinH17,
  author       = {Michael Blondin and
                  Christoph Haase},
  title        = {Logics for continuous reachability in Petri nets and vector addition
                  systems with states},
  booktitle    = {32nd Annual {ACM/IEEE} Symposium on Logic in Computer Science, {LICS}
                  2017, Reykjavik, Iceland, June 20-23, 2017},
  pages        = {1--12},
  publisher    = {{IEEE} Computer Society},
  year         = {2017},
  url          = {https://doi.org/10.1109/LICS.2017.8005068},
  doi          = {10.1109/LICS.2017.8005068},
  bibsource    = {dblp computer science bibliography, https://dblp.org}
}

@inproceedings{BlondinFHH16,
  author       = {Michael Blondin and
                  Alain Finkel and
                  Christoph Haase and
                  Serge Haddad},
  editor       = {Marsha Chechik and
                  Jean{-}Fran{\c{c}}ois Raskin},
  title        = {Approaching the Coverability Problem Continuously},
  booktitle    = {Tools and Algorithms for the Construction and Analysis of Systems
                  - 22nd International Conference, {TACAS} 2016, Held as Part of the
                  European Joint Conferences on Theory and Practice of Software, {ETAPS}
                  2016, Eindhoven, The Netherlands, April 2-8, 2016, Proceedings},
  series       = {Lecture Notes in Computer Science},
  volume       = {9636},
  pages        = {480--496},
  publisher    = {Springer},
  year         = {2016},
  url          = {https://doi.org/10.1007/978-3-662-49674-9\_28},
  doi          = {10.1007/978-3-662-49674-9\_28},
  bibsource    = {dblp computer science bibliography, https://dblp.org}
}

@article{Blondin20,
  author       = {Michael Blondin},
  title        = {The ABCs of petri net reachability relaxations},
  journal      = {{ACM} {SIGLOG} News},
  volume       = {7},
  number       = {3},
  pages        = {29--43},
  year         = {2020},
  url          = {https://doi.org/10.1145/3436980.3436984},
  doi          = {10.1145/3436980.3436984},
  bibsource    = {dblp computer science bibliography, https://dblp.org}
}

@inproceedings{BlondinHO21,
  author       = {Michael Blondin and
                  Christoph Haase and
                  Philip Offtermatt},
  editor       = {Jan Friso Groote and
                  Kim Guldstrand Larsen},
  title        = {Directed Reachability for Infinite-State Systems},
  booktitle    = {Tools and Algorithms for the Construction and Analysis of Systems
                  - 27th International Conference, {TACAS} 2021, Held as Part of the
                  European Joint Conferences on Theory and Practice of Software, {ETAPS}
                  2021, Luxembourg City, Luxembourg, March 27 - April 1, 2021, Proceedings,
                  Part {II}},
  series       = {Lecture Notes in Computer Science},
  volume       = {12652},
  pages        = {3--23},
  publisher    = {Springer},
  year         = {2021},
  url          = {https://doi.org/10.1007/978-3-030-72013-1\_1},
  doi          = {10.1007/978-3-030-72013-1\_1},
  bibsource    = {dblp computer science bibliography, https://dblp.org}
}

@article{FracaH15,
  author       = {Est{\'{\i}}baliz Fraca and
                  Serge Haddad},
  title        = {Complexity Analysis of Continuous Petri Nets},
  journal      = {Fundam. Informaticae},
  volume       = {137},
  number       = {1},
  pages        = {1--28},
  year         = {2015},
  url          = {https://doi.org/10.3233/FI-2015-1168},
  doi          = {10.3233/FI-2015-1168},
  bibsource    = {dblp computer science bibliography, https://dblp.org}
}

@inproceedings{HaaseH14,
  author       = {Christoph Haase and
                  Simon Halfon},
  editor       = {Jo{\"{e}}l Ouaknine and
                  Igor Potapov and
                  James Worrell},
  title        = {Integer Vector Addition Systems with States},
  booktitle    = {Reachability Problems - 8th International Workshop, {RP} 2014, Oxford,
                  UK, September 22-24, 2014. Proceedings},
  series       = {Lecture Notes in Computer Science},
  volume       = {8762},
  pages        = {112--124},
  publisher    = {Springer},
  year         = {2014},
  url          = {https://doi.org/10.1007/978-3-319-11439-2\_9},
  doi          = {10.1007/978-3-319-11439-2\_9},
  bibsource    = {dblp computer science bibliography, https://dblp.org}
}

@inproceedings{EsparzaLMMN14,
  author       = {Javier Esparza and
                  Rusl{\'{a}}n Ledesma{-}Garza and
                  Rupak Majumdar and
                  Philipp J. Meyer and
                  Filip Niksic},
  editor       = {Armin Biere and
                  Roderick Bloem},
  title        = {An SMT-Based Approach to Coverability Analysis},
  booktitle    = {Computer Aided Verification - 26th International Conference, {CAV}
                  2014, Held as Part of the Vienna Summer of Logic, {VSL} 2014, Vienna,
                  Austria, July 18-22, 2014. Proceedings},
  series       = {Lecture Notes in Computer Science},
  volume       = {8559},
  pages        = {603--619},
  publisher    = {Springer},
  year         = {2014},
  url          = {https://doi.org/10.1007/978-3-319-08867-9\_40},
  doi          = {10.1007/978-3-319-08867-9\_40},
  bibsource    = {dblp computer science bibliography, https://dblp.org}
}

\appendix
\section{Proof for Subsection~\ref{subsec:ordinals-to-trees}}
\label{sec:appendix-Minsky-machines}

Assuming Theorems~\ref{thm:Hardy-forward} and~\ref{thm:Hardy-backward}, we show
how to get an $\fastf{\Omega_{k}}$ lower bound for $k$-NRCS. To this end, we first formally define a $H^{\Omega_{k+1}}$-bounded Minsky machine. 

The simplest way to define a Minsky machine in our setting is to define it 
as a tuple $M = (Q,\delta_u,\delta_z)$ where $(Q,\delta_u)$ is a 1-NCS
and $\delta_z \subseteq Q^3$ is a finite set of \emph{zero-test} transitions,
each of which, we shall denote by $p_0 \xdashrightarrow{z, p} q_0$
Configurations of $M$ are defined exactly like configurations of $(Q,\delta_u)$.
In addition to the semantics of steps from $\delta_u$, we also have steps induced
by the zero-test transitions from $\delta_z$ as follows.

Let $t = p_0 \xdashrightarrow{z, p} q_0$ be a zero-test. We say that $t$ can be fired from a configuration $C$ iff in $C$ the root is labelled by $p_0$ and it has no children labelled by $p$. After firing $t$, $C$ can move to $C'$, where the only difference is that the root is now labelled by $q_0$. 

The intuitive idea is that Minsky machines represent machines with the usual counters (where the number of children of the root labelled by some label $p$ represents the value of the counter $p$). The machine can then increment or decrement these counters by means of transitions in $\delta_u$ and test them for zero by means of transitions in $\delta_z$.

The coverability problem for $H^{\Omega_{k+1}}$-bounded Minsky machines is defined as follows:
Given a Minsky machine $M = (Q,\delta_u,\delta_z)$ and two of its states $p, q$,
decide whether the configuration with a singleton node $q_{init}$ can cover the configuration with a singleton node $q_f$, along an $H^{\Omega_{k+1}}(|Q|)$ bounded run. Here a $H^{\Omega_{k+1}}(|Q|)$ bounded run is a run in which all the configurations satisfy the property that the total number of children of the root is at most $H^{\Omega_{k+1}}(|Q|)$.
Intuitively this corresponds to bounding the total value of all the counters along a run
by $H^{\Omega_{k+1}}(|Q|)$ and then checking for coverability along such bounded runs.
This problem is well-known to be $\fastf{\Omega_k}$-hard~\cite{Schmitz16hierarchies}.

Assuming Theorems~\ref{thm:Hardy-forward} and~\ref{thm:Hardy-backward}, we now show
how to reduce the coverability problem for $H^{\Omega_{k+1}}$-bounded Minsky machines
to the coverability problem for $k$-NRCS. To this end, let $M$ be a Minsky machine
with states $Q$ and let $q_{init},q_f \in Q$.  Without loss of generality,
by adding extra decrement and zero-test transitions if necessary, 
we can assume that the machine can cover $q_f$ from $q_{init}$ only when all the counter values are 0, i.e., only when it has no children.

\subsubsection{Convert $M$ into a $1$-NRCS $\net$}

First, we modify the Minsky machine $M$ so that 1) all the zero-test transitions are now reset transitions and 2) any increment/decrement to a label $p$ is accompanied by a corresponding decrement/increment to the label $\#$. The first one is trivial to do. For the second one, we simply replace each transition $(s,p) \xdashrightarrow{u} (s')$
with $(s,p) \xdashrightarrow{u} (s',\#)$ and $s \xdashrightarrow{u} (s',p)$
with $(s,\#) \xdashrightarrow{u} (s',p)$. Let $\net$ be this new $k$-NRCS. 

We say that a step $C \act{t} C'$ in $\net$ involving a reset transition $t = p_0 \xdashrightarrow{r, p} q_0$ is honest if the number of children labelled by $p$
in $C$ is 0. We say a run is honest iff all the reset steps in it are honest.

For any configuration $C$ of $M$ and for any $n \in \mathbb{N}$, we can assign
a unique configuration $(C,n)$ of $\net$ which is the same as $C$, except it also has
$n$ children labelled by $\#$. For any configuration $(C,n)$, let $|C|$ denote
the number of nodes of $C$. The following Fundamental Property of $\net$ is easy
to see from its construction.
\begin{quote}
    If $(C,n) \act{*} (C',n')$ then $|C'|+n' \le |C|+n$. Furthermore,
    $|C|'+n' = |C|+n$ iff the run is honest.
\end{quote}

It is clear that any run in $M$ between two configurations $C$ and $C'$ where
all the counter values are bounded by $B$ corresponds to a honest run between $(C,B)$ and $(C',B)$ in $\net$ and vice versa. We will now use this property along with $\net^{fwd}$ and $\net^{bwd}$ as follows.

\subsubsection{Chaining $\net^{fwd}, \net$ and $\net^{bwd}$}

Set $\ell = |Q|, \alpha = (\Omega_{k+1})_{\ell}$. Note that $H^{\Omega_{k+1}}(|Q|) = H^{\alpha}(|Q|)$. Now use Theorems~\ref{thm:Hardy-forward}
and~\ref{thm:Hardy-backward} to construct $\net^{fwd}$ and $\net^{bwd}$ respectively.
We now construct a $k$-NRCS that starts at $C_{\alpha,|Q|}$ and performs the following steps.

It initially runs the machine $\net^{fwd}$,
and then at some point,
when the root becomes labelled by $\omega$, it non-deterministically decides to 
start simulating $\net$ by means of the following transition.
$$(\omega) \xdashrightarrow{u} (q_{init})$$
    
By Theorem~\ref{thm:Hardy-forward}, right before firing the above transition,
the machine must have reached some configuration of the form $C_{\alpha_1,n_1}$
with $n_1 \le H^{\alpha}(|Q|)$, i.e., the number of children
with label $\#$ that the machine has at this point is some $n_1$ 
such that $H^{\alpha_1}(n_1) \le H^{\alpha}(|Q|)$.

After firing the above transition, it starts running the machine $\net$. Note
that the transitions of $\net$ will not disturb the subtrees of the root corresponding
to the terms in the CNF of $\alpha_1$. When the simulation of $\net$ makes
the root reach the state $q_f$, the machine decides to start simulating $\net^{bwd}$ by means
of the following transition.
$$(q_f) \xdashrightarrow{u} (\omega)$$

By the Fundamental Property of $\net$, it follows that the number of children 
with label $\#$ that the machine can have at this point is at most $n_2 \le n_1$. In addition, $n_2 = n_1$ iff this simulation of $\net$ was a honest run of $\net$.
At this point, the configuration of this machine must look like $C_{\alpha_1,n_2}$.
(This is because, we had assumed that any transition leading to $q_f$ in $M$ 
must have zero-tested all of its counters, i.e., must have zero-tested all of the labels
of its children. Hence $\net$ must have reset all of the labels present in $M$, which would leave us only with the subtrees corresponding to terms of the CNF of $\alpha_1$ and $\#$'s).

Now the machine simulates $\net^{bwd}$ from $C_{\alpha_1,n_2}$ and then at some point,
it decides to stop. Suppose $C_{\alpha_2,n_3}$ is the configuration reached at the end of the simulation of $\net^{bwd}$. By Theorem~\ref{thm:Hardy-backward} we have that $H^{\alpha_2}(n_3) \le H^{\alpha_1}(n_2) \le H^{\alpha_1}(n_1) \le H^{\alpha}(|Q|)$.
We now claim that 
\begin{quote}
    $C_{\alpha_2,n_3}$ covers $C_{\alpha,|Q|}$ implies that $H^{\alpha_2}(n_3) = H^{\alpha}(|Q|)$.
\end{quote}

If this claim were true, then $n_2$ must be equal to $n_1$ and so the simulation of $\net$ would have lead to an honest run, which would lead to a run in the original machine $M$ covering the final state $q_f$.
This would then prove that if $C_{\alpha,|Q|}$ can be covered in this machine,
then $q_f$ can be covered in $M$. Since the converse is also easily seen to be true by construction (for instance, by first simulating $\net^{fwd}$ correctly, then by simulating $\net$ honestly and then by simulating $\net^{bwd}$ correctly), the required reduction would then follow. 

Hence all that is left is to prove the above claim. To prove this, we just need
to note by construction that $C_{\alpha_2,n_3}$ covers $C_{\alpha,|Q|}$ iff
$T_{\alpha_2} \ge_{is} T_\alpha$ and $n_3 \ge |Q|$.
We now use Fact 5.2 and Proposition 5.4 of~\cite{PriorityChannelSystems}, to get\footnote{Proposition 5.4 from~\cite{PriorityChannelSystems} can only be applied when $\alpha$ is less than $\alpha'$ by the so-called structural ordering in~\cite{PriorityChannelSystems}, but this ordering is the same as the induced subgraph ordering between $T_{\alpha}$ and $T_{\alpha'}$} : \begin{proposition}\label{prop:induced-monotone-Hardy}
    For any $\alpha, \alpha', n, n'$, $T_{\alpha'} \ge_{is} T_\alpha$ and $n' \ge n$
    implies that $H^{\alpha'}(n') \ge H^{\alpha}(n)$
\end{proposition}

Applying this proposition to $\alpha_2, \alpha, |Q|$ and $n_3$ and using the fact that $H^{\alpha_2}(n_3) \le H^{\alpha}(|Q|)$ we get that $H^{\alpha_2}(n_3) = H^{\alpha}(|Q|)$,
which proves the claim.

\section{General Remark on Notation}\label{sec:notation-remark}

Before we proceed to the Hardy computations and the constructions of the gadgets, we make a general remark as follows. Throughout the description of the construction of these gadgets,
we will be working with various encoders. Recall that an encoder for an ordinal $\alpha$ 
is a tree that looks exactly like $T_\alpha$, except that some of the $\omega$ labels
could be renamed. Recall that the labels in the last level in $T_\alpha$ could be of the form $a \in \omega^0,\dots,\omega^\ell$, which represent that in $T'_\alpha$ they had $0,\dots,\ell$ nodes as children. We are only allowed to rename them as some label of the form $(n,a)$
where $n$ is some new label, i.e., new names for such nodes must retain the label in $T_\alpha$ somewhere in their notation. This will pose some notational challenge in the gadgets and the constructions that we describe below, because every time we introduce a transition, we would then have to go into a case analysis on whether this transition hits the last level or not. This would make it extremely cumbersome to state the constructions.

For this reason, whenever we express transitions working over encoders, if the transitions reference the last level, we will just treat is a single label $n$ and not as a pair $(n,a)$.
This is done purely for the sake of convenience, as the behavior of the transitions that we construct is (mostly) oblivious to the exact $i \in \{0,\dots,\ell\}$ that is present in the label $\omega^i$ in the original tree $T_\alpha$ this last level. 
Indeed, for the last level $k$, suppose any one of our transitions is presented as 
$(p_0,\dots,p_k) \xdashrightarrow{u} (q_0,\dots,q_k)$.
This must be interpreted as follows
\begin{itemize}
    \item If $p_k = q_k = \omega$
    then this actually represents 
    $\ell+1$ many transitions, $(p_0,\dots,\omega^i) \xdashrightarrow{u} (q_0,\dots,\omega^i)$
    for each $0 \le i \le \ell$.
    \item If $p_k = \omega$ and $q_k \neq \omega$, then this actually represents
    $\ell+1$ many transitions
    $(p_0,\dots,\omega^i) \xdashrightarrow{u} (q_0,\dots,(q_k,{\omega^i})))$ for each $0 \le i \le \ell$,
    \item If $p_k \neq \omega$ and $q_k = \omega$, then this actually represents
    $\ell+1$ many transitions
    $(p_0,\dots,(p_k,{\omega^i}))) \xdashrightarrow{u} (q_0,\dots,\omega^i)$ for each $0 \le i \le \ell$.
    \item If $p_k \neq \omega$ and $q_k \neq \omega$, then this actually represents $\ell+1$ many transitions
    $(p_0,\dots,(p_k,{\omega^i}))) \xdashrightarrow{u} (q_0,\dots,(q_k,{\omega^i}))$ for each $0 \le i \le \ell$.
\end{itemize}

Similarly, if we have a transition $(p_0,\dots,p_k) \xdashrightarrow{u} (q_0,\dots,q_j)$
with $j < k$, then for each $0 \le i \le \ell$, this must be interpreted as  $(p_0,\dots,\omega^i) \xdashrightarrow{u} (q_0,\dots,q_j)$  if $p_k = \omega^i$
and as $(p_0,\dots,(p_k,{\omega^i})) \xdashrightarrow{u} (q_0,\dots,q_j)$ if $p_k \neq \omega^i$. The case of $(p_0,\dots,p_j) \xdashrightarrow{u} (q_0,\dots,q_k)$
with $j < k$ must be interpreted as $(p_0,\dots,p_j) \xdashrightarrow{u} (q_0,\dots,\omega^0)$. 

Finally, if we say that we want to add a child to a node at the $k^{th}$ level with label $\omega$, then for each $0 \le i \le \ell$, this must be interpreted as changing the labels
of $\omega^i$ and $(p,{\omega^i})$ (for any $p$) to $\omega^{i+1}$ and $(p,{\omega^{i+1}})$
respectively if $i+1 < \ell$. Similarly in the case of decrements it must be moved to $i-1$ and in the case of resets, it must be moved to 0. 

With this notation stated, we now move on to describing our constructions.

\section{Proofs for Subsection~\ref{subsec:compute-Hardy}: Proof of Theorem~\ref{thm:Hardy-forward}}\label{sec:appendix-Hardy-forward}

In this section, we shall show that using the smallest-child gadgets and the copy gadgets
we can construct $\net^{fwd}$. We have already defined
small-initialized and lossy smallest-child encoders and smallest-child gadgets in Subsection~\ref{subsec:compute-Hardy}. We have also defined marked and lossy copy encoders in 
Subsection~\ref{subsec:compute-Hardy}. The only thing left to define 
is the notion of an $k$-copy gadget, which we define as follows.

\begin{definition}[Copy Gadgets]
    A $k$-copy gadget is a $k$-NRCS such that for any $\alpha \le (\Omega_{k+1})_{\ell}$ and any marked encoder $T$ for $\alpha$:
    \begin{itemize}
        \item There exists a run from $T$ to the perfect copy encoder of $T$.
        \item Suppose there exists a run from $T$ to some $T'$ such that the root of $T'$ is labeled by $\ecopy$. Then $T'$ is a lossy copy encoder of $T$.
    \end{itemize}
\end{definition}

Now, assuming that we can construct $i$-smallest-child and $i$-copy gadgets for any $i \le k$,
we shall now show how to prove the following.

\Hardyforward*

The construction of $\net^{fwd}$ is as follows: It initially guesses whether the 
current tree $C_{\alpha,n}$ is such that $\alpha$ represents either a successor ordinal $\beta+1$
or a limit ordinal $\lambda$. For the former case, we need to move to $C_{\beta,n+1}$.
To do this, we simply need to remove a child of the root marked by $\omega$ and convert it into $\#$. This can be done by
$$(\omega,\omega) \xdashrightarrow{u} ({succ})$$
$$({succ}) \xdashrightarrow{u} (\omega,\#)$$

If the child marked by $\omega$ that we removed did not have any children, then
we have indeed moved from $C_{\beta+1,n}$ to $C_{\beta,n+1}$ by these transitions; otherwise
we have committed a lossy error and moved to some $C_{\beta',n+1}$ such that
$T_{\beta'+1} < T_{\alpha}$ and so by Proposition~\ref{prop:induced-monotone-Hardy} we get
that $H^{\beta'}(n+1) = H^{\beta'+1}(n) \le H^{\alpha}(n)$. 

Suppose we guess that our current configuration is $C_{\lambda,n}$ for some limit ordinal $n$. We now need to move to $C_{\lambda_n,n}$. Recall that if $\lambda = \gamma + \omega^{\eta}$ where $\eta = \beta+1$ or $\eta = \lambda'$, then $\lambda_n$ is given by the following transitions

$$(\gamma + \omega^{\beta+1})_n = \gamma + \omega^{\beta} \cdot n, \qquad  (\gamma + \omega^{\lambda'})_n = \gamma + \omega^{\lambda'_n}$$

Hence, we need to find the smallest term in the CNF of $\lambda$ and if it represents an exponent with a successor ordinal, we need to reduce it by one and copy it $n-1$ times;
otherwise, we just need to go deeper on the smallest term. We now implement these steps
as follows: First, we change the label of the root
to mark that we are in the limit ordinal case
$$(\omega) \xdashrightarrow{u} (\bsmall)$$

We will now proceed recursively in at most $k$ phases. For each $0 \le i < k$,
we assume at the beginning of the $i^{th}$ phase that we have marked a unique path $v_0,\dots,v_{i}$ in the configuration by the labels $(\Phi_i,\bsmall)$ where $\Phi_i := \underbrace{\fwd{i},\cdots,\fwd{i}}_{i}$. Now, we proceed
to invoke the $(k-i)$-smallest-child gadget on $v_{i}$. Technically this means
that we take every update transition $(p_0,\dots,p_j) \xdashrightarrow{u} (q_0,\dots,q_{\ell})$ of this gadget and replace it with
$$(\Phi_i,p_0,\dots,p_j) \xdashrightarrow{u} (\Phi_i,q_0,\dots,q_\ell)$$
Similarly, we can prepend every reset transition of this gadget. 
This ensures that we are applying the $k-i$-smallest-child gadget on the node $v_{i}$.
We then wait for the gadget computation to conclude, which will happen when $v_i$
gets the label $\esmall$. At this point, by the property of the smallest-child gadget, 
we would have marked the unique child of $v_i$ (say $v_{i+1}$) by $smallest$, which represents the smallest ordinal term in the CNF of the ordinal of $v_i$. 
The ordinal represented
by $v_{i+1}$ could either be a limit ordinal or a successor ordinal. For the former,
we just need to go to the next phase on $v_{i+1}$.

$$(\Phi_i,\esmall,smallest) \xdashrightarrow{u} (\Phi_{i+1},\bsmall)$$

Otherwise, $v_{i+1}$ represent an ordinal of the form $\beta+1$. We need to subtract
a child labelled by $\omega$ from it and then copy it $n+1$ times. To this end, let $\xi_{i} =  \underbrace{\cp{i},\cdots,\cp{i}}_{i}$. First, we subtract a child labelled by $\omega$ from $v_{i+1}$, whilst simultaneously moving the other labels to $\xi_i$ to show that we are entering the copy phase.
$$(\Phi_i,\esmall,smallest,\omega) \xdashrightarrow{u} (\xi_{i},\bcopy, \marked)$$

[If the node labelled by $\omega$ which we removed did not have any children, then we have successfully gone from $\beta+1$ to $\beta$; on the other hand, if it did have children, then similar to the argument given before, the tree that we get
is a subtree of the original tree that we began with and so this can only make the Hardy computation go down.]

From here, we invoke the $(k-i)$-copy gadget on the node $v_{i}$ for $n-1$ times. 
Technically this means that we take every update transition $(p_0,\dots,p_j) \xdashrightarrow{u} (q_0,\dots,q_\ell)$ of this gadget and replace it with
$$(\xi_{i},p_0,\dots, p_j) \xdashrightarrow{u} (\xi_{i},q_0,\dots,q_\ell)$$
Similarly, we can prepend every reset transition of this gadget.
This ensures that we are applying the copy gadget on the node $v_{i}$, which allows us to copy the subtree of $v_{i+1}$ once as a subtree of $v_i$.

We then wait for the gadget computation to conclude, which will happen when $v_i$ gets
the label $\ecopy$. At this point, the label of the node $v_{i+1}$ would have become $\omega$
and the label of the copy of $v_{i+1}$ (say $v_{i+1}^1$) would have become $\copied$.
We turn $\copied$ back into $\marked$, convert an $\#$ to $\#'$ (to indicate that we have copied once) and start the copy process once again.
$$(\xi_i,\ecopy,\copied) \xdashrightarrow{u} (\convert,\xi_{i-1},\bcopy,\marked)$$
and do 
$$(\convert,\#) \xdashrightarrow{u} (\convert',\#')$$

From $\convert'$, we can guess that we need to proceed further to do another copy
or we have done copying for the required $n-1$ times. In the former case,
we just need to change $\convert'$ to $\cp{i}$, which would once again give us a path
$(\xi_i,\bcopy,\marked)$ from which we can apply the copy gadget again.
$$(\convert') \xdashrightarrow{u} (\cp{i})$$

In the latter case, we first convert a $\#$ to $\#'$ (to check that we have not overshot and done the copy $n$ times):
$$(\convert',\#) \xdashrightarrow{u} (\convert'',\#')$$

Then we start transferring all the $\#'$s to $\#$
$$(\convert'',\#') \xdashrightarrow{u} (\convert'', \#)$$

At some point, we guess that we have transferred all the $\#'$ back to $\#$ and so we 
reset all the $\#'$ and conclude the gadget.

$$(\convert'') \xdashrightarrow{r, \#'} (\last{i})$$
$$(\last{i},\xi_{i-1},\bcopy,\marked) \xdashrightarrow{u}
(\underbrace{\omega,\cdots,\omega}_{i+2})$$

\section{Proofs for Subsection~\ref{subsec:compute-Hardy}: Proof of Theorem~\ref{thm:Hardy-backward}}\label{sec:appendix-Hardy-backward}

In this section, we shall show how to construct $\net^{bwd}$. To construct this, we will need one more type of gadget in addition to the smallest-child gadget. This is the \emph{comparator gadget} which lets us compare two different ordinals. (This is the same comparator gadget that we had mentioned in Subsection~\ref{subsec:smallest-child}).
Similar to the other gadgets, we first define the encoders on which they will act and then define the gadget formally. 

The first encoder we need is called a \emph{comparator encoder}, which will intuitively allow us to mark two subtrees (corresponding to two different ordinal terms in the CNF, which we will later compare against each other).

 \begin{definition}
     A \emph{comparator encoder} for an ordinal $\alpha$ is an encoder for $\alpha$ whose root is labeled by $\bcmp$ and exactly two of its children are labeled by $\Acmp$ and $\Bcmp$, respectively.
 \end{definition}

The children marked by $\Acmp$ and $\Bcmp$ correspond to the terms in the CNF of $\alpha$ which we want to compare. Ultimately, we want gadgets that will start from a comparator encoder for $\alpha$ and decide which one of the two ordinals corresponding to the two marked subtrees is bigger. To this end, we define a \emph{lossy compared encoder} as follows.

\begin{definition}
    Given a comparator encoder $T$ for $\alpha$, a \emph{lossy compared encoder}
    for $T$ is a tree $T'$ that is exactly like $T$ except for the following changes:
    \begin{itemize}
        \item Let $v_A, v_B$ be the nodes in $T$ labeled by $\Acmp$ and $\Bcmp$, respectively. Then $T'$ contains all the nodes of $T$, except possibly for some nodes in the subtrees rooted at $v_A$ and $v_B$.
        \item Every node in $T'$ has the same label as its corresponding node in $T$,
        except that the root is labeled by $\ecmp$, and $v_A, v_B$ are labeled as follows.
        \begin{itemize}
        \item If the ordinals represented by the subtrees rooted at $v_A$ and $v_B$ in $T'$ are the same, then they are both labeled by $\equalcmp$.
        \item If the ordinal of one subtree is bigger than the ordinal of the other subtree, then the node with the bigger subtree is labeled by $\bigcmp$ and the other one by $\smallcmp$.
        \end{itemize}
    \end{itemize}
    A lossy compared encoder is called a \emph{perfect compared encoder} iff 
    $T'$ contains all the nodes of $T$.
\end{definition}

We can now define the notion of a comparator gadget as follows.

\begin{definition}[Comparator Gadget]
    A $k$-comparator gadget is a $k$-NRCS such that for any $\alpha < (\Omega_{k+1})_{\ell}$ and any comparator encoder $T$ for $\alpha$, the following holds:
    \begin{itemize}
        \item There exists a run from $T$ to the perfect compared encoder of $T$.
        \item Suppose there exists a run from $T$ to some $T'$ such that the root of $T'$ is labeled by $\ecmp$. Then $T'$ is a lossy compared encoder of $T$.
    \end{itemize}
\end{definition}

Note that, similar to smallest-child gadgets, here we allow the comparator gadget to make lossy errors. Because of these lossy errors, we are only guaranteed that the ordinals represented by the nodes $v_A$ and $v_B$ in $T'$ (not in $T$!) are correctly marked as big or small. 

Now, assuming that we can construct $i$-smallest-child and $i$-comparator gadgets for any $i \le k$, we shall now show how to prove the following.

\Hardybackward*

Similar to the forward case, the backward case initially guesses whether the current tree is of the form $C_{\alpha,n+1}$ or of the form $C_{\lambda_n,n}$. In the former case,
we need to move to $C_{\alpha+1,n}$ and this can be easily done as follows:
$$(\omega,\#) \xdashrightarrow{u} (\omega,\omega)$$

In the latter case, we need to move to $C_{\lambda,n}$. For this we need to convert
$\lambda_n$ into $\lambda$. Note that $\lambda$ could either be of the form  $\gamma + \omega^{\lambda'}$ or of the form $\gamma+ \omega^{\beta+1}$. Hence, $\lambda_n$
is either of the form $\gamma + \omega^{\lambda'_n}$ or of the form $\gamma + \omega^{\beta} \cdot n$. In the first case, we need to identify the subtree corresponding
to the smallest term $\lambda'_n$ and recursively convert $\lambda'_n$ into $\lambda$.
In the second case, we need to identify $(n-1)$  other copies of the subtree corresponding
to the smallest term $\beta$, remove all of them and convert the remaining copy into $\beta+1$. To identify $n-1$ other copies of the same ordinal, we use the comparator gadget $n-1$ times
and remove one copy for each invocation of the comparator gadget.
We now implement these steps as follows, in a way that is symmetric to the forward case.

First, we change the label of the root to mark that we are in the $\lambda_n$ case:
$$(\omega) \xdashrightarrow{u} (\bsmall)$$

We will now proceed recursively in at most $k$ phases. For each $0 \le i < k$, we assume at the beginning of the $i^{th}$ phase that we have marked a unique path $v_0,\cdots,v_i$
in the configuration by the labels $(\Phi_i,\bsmall)$ where $\Phi_i := \underbrace{\bwd{i},\cdots,\bwd{i}}_{i}$. Now, we proceed to invoke the $(k-i)$-smallest-child gadget on $v_i$.
This is similar to the forward case, where we take that gadget and prepend $\Phi_i$ to each transition. We then wait for the gadget computation to end, which happens when $v_i$
gets the label $\esmall$. At this point, the smallest-child gadget would have marked the unique child of $v_i$ (say $v_{i+1}$) by $\smallest$, which represents the smallest ordinal term in the CNF of the ordinal of $v_i$. This smallest ordinal could either be of the form
$\lambda'_n$ or of the form $\beta$ with $n-1$ other copies of it present.
If we guess the former case, we just need to go to the next phase on $v_{i+1}$.

$$(\Phi_i,\esmall,\smallest) \xdashrightarrow{u} (\Phi_{i+1},\bsmall)$$

In the latter case, $v_{i+1}$ must represent an ordinal $\beta$ with $n-1$ other copies of it present. Therefore, we start identifying $n-1$ other subtrees whose represented ordinal is the same as $v_{i+1}$ and for this we have to use the comparator gadget. 
To this end, let $\Phi_i' = \underbrace{\bwd{i}',\cdots,\bwd{i}'}_{i}$ and $\xi_i = \underbrace{\cm{i},\cdots,\cm{i}}_i$. First we prepare the nodes $v_i$ and $v_{i+1}$ for
the comparator gadget.
$$(\Phi_i,\esmall,\smallest) \xdashrightarrow{u} (\Phi_i',\bcmp,\Acmp)$$
Then, we non-deterministically select a sibiling $v_{i+1}'$ of $v_{i+1}$ to compare it with
$$(\Phi_i',\bcmp,\omega) \xdashrightarrow{u} (\xi_i,\bcmp,\Bcmp)$$

From here, we invoke the $(k-i)$-comparator gadget on the node $v_i$ with the two marked
nodes $v_{i+1}$ and $v_{i+1}'$. This we can once again do by prepending $\xi_i$ to every transition of that gadget. The computation of the gadget ends when $v_i$ reaches the
state $\ecmp$. At this point, we want the labels of both $v_{i+1}$ and $v_{i+1}'$ to be $\equalcmp$ (in fact, by the properties of the comparator gadget, just checking for one of them will do). We then prepare the comparator gadget once again, convert an $\#$ to $\#'$
(to indicate that we have found a copy of the ordinal we were searching for) and start
the process once again.

$$(\xi_i,\ecmp,\equalcmp) \xdashrightarrow{u} (\convert,\xi_{i-1},\bcmp,\Acmp)$$
$$(\convert,\#) \xdashrightarrow{u} (\convert',\#')$$

From $\convert'$, just like in the forward case, we guess either to proceed further
to do another comparison or we guess that we have found all the other required $n-1$ copies
of $v_{i+1}$. In the former case, we just need to mark another sibiling of $v_{i+1}$ for the comparator gadget, from which we can call it again.

$$(\convert',\xi_{i-1},\bcmp,\omega) \xdashrightarrow{u} (\xi_i,\bcmp,\Bcmp)$$

In the latter case, we start transferring all the $\#'$'s to $\#$, exactly like
in the forward case. (We also convert one $\#$ to $\#$'s to ensure that we have not overshot and invoked the comparator gadget $n$ times).
$$(\convert',\#) \xdashrightarrow{u} (\convert'',\#')$$
$$(\convert'',\#') \xdashrightarrow{u} (\convert'', \#)$$
$$(\convert'') \xdashrightarrow{r, \#'} (\last{i})$$

Finally, we have to increase the ordinal $\beta$ to $\beta+1$ and so we have to
add one child labelled by $\omega$ to the node marked with $\Acmp$ and conclude the computation. 
$$(\last{i},\xi_{i-1},\bcmp,\Acmp) \xdashrightarrow{u}
(\underbrace{\omega,\cdots,\omega}_{i+3})$$

\section{Proofs for Subsection~\ref{subsec:smallest-child}}
\label{sec:appendix-gadgets}

In this subsection, we will present the construction of the copy gadgets, smallest-child gadgets, comparator gadgets and the biggest-child gadgets. All of them have been defined either in the main part of the paper or in the preceding sections of the appendix, except for the biggest-child gadget, which we now define as follows. First, we
define a biggest-child encoder, similar to a smallest-child encoder.

A big-initialized encoder is like a small-initialized encoder except
that the root is labelled by $\bbig$. Now, we define a lossy biggest-child encoder.

\begin{definition}
    Given an initialized encoder $T$ for $\alpha$, a lossy biggest-child encoder for $T$
    is a tree $T'$ that looks exactly like $T$, except for the following changes:
    \begin{itemize}
        \item Let $v_1,\dots,v_n$ be the children of the root labeled by $\omega$ in $T$. $T'$ contains all the nodes of $T$, except possibly for some nodes in the subtree of any of the $v_i$.
        \item Every node in $T'$ has the same label as its corresponding node in $T$,
        except that the root is labeled by $\ebig$, and each $v_i$ (if it exists in $T'$) is labeled as follows:
        \begin{itemize}
        \item If the ordinal represented by the subtree rooted at $v_i$ in $T'$ is the biggest among all the ordinals of any of the subtrees of all the other $v_j$ in $T'$, then $v_i$ is labeled by $\biggest$.
        \item Otherwise, $v_i$ has the same label as it had in $T$.
        \end{itemize}
    \end{itemize}
    A lossy biggest-child encoder is called a \emph{perfect biggest-child encoder} if 
    $T'$ contains all the nodes of $T$.    
\end{definition}

Now, we can define a biggest-child gadget as follows.

\begin{definition}[Biggest-child Gadgets]
    A $k$-biggest-child gadget is a $k$-NRCS such that for any $\alpha < (\Omega_{k+1})_{\ell}$ and any big-initialized encoder $T$ for $\alpha$:
    \begin{itemize}
        \item There exists a run from $T$ to the perfect biggest-child encoder of $T$.
        \item Suppose there exists a run from $T$ to some $T'$ such that the root of $T'$ is labeled by $\esmall$. Then $T'$ is a lossy biggest-child encoder of $T$.
    \end{itemize}
\end{definition}

Having defined the biggest-child gadgets, we now  show how to construct all the types of gadgets defined. We begin with the copy gadgets.

\subsection{Proof of Theorem~\ref{thm:copy-gadgets}: Constructing Copy Gadgets}
\label{subsec:appendix-copy-gadgets}

Recall the intuition that we had presented for the copy gadgets - we simply do a DFS
over the subtree that we want to mark and whenever we hit a maximum branch we simply copy it onto another copy. We now describe this formally.

A copy gadget starts from a marked encoder and so this means that the root $u$ is labelled by $\bcopy$ and one of its children $v_1$ is labelled by $\marked$. As a first step, 
we change the label of $v_1$ to $\marked_1$ as well as create a new child $v_1'$ labelled
by $\marked'_1$.

$$(\bcopy,\marked) \xdashrightarrow{u} (\bcopy', \marked_1)$$
$$(\bcopy') \xdashrightarrow{u} (\cp{1},\marked'_1)$$

Now, assume that we have two paths from the root, one as $u,v_1,\dots,v_i$ labelled by $\cp{i},\Phi_i$ where $\Phi_i := \underbrace{\marked_i,\cdots,\marked_i}_i$ and the other as $u,v_1',\dots,v_i'$ labelled by 
$\cp{i},\Phi_i'$ where $\Phi_i'$ is defined as $:=\underbrace{\marked'_i,\cdots,\marked'_i}_i$. The first one is the current DFS branch along which we are exploring and the second one is where we have to copy the nodes from the first one. To this end, we either guess now to extend the branch $v_i$ to $v_{i+1}$
or guess that there are no other children of $v_i$ left to explore. In the former
case we just have to do 
$$(\cp{i},\Phi_i,\omega) \xdashrightarrow{u} (\cp{i+1}',\Phi_{i+1})$$
$$(\cp{i+1}',\Phi_i') \xdashrightarrow{u} (\cp{i+1},\Phi_{i+1}')$$

For the latter case, we simply need to reset all children of $v_i$ labelled by $\omega$,
mark $v_i$ as $\omega'$ (hence ensuring that we will not explore any branch with $v_i$ ever again), mark $v_i'$ as $\omega$ (hence ensuring that $v_i$ has been copied safely in the other subtree) and move up the branch for $i > 1$.
$$(\cp{i},\Phi_i) \xdashrightarrow{r, \omega} (\convert_i,\Phi_{i-1},\omega')$$
$$(\convert_i,\Phi_i') \xdashrightarrow{u} (\cp{i-1},\Phi_{i-1}',\omega)$$

On the other hand for $i = 1$, moving up the branch simply means moving up from $v_1$, which means the DFS is now complete. So in this case we mark $v_1$ as $\omega'$ and $v_1'$
as $\copied$.
$$(\cp{1},\marked_1) \xdashrightarrow{r, \omega} (\convert_1,\omega')$$
$$(\convert_i,\marked_1') \xdashrightarrow{u} (\cp{0},\copied)$$

The construction is almost over, since the subtrees of $v_1$ and $v_1'$
are the same except now all the nodes in the subtree of $v_1$ 
are labelled by $\omega'$ and we want to convert them to $\omega$.
This we do (lossily) once again by a DFS on the subtree of $v_1$ 
and convert every $\omega'$ to $\omega$, in exactly the same way as the construction from above.

\begin{remark}
    As a final remark, as we had already mentioned in Section~\ref{sec:notation-remark}, when dealing with nodes in the last level, we actually have to not check for nodes labelled by $\omega$, but rather $\omega^0,\dots,\omega^\ell$ and keep track of them when converting them to $\marked_i, \marked_i'$ etc.
\end{remark}

\subsection{Constructing $k$-Smallest-Child and $k$-Biggest-Child Gadgets}

As mentioned in Subsection~\ref{subsec:smallest-child}, it is quite easy to construct
a 1-comparator gadget. In this subsection, we will show how to construct $k$-smallest-child and $k$-biggest-child gadgets, assuming that a $k$-comparator gadget exists.

\subsubsection{Constructing $k$-Smallest-Child Gadget}

\ksmallest*

The idea is to iterate over pairs of all subtrees corresponding to terms in the CNF of the ordinal, compare them and continue with the one that is smaller. This is roughly the same idea that we used in the backward Hardy computation to identify $n-1$ copies of the same ordinal. We now make this idea concrete.

A $k$-smallest-child gadget starts from a small-initialized encoder where the root is labelled by $\bsmall$.
We ultimately want to identify the smallest ordinal represented by all the children
labelled by $\omega$. As a first step, we mark some child $v_A$ to be compared
$$(\bsmall,\omega) \xdashrightarrow{u} (\bsmall',\Acmp)$$
Then we mark another child $v_B$ to be compared and run the $i$-comparator gadget
$$(\bsmall,\omega) \xdashrightarrow{u} (\bcmp,\Bcmp)$$

At this point, the root is labelled by $\bcmp$ and two of its children $v_A, v_B$ are labelled by
$\Acmp$ and $\Bcmp$. We simply run the $k$-comparator gadget which ends
when the root becomes $\ecmp$. At that point, either both $v_A$ and $v_B$ are labelled by 
$\equalcmp$ or exactly one of them is labelled by $\smallcmp$ and the other by $\bigcmp$.
In the former case, we simply mark one of them (say $v_B$) as $\omega'$ and the other one (say $v_A$) as $\Acmp$, which is the current smallest candidate  so that it can be compared to some other node next.

$$(\ecmp,\equalcmp) \xdashrightarrow{u} (\ecmp',\omega')$$
$$(\ecmp',\equalcmp) \xdashrightarrow{u} (\ecmp'',\Acmp)$$

In the latter case, we mark the one labelled by $\bigcmp$ (say $v_B$) as $\omega'$ 
and the other one (say $v_A$) as $\Acmp$, so that it can be compared to some other node next.

$$(\ecmp,\bigcmp) \xdashrightarrow{u} (\ecmp',\omega')$$
$$(\ecmp',\smallcmp) \xdashrightarrow{u} (\ecmp'',\Acmp)$$

At this point, in both the cases, we now guess whether there are any other nodes to be compared to $\Acmp$
or not. In the former case, we now find another node to mark as $\Bcmp$ and run the comparator gadget again.

$$(\ecmp'',\omega) \xdashrightarrow{u} (\bcmp,\Bcmp)$$

Otherwise, we now reset all nodes labelled by $\omega$, and declare the node labelled by $\Acmp$ as the smallest

$$(\ecmp'') \xdashrightarrow{r, \omega} (\textit{declare})$$
$$(\textit{declare},\Acmp) \xdashrightarrow{u} (\convert,\smallest)$$

We now convert all the nodes labelled by $\omega'$ back to $\omega$ lossily and once that is done, we set the root as $\esmall$,  declaring the end of the $i$-smallest gadget computation.
$$(\convert,\omega') \xdashrightarrow{u} (\convert,\omega)$$
$$(\convert) \xdashrightarrow{r, \omega'} (\convert')$$
$$(\convert') \xdashrightarrow{u} (\esmall)$$

\subsubsection{Constructing $k$-Biggest-Child Gadget}

\kbiggest*

The idea is to pick one subtree, copy it once, and check if the ordinal of the original is bigger than the ordinals corresponding to the other subtrees. 
We now make this idea concrete.

A $k$-biggest-child gadget starts from a big-initialized encoder where the root is labelled by $\bbig$. We ultimately want to identify the biggest ordinal represented by all the children
labelled by $\omega$. As a first step, we mark a candidate node $v$ non-deterministically
as the candidate for the biggest child and copy the subtree of this node.
$$(\bbig,\omega) \xdashrightarrow{u} (\bcopy,\marked)$$

Now, we run a modified version of the $k$-copy gadget, which will work exactly as the original $k$-copy gadget, except instead of marking $v$ as $\omega$ in the end, it will mark it 
as $\omega'$. (It is straightforward to modify the construction given in Subsection~\ref{subsec:appendix-copy-gadgets} so that this holds).
We then wait till the root has received the label $\ecopy$.
At this point, $v$ will be marked by $\omega'$ and there will be another new node $v'$ labelled by $\copied$. 

We now compare the subtree rooted at $v$ (the only node marked by $\omega'$) 
with all the subtrees of all the nodes marked by $\omega$ and check that $v$ remains bigger each time. To this end, we first mark $v$ as $\Acmp$ and pick some other node labelled by $\omega$ to be marked as $\Bcmp$.

$$(\ecopy,\omega') \xdashrightarrow{u} (\ecopy',\Acmp)$$
$$(\ecopy',\omega) \xdashrightarrow{u} (\bcmp,\Bcmp)$$

We now run a modified version of the $k$-comparator gadget, which will work exactly as the 
original $k$-comparator gadget, except that it will keep track of the label $\Acmp$ in the node $v$ throughout. More precisely, we modify every transition $(p_0,p_1,\dots,p_i) \xdashrightarrow{u} (q_0,q_1,\dots,q_j)$ of the $i$-comparator gadget in the following way:
If $p_1 = \Acmp$, then we modify it so that it now becomes
$$(p_0,p_1,\dots,p_i) \xdashrightarrow{u} (q_0,(\Acmp,q_1),\dots,q_j)$$
If $p_1 \neq \Acmp$, then we keep the transition as well as add the following transition:
$$(p_0,(\Acmp,p_1),\dots,p_i) \xdashrightarrow{u} (q_0,(\Acmp,q_1),\dots,q_j)$$

We do a similar change for reset transitions.
This ensures that the label $\Acmp$ is preserved in $v$ throughout in the $k$-comparator gadget. We then wait for the gadget to end, which will happen when the root is labelled
by $\ecmp$. At this point, $v$ will be labelled $(\Acmp,x)$ and the other compared node
will be labelled by $y \in \{\bigcmp,\equalcmp,\smallcmp\}$. If $x \notin \{\bigcmp,\equalcmp\}$, then this means that the node $v$ is not the biggest-child gadget and so our original guess was incorrect.
We then simply deadlock and do not move at all. If $x$ does equal any one of these two choices, we relabel it back to $\Acmp$ and relabel $y$ to $\omega'$.

$$(\ecmp,(\Acmp,x)) \xdashrightarrow{u} (\ecmp',\Acmp)$$
$$(\ecmp',y) \xdashrightarrow{u} (\ecmp'',\omega')$$

We can now choose to compare $\Acmp$ with another node labelled by $\omega$ or decide
that we have compared it with everything.
In the former case, we do
$$(\ecmp'',\omega) \xdashrightarrow{u} (\bcmp,\Bcmp)$$

In the latter case, we reset all nodes labelled by $\omega$, convert all $\omega'$ back to $\omega$.
$$(\ecmp') \xdashrightarrow{r, \omega} (\convert)$$
$$(\convert,\omega') \xdashrightarrow{u} (\convert,\omega)$$
$$(\convert) \xdashrightarrow{r, \omega'} (\convert')$$

We then declare the node labelled by $\copied$ as the winner and remove the node labelled by $\Acmp$ (since it is a copy of the winner).

$$(\convert',\copied) \xdashrightarrow{u} (\convert'',\biggest)$$
$$(\convert'',\Acmp) \xdashrightarrow{u} (\ebig)$$

\subsection{Constructing $(k+1)$-Comparator Gadget}

\kplusonecomparator*

The main idea is a culmination of various different ideas from different gadgets. First we isolate the biggest-grandchildren $w_A, w_B$ in each of the two marked children $v_A, v_B$ of the root. Then, we copy the subtrees of $w_A, w_B$ as $w_A', w_B'$. Then, we apply the $k$-comparator gadget between $w_A'$ and $w_B'$, whilst also ensuring
that whatever operations are being done on the subtree of $w_A'$ are also being done on the  subtree of $w_A$ (similarly for $w_B'$ and $w_B$). We do this to ensure that
the subtree of $w_A$ is always an induced subtree of $w_A'$ at each point (similarly for $w_B'$ and $w_B$). Then, once we get the answer for $w_A', w_B'$, we are no longer guaranteed
that they still represent the biggest-children in $v_A, v_B$. So, we run another biggest-child check using $w_A$ and $w_B$, which if passes, will imply that $w_A'$ and $w_B'$
are still the biggest children of $v_A$ and $v_B$ respectively.

Once this check passes, if $w_A'$ is bigger than $w_B'$, then we know that $v_A$ is bigger than $v_B$ and vice versa. If $w_A'$ turns out to be equal to $w_B'$ then 
we mark them as $\omega'$ and move on to the next biggest-child of $v_A$ and $v_B$
and continue the gadget. This way we will either conclude that one of $v_A, v_B$ is smaller than the other or we will exhaust all the children of $v_A, v_B$, at which point,
we convert all of its children labelled by $\omega'$ into $\omega$ and declare them as equal.
We now make this idea more concrete.

A $(k+1)$-comparator gadget begins with the root $u$ labelled by $\bcmp$
and two of its children $v_A, v_B$ labelled by $\Acmp$ and $\Bcmp$. 
We first employ the $k$-biggest-child gadget on $v_A, v_B$ to mark the node which represents
the biggest term in the CNF of the ordinals represented by $v_A, v_B$. More precisely, 
we first do
$$(\bcmp, \Acmp) \xdashrightarrow{u} (\bcmp^A,\bbig)$$
and then run the $k$-biggest-child gadget on $v_A$. Technically, this would require us to prepend $\bcmp^A$ with every transition of this gadget. At the end of this gadget,
the root $u$ will be labelled by $\bcmp^A$, the node $v_A$ will be labelled by $\ebig$
and one of its children (say $w_A$) will be labelled by $\biggest$. We now change $\ebig$ to $\bcopy$,
$\biggest$ to $\marked$ and copy this subtree.

$$(\bcmp^A, \ebig, \biggest) \xdashrightarrow{u} (\bcmp^A,\bcopy,\marked)$$

We now run the $k$-copy gadget to copy the subtree rooted at $w_A$, i.e., the node with the label $\marked$. Similar to the previous subsection, we modify the copy gadget so that it 
leaves $w_A$ as $\omega'$ and creates a copy $w_A'$ labelled by $\copied$. So at this point,
the root will be $\bcmp^A$, $v_A$ will be $\ecopy$, $w_A$ will be $\omega'$ and $w_A'$
will be $\copied$. We now do the exact same steps for the node $v_B$, i.e., the node labelled by $\Bcmp$. This will lead us to a tree where $u$ will be $\bcmp^B$, $v_A, v_B$ will be $\ecopy$, $w_A, w_B$ will be $\omega'$ and $w_A',w_B'$ will be $\copied$.

We now run the $k$-comparator gadget on $w_A',w_B'$ to check if one is bigger than the other.
For this, we first set them up.
$$(\bcmp^B, \ecopy, \copied) \xdashrightarrow{u} (\bcmp',\bcmp,\Acmp)$$
$$(\bcmp', \ecopy, \copied) \xdashrightarrow{u} (\bcmp'',\bcmp,\Bcmp)$$

Then, we run the $k$-comparator gadgets on them. However, note that they are not actually in the same subtree. Nevertheless, we can still modify the $k$-comparator gadget in a straightforward manner so that it compares them across two different subtrees.
For instance, this would involve taking every update transition $(p_0,\dots,p_\ell) \xdashrightarrow{u} (q_0,\dots,q_j)$ of the $i$-comparator gadget and transforming it as
$$(\bcmp'',p_0,\dots,p_\ell) \xdashrightarrow{u} (\bcmp''',q_0',\dots,q_j)$$
$$(\bcmp''',p_0) \xdashrightarrow{u} (\bcmp'''',q_0)$$
$$(\bcmp'''',q_0') \xdashrightarrow{u} (\bcmp'',q_0)$$
This just ensures that the labels of $v_A, v_B$ are the same at each step, thereby acting as if they are one node and therefore simulating as if $w_A', w_B'$ belong to the same subtree.

Furthermore, as we run the $k$-comparator gadgets on both of them, we also ensure
that any change made to $w_A'$ (resp. $w_B'$) is also mirrored in $w_A$ (resp. $w_B$).
More precisely, if we execute a transition
$(p_0,p_1,p_2,\dots,p_\ell) \xdashrightarrow{u} (q_0,q_1,q_2,\dots,q_j)$
of the $k$-comparator gadget we replace it with
$$(p_0,p_1,p_2,\dots,p_\ell) \xdashrightarrow{u} (q_0',q_1',q_2,\dots,q_j)$$
$$(q_0',q_1',\omega',\dots,p_\ell) \xdashrightarrow{u} (q_0,q_1,\omega',\dots,q_j)$$

This ensures that whatever changes were made to the subtree rooted at $w_A'$ are also maintained on the subtree rooted at $w_A$, since $w_A$ is the only node labelled by $\omega'$
in the subtree of $v_A$. (Similarly for $w_B$ and $v_B$).

We then wait till the gadget finishes, which will happen when both $v_A$ and $v_B$ get the label $\ecmp$ (and the root by the label $\bcmp''$). At this point, the nodes $w_A', w_B'$ would have been compared. Before, we check the result of this comparison, we
first check if $w_A', w_B'$ are still the biggest-children in the subtrees of $v_A, v_B$,
by running the $k$-biggest-child gadgets on $w_A$ and $w_B$ respectively. If the check passes
on $w_A$ and $w_B$, then since we have guaranteed that $w_A$ and $w_B$ are respectively smaller than $w_A'$ and $w_B'$, we would be guaranteed that $w_A'$ and $w_B'$ still 
represent the biggest-children of $v_A$ and $v_B$ respectively.

To do this, we mark $w_A$ as a candidate for the biggest child and compare it with 
all the children of $v_A$ labelled by $\omega$, in exactly the same way as we constructed
the $k$-biggest child gadget. Similarly we do this for $w_B$.

If these checks pass, we can destroy $w_A,w_B$ and then look at the results of the comparator gadget on $w_A'$ and $w_B'$. If say $w_A'$ is smaller than $w_B'$, then we can immediately conclude that $v_A$ is smaller than $v_B$ and vice versa. 

$$(\bcmp'',\ecmp,\smallcmp) \xdashrightarrow{u} (\bcmp''',\smallcmp,\omega)$$
$$(\bcmp''',\ecmp,\bigcmp) \xdashrightarrow{u} (\textit{declare},\bigcmp,\omega)$$

After declaring this, we just need to convert all the $\omega'$'s in the subtrees of $v_A$
and $v_B$ back to $\omega$, which we have already seen how to do before, after which we can end the computation.

On the other hand, if $w_A, w_B$ turn out to be equal, then we mark them as $\omega'$
and repeat the computation all over again.

$$(\bcmp'',\ecmp,\equalcmp) \xdashrightarrow{u} (\textit{again},\Acmp,\omega')$$
$$(\textit{again},\ecmp,\equalcmp) \xdashrightarrow{u} (\bcmp,\Bcmp,\omega')$$

This completes the construction of the $(k+1)$-comparator gadget.

\section{Proofs of Section~\ref{sec:upper-bounds}}

\subsection{Proofs of Subsection~\ref{subsec:WSTS}}\label{subsec:appendix-WSTS}

\compatibility*

\begin{proof}
    We discuss the case when $t$ is a decrementing transition; the argument for the other transitions are similar.

    Let $t = (p_0,\dots,p_i) \xdashrightarrow{u} (q_0,\dots,q_j)$ be a decrementing transition with $j < i$. Since $C \act{t} C'$, it follows that
    $C$ must have a path $v_0,\dots,v_i$ from the root labelled by $p_0,\dots, p_i$.
    Furthermore, since $C \act{t} C'$, it follows that $C'$ is the same as $C$
    except that the nodes $v_0,\dots,v_j$ are now labelled by $q_0,\dots,q_j$ and the subtree
    rooted at $v_{i+1}$ has been removed. 

    Now, since $C \le_{is} D$, it follows that $D$ also has the path $v_0,\dots,v_i$ from the root labelled by $p_0,\dots, p_i$. Hence, we can fire the transition $t$ from $D$
    to reach some configuration $D'$. Since $C$ is an induced subgraph of $D$,
    it is then easy to see that $C'$ is an induced subgraph of $D'$ as well.
\end{proof}

\effectivepred*

\begin{proof}
    Let $C$ be a configuration and $t$ be a transition. Define $pre_t(\uparrow C)$ to be the set of all configurations $C'$ such that $C' \act{t} D$ for some some $D \in \uparrow C$.
    Clearly $pre(\uparrow C) = \cup_{t \in \delta_u \cup \delta_r}\  pre_t(\uparrow C)$.
    Hence, to prove the proposition, it suffices to show the following: For every $t \in \delta_u \cup \delta_t$, in exponential time we can compute a set of configurations 
    $C_1,\dots,C_m$ such that $\cup_{i=1}^m \uparrow C_i = pre_t(\uparrow C)$ and $|C_i| \le |C| + k+1$
    for every $i$. We will call such a collection $\{C_1,\dots,C_m\}$ a basis for $pre_t(\uparrow C)$.

    We discuss the case when $t$ is a decrementing transition; the argument for the other cases are similar. We claim the following, which we call Claim A:
    \begin{quote}
        Suppose $C' \in pre_t(\uparrow C)$. Then there exists $C''$
        such that $C'' \le_{is} C'$, $|C''| \le |C|+k+1$  and $C'' \in pre_t(\uparrow C)$.
    \end{quote}
    Assuming Claim A is true, it follows that we can compute a basis 
    for $pre_t(\uparrow C)$ as follows: Collect all configurations $C''$ satisfying $|C''| \le |C| + k +1$ and $C'' \in pre_t(\uparrow C)$. Let this collection be $\{C_1,\dots,C_m\}$. 
    Note that this computation can take at most exponential time, since there can only be 
    exponentially many configurations $C''$ satisfying $|C''| \le |C| + k+1$.
    By Claim A, $pre_t(\uparrow C) \subseteq \cup_{i=1}^m \uparrow C_i$.
    By the compatibility property, $\cup_{i=1}^m \uparrow C_i \subseteq pre_t(\uparrow C)$.
    Hence, we get that $\cup_{i=1}^m \uparrow C_i = pre_t(\uparrow C)$.
    
    Therefore, it suffices only to prove Claim A. Suppose $C' \in pre_t(\uparrow C)$ for some $t = (p_0,\dots,p_i) \xdashrightarrow{u} (q_0,\dots,q_j)$ with $j < i$. Hence $C' \act{t} D$ for some $D \in \uparrow C$.
    Since $C' \act{t} D$, by definition, this means that $C'$ contains a path 
    $v_0,\dots,v_i$ from the root labelled by $p_0,\dots,p_i$
    and $D$ is the same as $C'$ except that $v_0,\dots,v_j$ are labelled
    by $q_0,\dots,q_j$ and the subtree at $v_{j+1}$ is now missing.

    Now since $D \in \uparrow C$, it follows that there must be a subset $S$ of nodes of $D$ such that when $D$ is restricted to $S$, we get $C$. 
    We split this subset $S$ as $S^- \cup S^+$ where $S^- = S \cap \{v_0,\dots,v_j\}$
    and $S^+ = S \setminus S^-$. By definition of the transition $t$, note that $S^+$ must also be present in $C'$.

    Now, we remove all nodes from $C'$ except for $v_0,\dots,v_i$ and $S^+$ to get a new tree $C''$. Clearly, $C'' \le_{is} C'$ and $|C''| = |S^+| + (i+1) \le |C| + k+1$.
    Also, since we keep $v_0,\dots,v_i$ in $C''$, we must be able to fire $t$ from $C''$ 
    to reach some $D''$. It is clear that $D''$ contains the nodes
    $v_0,\dots,v_j$ and $S^+$ and so contains all of $S$. Hence, $D'' \in \uparrow C$
    and so $C'' \in pre_t(\uparrow C)$. This then proves Claim A.
\end{proof}

\subsection{Proofs of Subsection~\ref{subsec:upper-bounds-nested-multiset-nwqos}}
\label{subsec:appendix-reflections}

\reflections*

\begin{proof}
    The proof considers all of the four cases in the definition of $R_n(A,a)$. The proofs for the first and the last cases are easy to see. The other two cases
    follow from an identity of~\cite[Proposition 1]{Rosa-Velardo17}. Because the identity in that paper is stated in different terms compared to ours, we give the complete proof here, by structural induction on $A$.     

    Before we begin the induction, we make a remark on notation. For any two elements $m, m'$,
    whenever it is clear from the context which nested multiset nwqo $A$ they belong to, we will simply use $m \le m'$ to denote $m \le_A m'$. This is done for the sake of brevity, as we will be dealing
    with multiple complicated nested multiset nwqo expressions such as $\multi{\sum_{i=1}^d B_i}$. 
    We will similarly use $|m|$ to denote $|m|_A$. 
    
    \paragraph{Case 1: $A = \multi{\Gamma_0}$} In this case, $a$ must be the only element of $A$. It follows then that $A/a \hookrightarrow \Gamma_0 = R_n(A,a)$.

    \paragraph{Case 2: $A = \multi{\sum_{i=1}^d B_i}$} Let $b_i$ denote the sub-multiset
    of $a$ restricted to only elements from $B_i$. Note that each $b \in \multi{B_i}$ and
    $a$ is the multiset sum of $b_1,\cdots,b_d$.
    Furthermore, by construction the norm of each $b_i$ is at most $n$ (since the norm of $a$ is at most $n$).
    Let $R_n(\multi{B_i},b_i) = \sum_{j=1}^{\ell_i} \multi{B_i^j}$. 
    By induction hypothesis, $r_{i}: \multi{B_i}/b_i \hookrightarrow R_n(\multi{B_i},b_i)$ is a reflection for each $i$. 
    Furthermore, by definition, $R_n(A,a) = \sum_{i=1}^d \sum_{j=1}^{\ell_i} \multi{\sum_{k \neq i} B_i + B_i^j}$. Our task is now to come up with a reflection $r: A/a \hookrightarrow R_n(A,a)$. We divide this into three steps - in the first step, we define the mapping $r$;
    in the second step, we show that $r(m) \le r(m')$ implies $m \le m'$;
    and in the third step, we show that $|r(m)| \le |m|$.

    \textit{First step: } For each $m \in A/a$, we define an element $r(m) \in R_n(A,a)$ as follows:
    For any $1 \le k \le d$, let $m_k$ be the sub-multiset of $m$ restricted only to elements
    from $B_k$. Note that each $m_k \in \multi{B_k}$.
    Furthermore, since $m \in A/a$, there must be an index $i$ such that 
    $m_i \in \multi{B_i}/b_i$. Based on this we define $r(m)$ as the sum of the multisets 
    $\{m_k : k \neq i\}$ and $r_i(m_i)$, i.e.,
    $r(m) = \sum_{k \neq i} m_k + r_{i}(m_i)$. By assumption $r_i(m_i) \in \multi{B_i^j}$ for some $1 \le j \le \ell_i$
    and so, by construction, $r(m)$ belongs to 
    $\multi{\sum_{k \neq i} B_k + B_i^j}$
    and is hence an element of $R_n(A,a)$.
    We now have to show that this satisfies the properties of a reflection.

    \textit{Second step: }
    Suppose $r(m) \le r(m')$. By assumption, this means that there is some $1 \le i \le d$ and $1 \le j \le \ell_i$
    such that both $r(m),r(m') \in \multi{\sum_{k \neq i} B_k + B_i^j}$. (This is because only multisets belonging to the same summand of $R_n(A,a)$ can be compared). We want to now show that $m \le m'$.
    For each $1 \le k \le d$, let $m_k$ be the sub-multiset of $m$ restricted to only elements 
    of $B_k$. Similarly, we can obtain $m'_k$ from $m$.
    By definition of $r$ and by the fact that $r(m) \le r(m')$, it follows that
    $m_k \le m'_k$ for each $k \neq i$ and $r_{i}(m_i) \le r_i(m'_i)$. (Once again this is because elements from $m_k$ can only be compared with elements from $m'_k$).
    Since $r_i$ is a reflection, it follows that $m_i \le m'_i$ as well.
    Hence, it follows that $m_k \le m_k'$ for all $k$. From this it immediately follows that 
    $m \le m'$.

    \textit{Third step: } Let us now bound the norm of $r(m)$. 
    Let $r(m) = \sum_{k \neq i} m_k + r_i(m_i)$ where each $m_k$ is the sub-multiset
    of $m$ restricted only to elements from $B_k$. Hence, by definition
    of the norm function, we get that $|r(m)| = \sum_{k \neq i} |m_k| + |r_i(m_i)|$,
    By properties of a reflection, we have that $|r_i(m_i)| \le |m_i|$. It follows that $|r(m)| \le \sum_{1 \le k \le d} |m_k| = |m|$ and so we are done.

    \paragraph{Case 3: $A = \multi{\multi{B}}$}
    Let $a = \multiset{a_1,\dots,a_d}$. Since $|a| \le n$, it follows that $|a_i| \le n$ for each $i$ and $d \le n$.
    By induction hypothesis, $r_i : \multi{B}/a_i \hookrightarrow R_n(\multi{B},a_i)$ are reflections for each $i$.
    Moreover, by definition, $R_n(A,a) = \sum_{i=1}^n \multi{B \cdot (n-1) + R_n(\multi{B},a_i)}$.
    Our task is now to come up with a reflection $r: A/a \hookrightarrow R_n(A,a)$.
    Similar to the previous case, we divide this into three steps - in the first step, we define the mapping $r$;
    in the second step, we show that $r(m) \le r(m')$ implies $m \le m'$;
    and in the third step, we show that $|r(m)| \le |m|$.

    \textit{First step: } Let $m_1 \in A/a$. Hence, either there must be no element $x_1 \in m_1$ such that $a_1 \le x_1$ 
    or there is such an element $x_1$ and $\multiset{a_2,\dots,a_d} \nleq m_2 = m_1-\multiset{x_1}$, where $m_1 - \multiset{x_1}$ is the multiset obtained from $m_1$ by decrementing $m_1(x_1)$ by 1.
    In the latter case, either there is no element $x_2 \in m_2$ such that $a_2 \le x_2$ or there is such an element
    $x_2$ and $\multiset{a_3,\dots,a_d} \nleq m_3 = m_2 - \multiset{x_2}$. 
    We can continue this process until
    we reach an $i$ such that there is no element $x_i \in m_i$ such that $a_i \le x_i$. (Such an $i$ has to exist,
    as otherwise we would have found $x_1,x_2,\dots,x_d \in m_1$ such that $a_1 \le x_1, a_2 \le x_2, \dots, a_d \le x_d$,
    contradicting the fact that $m_1 \in A/a$). Hence, $m_1 = \multiset{x_1,\dots,x_{i-1}} + m_i$.

    We now make two observations: First, since there is no $x_i \in m_i$ such that $a_i \le x_i$, we have that
    $m_i \in A/\multiset{a_i} = \multi{\multi{B}/a_i}$. Hence, it is possible to apply the reflection $r_i$ to each element of $m_i$. Therefore, we can define a mapping $s(m_1) = \multiset{r_i(x) : x \in m_i}$ which gives us that $s(m_1) \in \multi{R_n(B,a_i)}$.
    Second, suppose for each $1 \le j \le i-1$, we have 
    $x_j = \multiset{x_j^1,\dots,x_j^{\ell_j}}$. Then, if we construct the multiset
    $s'(m_1) = \multiset{x : x = (j,x_j^k) \text{ for some } 1 \le j \le i-1, 1 \le k \le \ell_j}$,
    we have that $s'(m_1) \in \multi{B \cdot (n-1)}$. 
    
    Putting these two observations together, 
    we then define $r(m_1) = \multiset{x : x \in s(m_1) \text{ or } x \in s'(m_1)}$. Note that 
    $$r(m_1) \in \multi{B \cdot (n-1) + R_n(B,a_i)}$$ and is hence an element of $R_n(A,a)$. We now claim that $r$ satisfies the properties of a reflection.

    \textit{Second step: } Suppose $r(m) \le r(m')$. By assumption, there is some $i$ such that both $r(m), r(m') \in \multi{(n-1)B + R_n(B,a_i))}$ for some $i$. (This is because only multisets belonging to the same summand of $R_n(A,a)$ can be compared). Hence, there must be elements $x_1,\dots,x_{i-1} \in m$
    and $x_1',\dots,x_{i-1}' \in m'$ such that $m = \multiset{x_1,\dots,x_{i-1}} + m_i, m' = \multiset{y_1,\dots,y_{i-1}} + m_i'$, each $a_j \le x_j, a_j \le y_j$ and $m_i, m_i' \in A/\multiset{a_i}$.

    Let each $x_j = \multiset{x_j^1,\cdots,x_j^{d_j}}$ and let each $y_j = \multiset{y_j^1,\cdots,y_j^{\ell_j}}$. 
    By construction of $r$ we have that $$r(m) = \multiset{x : x = r_i(x') \text{ for some } x' \in m_i \text{ or } x = (j,x_j^k) \text{ for some } j,k}$$
    and $$r(m') = \multiset{y : y = r_i(y') \text{ for some } y' \in m_i' \text{ or } y = (j,y_j^q) \text{ for some } j,q}$$
    Since $r(m) \le r(m')$, it follows that there is an injection $h$ from $r(m)$ to $r(m')$ such that $x \le h(x)$ for all $x \in r(m)$. 
    By construction of $r, r'$ it follows that if $x = (j,x_j^k)$ for some $j,k$ then $h(x)$ must be $(j,y_j^q)$ for some $q$.
    Similarly, if $x = r_i(x')$ for some $x' \in m_i$ then $h(x)$ must be $r_i(y')$ for some $y' \in m_i'$.
    
    We now use $h$ to define an injection
    $w$ from $m$ to $m'$ such that $x \le w(x)$ for all $x \in m$.
    We split the definition of $w$ into two cases:
    \begin{itemize}
        \item Suppose $x = x_j$ for some $1 \le j \le i-1$. We then set $w(x) = y_j$.
        We now have to prove that $x_j \le y_j$, i.e., provide an injection $h'$
        from $\multiset{x_j^1,\dots,x_j^{d_j}}$ to $\multiset{y_j^1,\dots,y_j^{\ell_j}}$
        such that $x_j^k \le h'(x_j^k)$ for all $k$.
        We simply take this injection to be $h$ restricted to $\multiset{x_j^1,\dots,x_j^{d_j}}$, i.e., if $h(j,x_j^k) = (j,y_j^q)$,
        then $h'(x_j^k) = y_j^q$. Because $(j,x_j^k) \le (j,y_j^q)$,
        we have that $x_j^k \le y_j^q$. This provides the desired injection $h'$
        and so we have that $x_j \le y_j$.
        \item Suppose $x \in m_i$. Then, we have that $r_i(x) \in r(m)$
        and so we have that $h(r_i(x)) = r_i(y)$ for some $y \in m_i$.
        We then set $w(x) = y$. By assumption on $h$, we have that $r_i(x) \le r_i(y)$.
        By properties of a reflection, we have that $x \le y$.
    \end{itemize}
    Hence, we have shown that $m \le m'$, thereby completing this step.

    \textit{Third step: } Let us now bound the norm of $r(m)$. By the same reasoning as in the previous case, it follows that
    $|r(m)| = \sum_{1 \le j \le i-1} |x_j| + |r_i(m_i)|$.
    By properties of a reflection we have that $|r_i(m_i)| \le |m_i|$. It follows then that $|r(m)| \le \sum_{1 \le j \le i-1} |x_j| + |m_i| = |m|$ and so we are done.

    \paragraph{Case 4: $A = \sum_{i=1}^d \multi{B_i}$}
    Let $a \in \multi{B_i}$ for some $i$. 
    Note that $A/a$ is then simply $\sum_{k \neq i} \multi{B_k} + \multi{B_i}/a$.
    By induction hypothesis, we know that $r_i : \multi{B_i}/a \hookrightarrow R_n(\multi{B_i},a)$ is a reflection.
    Furthermore, $R_n(A,a) = \sum_{k \neq i} \multi{B_k} + R_n(\multi{B_i},a)$.
    It is then easy to see that the mapping $r: A/a \to R_n(A,a)$ defined by $r(a') = a'$ if $a' \notin \multi{B_i}$ 
    and $r(a') = r_i(a')$ if $a' \in \multi{B_i}/a$ is also a reflection.
\end{proof}

\subsection{Proofs of Subsection~\ref{subsec:order-types}}
\label{subsec:appendix-order-types}

\ordertype*

\begin{proof}
    Let us prove by structural induction on $A$ that $R_n(A,a) \hookrightarrow C(\alpha')$ for some $\alpha' \in \delta_n(o(A))$. 
    
    \paragraph*{Case 1: $A = \multi{\Gamma_0}$}
    In this case, $\delta_n(o(A)) = \{D_n(o(A))\} = \{D_n(\omega^0)\} = \{0\}$.
    Since $R_n(A,a) = \Gamma_0 = C(0)$, this proves the claim for this case.

    \paragraph*{Case 2: $A = \multi{\sum_{i=1}^d B_i}$} 
    In this case, $\delta_n(o(A)) = \{D_n(o(A))\} = \{D_n(\omega^{\oplus_{i=1}^d o(B_i)})\}$.
    By definition, $$D_n(\omega^{\oplus_{i=1}^d o(B_i)}) = \bigoplus_{i=1}^d \bigoplus_{j=1}^{\ell_i} \omega^{\oplus_{k \neq i} \ o(B_k) \oplus \beta_i^j}$$ where for each $1 \le i \le d$, $\bigoplus_{j=1}^{\ell_i} \omega^{\beta_i^j} = D_n(\omega^{o(B_i)})$.

    Let $b_i$ denote the sub-multiset of $a$ restricted to only elements from $B_i$.
    By induction hypothesis, $R_n(\multi{B_i},b_i) \hookrightarrow C(\bigoplus_{j=1}^{\ell_i} \omega^{\beta_i^j}) = \sum_{j=1}^{\ell_i} C(\omega^{\beta_i^j})$.
    Now, by definition of $R_n(A,a)$ and using the above property of $R_n(\multi{B_i},b_i)$, it follows that $$R_n(A,a) \hookrightarrow \sum_{i=1}^d \sum_{j=1}^{\ell_i} \multi{\sum_{k\neq i} B_k + C(\beta_i^j)} = C\left(\bigoplus_{i=1}^d \bigoplus_{j=1}^{\ell_i} \omega^{\oplus_{k \neq i} \ o(B_k) \oplus \beta_i^j}\right)
    $$ 

    This proves the claim for this case.

    \paragraph*{Case 3: $A = \multi{\multi{B}}$} 
    In this case, $\delta_n(o(A)) = \{D_n(o(A))\} = \{D_n(\omega^{\omega^{o(B)}})\}$.
    By definition, $$D_n(\omega^{\omega^{o(B)}}) = \omega^{o(B) \cdot (n-1) \oplus D_n(\omega^{o(B)})} \cdot n$$

    Let $a = \multiset{a_1,\dots,a_d}$. Note that $d \le n$.
    By induction hypothesis, $R_n(\multi{B},a_i) \hookrightarrow C(D_n(\omega^{o(B)}))$ for any $a_i$. Now, by definition of $R_n(A,a)$ and using the above property, it follows that $$R_n(A,a) \hookrightarrow \multi{B \cdot (n-1) + C(D_n(\omega^{o(B)}))} \cdot n = C \left(\omega^{o(B) \cdot (n-1) \oplus D_n(\omega^{o(B)})} \cdot n\right)$$

    This proves the claim for this case.
    
    \paragraph*{Case 4: $A = \sum_{i=1}^d \multi{B_i}$} 
    Let $a \in \multi{B_i}$ for some $i$. 
    In this case, by definition of $\delta_n(o(A)))$ we have that $$\alpha' = D_n(o(\multi{B_i})) \oplus \bigoplus_{k \neq j} o(\multi{B_k}) \in \delta_n(o(A))$$
    
    By induction hypothesis, $R_n(\multi{B_i},a) \hookrightarrow C(D_n(o(\multi{B_i})))$. Now, by definition of $R_n(A,a)$ and using the above property, it follows that $$R_n(A,a) \hookrightarrow \sum_{k \neq i} \multi{B_k} + C(D_n(o(\multi{B_i}))) = C(\alpha')$$
    
    This proves the claim for this case.
\end{proof}

\subsection{Proofs of Subsection~\ref{subsec:Cichon}}
\label{subsec:appendix-Cichon}

\lean*

\begin{proof}
    We prove this proposition by induction on $\alpha$. 
    
    \paragraph{Case 1: $\alpha = 1 = \omega^0$} In this case, the claim immediately follows. 

    \paragraph{Case 2: $\alpha = \omega^\beta$ where $\beta = \omega^{\gamma}$}
    In this case, $\alpha = \omega^{\omega^\gamma}$. By definition, $D_n(\alpha) = \omega^{\gamma \cdot (n-1) + D_n(\beta)} \cdot n$. Note that $\Omega_{k-1} \le \beta < \Omega_k$ and $\beta$ is $\ell$-lean. Hence, by induction hypothesis,
    $D_n(\beta)$ is $\ell n (k-1)$ lean. Furthermore, since $\alpha$ is $\ell$-lean, it follows that $\gamma$ is also $\ell$-lean
    and hence $\gamma \cdot (n-1)$ is $\ell (n-1)$-lean. 
    It then follows that $D_n(\alpha)$ is $\max(n,\ell (n-1) + \ell n (k-1))$-lean.
    Since $\max(n,\ell (n-1) + \ell n (k-1)) \le \ell n k$, it follows that $D_n(\alpha)$ is $\ell n k$-lean.

    \paragraph{Case 3: $\alpha = \omega^\beta$ where $\beta = \sum_{i=1}^{m} \omega^{\eta_i}$}
    By definition, $$D_n(\omega^{\omega^{\eta_i}}) = \omega^{\eta_i \cdot (n-1) \oplus D_n(\omega^{\eta_i})} \cdot n
    = \oplus_{j=1}^n \ \omega^{\eta_i \cdot (n-1) \oplus D_n(\omega^{\eta_i})}$$
    Now, by definition, 
    \begin{equation}\label{eq:Dn}
      D_n(\alpha) = D_n(\omega^\beta) = \oplus_{i=1}^{m} \oplus_{j=1}^n \omega^{\oplus_{p \neq i} \omega^{\eta_p} \oplus \eta_i \cdot (n-1) \oplus D_n(\omega^{\eta_i})}  
    \end{equation}

    Now let $\beta = \sum_{i=1}^d \omega^{\gamma_i} \cdot c_i$ in strict CNF.
    Let us now group equal factors of $\eta_k$ in the exponent in equation~\ref{eq:Dn}
    and rewrite the equation as
    \begin{equation}
      D_n(\alpha) =\oplus_{i=1}^d (\oplus_{j=1}^n \omega^{\oplus_{p \neq i} \omega^{\gamma_p} \cdot c_p \oplus \omega^{\gamma_i} \cdot (c_i-1) \oplus \gamma_i \cdot (n-1) \oplus D_n(\omega^{\gamma_i})}) \cdot c_i
    \end{equation}

    Furthermore, since the $j$ plays no role in the summand, it follows that 
    $$D_n(\alpha) = \oplus_{i=1}^d (\omega^{\oplus_{p \neq i} \omega^{\gamma_p} \cdot c_p \oplus \omega^{\gamma_i} \cdot (c_i-1) \oplus \gamma_i \cdot (n-1) \oplus D_n(\omega^{\gamma_i})}) \cdot c_i n$$

    Let 
    \begin{equation}
    \beta_i = \oplus_{p \neq i}  \omega^{\gamma_p} \cdot c_p \oplus \omega^{\gamma_i} \cdot (c_i-1) \oplus \gamma_i \cdot (n-1) \oplus D_n(\omega^{\gamma_i})    
    \end{equation}
    Note that $\beta_i \neq \beta_j$ for two distinct $i,j$.
    Indeed, suppose $\beta_i = \beta_j$ with $i \neq j$. Then removing common terms from both sides, we are left 
    with the following equality
    $$\omega^{\gamma_j} \oplus \gamma_i \cdot (n-1) \oplus D_n(\omega^{\gamma_i}) = \omega^{\gamma_i} \oplus  \gamma_j \cdot (n-1) \oplus D_n(\omega^{\gamma_j})$$
    Without loss of generality, let $\gamma_j > \gamma_i$. Then, each term of the RHS is strictly less than $\omega^{\gamma_j}$,
    but the LHS contains $\omega^{\gamma_j}$, which leads to a contradiction. Hence, $\beta_i \neq \beta_j$ for $i \neq j$.

    Therefore, since $c_i n \le \ell n$, in order to prove that $D_n(\omega^\beta)$ is $\ell nk $-lean, it suffices to show that each $\beta_i$ is $\ell nk$-lean. Now, each $\gamma_p$ 
    for $p \neq i$ as well as $\gamma_i$ and $D_n(\omega^{\gamma_i})$ are $\ell n (k-1)$-lean (The first two are simply $\ell$-lean by assumption and the third one follows from induction hypothesis).
    Let us now analyze the coefficients that can appear in $\beta_i$.
    The maximum coefficient that can occur in $\beta_i$ is when we  combine the coefficient $c_p$ of some $\omega^{\gamma_p}$ with $p \neq i$ along with the coefficients of $\gamma_i \cdot (n-1)$ and $D_n(\omega^{\gamma_i})$. This value can be at most $\ell + \ell (n-1) +\ell n (k-1) = \ell n k$, which proves the claim for this case.

    \paragraph{Case 4: $\alpha = \sum_{i=1}^d \omega^{\beta_i} \cdot c_i$}
    In this case, if $\alpha' \in \delta_n(\alpha)$, we have that $\alpha' = \oplus_{p \neq i} \omega^{\beta_p} \cdot c_p \oplus \omega^{\beta_i} \cdot (c_i-1) \oplus D_n(\omega^{\beta_i})$. By induction hypothesis, we have that
    $D_n(\omega^{\beta_i})$ is $\ell n k$-lean. It then follows that the maximum coefficient 
    that can occur in $\alpha'$ is when we combine the coefficient $c_p$ of some $\omega^{\gamma_p}$ with $p \neq i$ along with the coefficients of $D_n(\omega^{\beta_i})$.
    This value can be at most $\ell + \ell n k$-lean, which proves the claim for this case.
\end{proof}

\subsubsection*{Proof of Lemma~\ref{lem:final-bound}}

Before, we give the proof for this lemma, we state some notation and recall some facts regarding ordinals.

For any ordinal $\alpha > 0$ and any $n \in \mathbb{N}$, set $P_n(\alpha) = \beta$ if $\alpha = \beta + 1$ and otherwise set $P_n(\alpha) = P_n(\alpha_n)$ where $\alpha_n$ is the $n^{th}$ element in the fundamental sequence of $\alpha$. From Proposition~\ref{prop:lean}, we get the following result.

\begin{proposition}\label{prop:delta-pred-relation}
    Let $\alpha < \Omega_{k+1}$ be $\ell$-lean and $\alpha' \in \delta_n(\alpha)$.
    Then for all $x \ge \ell + \ell n k$, $h_{\alpha'}(x) \le h_{P_{\ell + \ell n k}(\alpha)}(x)$.
\end{proposition}

\begin{proof}
    It is known that if $\beta < \alpha$ and $\beta$ is $m$-lean then for all $x \ge m$, we have $h_\beta(x) \le h_{P_m(\alpha)}(x)$~\cite[Lemma B.1, Lemma C.9]{arxiv-SchmitzS}. Plugging in $\beta = \alpha'$ and $m = \ell + \ell nk$ and using Proposition~\ref{prop:lean}, we get the required result.
\end{proof}

\finalbound*

\begin{proof}
    We proceed by induction on $\alpha$. If $\alpha = 0$, then $M_0(n) = 0 = h_0(\ell nk)$. 
    Similarly, if $0 < \alpha < \omega$, we can show that $M_\alpha(n) = \alpha = h_\alpha(\ell nk)$. Therefore, we can assume that $\alpha \ge \omega$ and so $k \ge 1$.
    
    By definition of $M_{\alpha,g}$, there must be some $\alpha' \in \delta_n(\alpha)$ such that
    $$M_{\alpha,g} = 1 + M_{\alpha',g}(g(n))$$
    Since $\alpha$ is $k$-lean, by Proposition~\ref{prop:lean}, $\alpha'$ is $\ell + \ell n k$-lean. We now apply the induction
    hypothesis to $\alpha'$ to get
    $$1+M_{\alpha',g}(g(n)) \le 1+h_{\alpha'}((\ell + \ell n k) + (\ell + \ell n k) \cdot g(n) \cdot k) = 1+h_{\alpha'}((\ell + \ell n k) \cdot (1 + g(n) \cdot k)) $$

    Since $g$ is a superadditive function, it follows that $1 + g(n) \cdot k = 1 + \sum_{i=1}^k g(n) \le 1 + g(nk)$. Furthermore, since $g$ is strictly increasing, it follows that $1 + g(nk) \le g(\ell + \ell nk)$.
    Combining this observation with the fact that $h_\alpha'$ is monotonic, gives us that
    $$1+h_{\alpha'} ((\ell + \ell n k) \cdot (1 + g(n) \cdot k)) \le
    1+h_{\alpha'} ((\ell + \ell nk) \cdot g(\ell + \ell n k)) = 1+h_{\alpha'}(h(\ell + \ell n k))$$
    Observing that $h(\ell + \ell n k) \ge \ell + \ell n k $, if we now apply Proposition~\ref{prop:delta-pred-relation}, we get
    $$1+h_{\alpha'}(h(\ell + \ell n k)) \le 1+h_{P_{\ell + \ell n k}(\alpha)}(h(\ell + \ell n k)) = h_\alpha(\ell + \ell n k)$$
    where the last equality follows from the identity $1+h_{P_x(\alpha)}(h(x)) = h_\alpha(x)$
    ~\cite[Lemma 5.1]{SchmitzS11}.
    This completes the proof of the lemma.
\end{proof}

\section{Proofs for Subsection~\ref{subsec:bounded-depth-graphs}}
\label{sec:appendix-bounded-depth-graphs}

We now give the desired reflection from labelled $k$-depth graphs with labels from some set $Q$
to labelled trees of height at most $k$ with labels from $Q \times 2^{\{0,\dots,k-1\}}$.

The reflection $r$ is as follows: Suppose $G$ is a $k$-depth graph. This means that $G$ is obtained by taking some tree $T$ of height $k$ and then adding some edges $(u_1,v_1),\dots,(u_\ell,v_\ell)$, where each $v_i$ is a descendant of $u_i$ in $T$.
Let us call these the extra edges of $G$.

Hence, for each node $v$, we can assign the set $A(v)$ which contains a number $i$
iff the ancestor of $v$ in the tree $T$ at level $i$ is connected to $v$ by an extra edge. 
Note that $A(v) \subseteq \{0,\dots,k-1\}$. Now, we set $r(G) = T'$ 
where $T'$ is the same as $T$, except the label of each node $v$ in $T'$ is now a pair containing the label of $T$ and $A(v)$. From the construction, it is clear that $|r(G)| = |G|$. Hence, we only need to show that $r(G) \le_{is} r(G')$ implies $G \le_{is} G'$. To this end, suppose $T:= r(G) \le_{is} T' := r(G')$. This means that there is an injection from $T$ to $T'$
such that $(u,v) \in T$ iff $(h(u),h(v)) \in T'$. In other words, this means that there exists a subset $S$ of nodes of $T'$ such that when $T'$ is restricted to $S$, we get $T$. 
Similarly, let us now restrict $G'$ to the subset $S$ to get a graph $H$. We claim that $H = G$,
which would prove that $G \le_{is} G'$ and conclude the proof.

Indeed, since $T$ and $G$ have the same set of nodes, it follows that the nodes of $G$ and the nodes of $H$ are the same. We now have to show that the label of a node in $H$ and $G$ are the same and also that two nodes are connected by an edge in $H$ iff they are connected by an edge in $G$.

First, let us pick some node $v \in S$. Suppose its label in $T'$ is $(q,S)$.
By construction, its label in $G'$ (and hence also in $H$) must be $q$.
Furthermore, by assumption, the label of $v$ in $T'$ is the same as the label of $v$ in $T$.
Hence, the label of $v$ in $T$ is $(q,S)$ which implies that the label of $v$ in $G$ is also $q$.

Now, suppose we have an edge $(u,v) \in H$. By construction, this edge is also present in $G'$. If it is also present in $T'$, it must also be present in $T$ and hence also in $G$.
Therefore, assume that this edge is not present in $T'$. This implies that $(u,v)$ 
is an extra edge of $G'$. Let $u$ be the ancestor of $v$ at some level $i$ in $T'$. Since we have the extra edge $(u,v)$ in $G'$, the label of $v$ in $T'$ must be $(q,S)$ with $i \in S$. Therefore, the label of $v$ in $T$ must also be $(q,S)$ with $i \in S$. 
Furthermore, by assumption $u$ is the ancestor of $v$ at level $i$ in $T$.
This indicates that $G$ also has the extra edge $(u,v)$. Similarly, we can argue that if we have
an edge $(u,v) \in G$ then $(u,v) \in H$ as well. This proves that $H = G$ and so we are done.

\section{Proofs for Subsection~\ref{subsec:broadcast-networks}}
\label{sec:appendix-broadcast-networks}

\bcastnets*

Note that the upper bound for these problems have already been discussed in Subsection~\ref{subsec:broadcast-networks} and so we will concentrate on the lower bounds here.
First, we formalize the intuition behind broadcast networks discussed in Subsection~\ref{subsec:broadcast-networks}.

\subsection{$k$-depth Broadcast Networks}

Formally, a (broadcast) protocol is a tuple $\mathcal{P} = (Q,\Sigma,\delta,Q_{init})$ where $Q$ is a finite set of states, $\Sigma$ is a finite communication alphabet, $\delta \subseteq Q \times \{!a, ?a: a \in \Sigma\} \times Q$ is a finite set of transitions and $Q_{init} \subseteq Q$ is the set of initial states. For ease of notation, we will write $q \act{!a} q'$ (resp. $q \act{?a} q'$) for $(q,!a,q') \in \delta$ (resp. $(q,?a,q') \in \delta$). Intuitively, a transition of the form $q \act{!a} q'$ (resp. $q \act{?a} q'$) corresponds to broadcasting the message $a$ (resp. receiving the message $a$).

As discussed in Subsection~\ref{subsec:broadcast-networks}, a $k$-depth broadcast network consists of agents situated on the nodes of a $k$-depth graph and executing the protocol $\mathcal{P}$.
We formalize this as follows: A \emph{$k$-depth configuration} (or simply a configuration) $C = (V,E,L)$
is a tuple such that $(V,E)$ is a $k$-depth graph and $L : V \to Q$
is a labelling function that specifies the current state of the agent in a node. 
An initial configuration is a configuration in which all the nodes
have labels from $Q_{init}$.

The semantics of a $k$-depth broadcast network is given by means of steps between its configurations. Given two configurations $C = (V,E,L)$ and $C' = (V',E',L')$,
we say that there is a step between $C$ and $C'$ if $V = V', E = E'$ 
and there exists a node $v \in V$ and a message $a \in \Sigma$ such that
\begin{itemize}
    \item $(L(v),!a,L'(v)) \in \delta$ and 
    \item $(L(v'),?a,L'(v')) \in \delta$ if $v'$ is a neighbor of $v$ and 
    \item $L(v') = L'(v')$ if $v'$ is not a neighbor of $v$. 
\end{itemize}
In this case, we write $C \act{v,a} C'$ or simply $C \act{} C'$. Intuitively, in such a step,
the agent at node $v$ broadcasts the message $a$ and this is simultaneously received by all of its neighbors; all the other agents do nothing.

A run from a configuration $C$ to a configuration $C'$ is a sequence of steps
$C \act{} C_1 \act{} C_2 \act{} \cdots \act{} C'$. We use $C \act{*} C'$ to denote
that there exists a run between $C$ and $C'$. An execution is a run starting from some initial configuration.

Given a state $f \in Q$ and a configuration $C$, we say that $C$ covers $f$ if 
$f$ is the label of some node in $C$. We say that an execution $C_0 \act{*} C$ covers
$f$, if $C$ covers $f$. The \emph{coverability} problem for $k$-depth broadcast networks,
is to decide, given a broadcast protocol $\mathcal{P}$ and a state $f$, whether $f$ is coverable by some execution. The coverability problem for $k$-depth broadcast networks over trees is the same, except the configurations are now restricted to be $k$-depth trees, i.e.,
trees of height at most $k$. 

We now show $\fastf{\Omega_k}$-hardness of the coverability problem for both these models
by giving a reduction from coverability for $k$-NRCS. To this end, we first need three ingredients regarding coverability in $k$-NRCS as well as coverability in $k$-depth broadcast networks, which we shall now discuss.

\subsection{First Ingredient: Simple coverability for $k$-NRCS}
\label{subsec:simple-coverability}
The first ingredient that we shall need is the notion of simple coverability.
A simple coverability problem is just like the coverability problem, except that
the initial and final configurations are trees having one node labelled by states $q_{init}$
and $q_f$ respectively. If the coverability instance is true in this case, we simply say that $q_{init}$ can cover $q_f$.
We can immediately see that the simple coverability problem
is as hard as the coverability problem, by standard reductions.

Indeed, let $\net$ be any $k$-NRCS and let $C$ be any configuration of $\net$. 
Let $q_r$ be the label of the root of $C$. We can now construct two gadgets
$\net_{fwd}^C$ and $\net_{bwd}^C$ with the following properties:
\begin{itemize}
    \item Both $\net_{fwd}^C$ and $\net_{bwd}^C$ contain all the states of $\net$ (along with some more auxiliary states).
    \item In $\net_{fwd}^C$, starting from a configuration with a single node labelled by a state $q_{init} \notin \net$, we can reach a configuration $C'$ where the root is labelled by some state of $\net$ iff $C' = C$.
    \item In $\net_{bwd}^C$, starting from any configuration $C'$, we can reach a configuration with the root labelled by a state $q_f \notin \net$ iff $C'$ covers $C$.
\end{itemize}

Equipped with these two gadgets, it is now clear how to reduce coverability to simple coverability: Given an instance $(\net,C,C')$ of coverablity, we construct $\net_{fwd}^C,
\net_{bwd}^{C'}$ and then take our new $k$-NRCS to be $\net'$ which first executes
$\net_{fwd}^C$, then upon reaching a configuration where the root is labelled by some state in $\net$ switches to executing $\net$ and then
finally decides to non-deterministically switch to executing $\net_{bwd}^{C'}$ at some point.
We then take our initial and final configurations to be trees with one node
labelled by $q_{init}$ and $q_f$ respectively. By the above properties, it is clear 
that in $\net'$ the initial configuration can cover the final configuration
iff in $\net$ $C$ can cover $C'$.

All that is needed is to construct the above gadgets $\net_{fwd}^C$ and $\net_{bwd}^C$.
To this end, let $B_1,B_2,\cdots,B_n$ be the branches of $C$. Note that all of these branches share the root, whose state we shall take to be $q$.
Let $B_1',B_2',\cdots,B_n'$ be the same as these branches, except the label of the root
is now $(q',1),\dots,(q',n)$. In $\net_{fwd}^C$, we will have the following transitions,
which successively add the branches $B_1,B_2,\cdots,B_n$
\begin{align*}
(q_{init}) \xrightarrow{u} B'_1 \qquad  
(q',i) \xrightarrow{u} B'_{i+1} \text{ for every } i \qquad
(q',n) \xrightarrow{u} q'
\end{align*}

Similarly, we can construct $\net_{bwd}^C$ to have the following transitions,
which successively remove the branches $B_1,B_2,\cdots,B_n$
\begin{align*}
q' \xrightarrow{u} (q',n) \qquad
B'_{i+1} \xrightarrow{u} (q',i) \text{ for every } i \qquad
B'_1 \xrightarrow{u} q_f
\end{align*}
From the construction, it is clear that $\net_{fwd}^C$ and $\net_{bwd}^C$ satisfy their required properties.

\subsection{Second Ingredient: Lossy semantics of $k$-NRCS}\label{subsubsec:lossy-semantics}

The next ingredient that we shall require is a new semantics for $k$-NRCS
called the \emph{lossy semantics}. This is similar to the lossy semantics for $k$-NCS~\cite{DeckerT16, Balasubramanian21} and is explained below.

Let $\net = (Q,\delta_u,\delta_r)$. A \emph{lossy step} of $\net$ between two configurations
$C$ and $D$ is said to occur if there exists $C'$ such that $C \ge_{is} C'$ and $C' \act{} D$.
Intuitively, this means that we are first allowed to drop some subtrees from $C$ to reach a configuration $C'$ and then take an usual step from there to $D$.
Using this notion of a lossy step relation, we can now define runs, reachability and coverability
in the same manner as was done using the usual $\act{}$ step relation.

Because of the compatibility property of $k$-NRCS (Proposition~\ref{prop:compatibility}), it immediately follows that a configuration $C$ can cover a configuration $C'$ in the usual semantics of $\net$ iff $C$ can cover $C'$ in the lossy semantics of $\net$. This observation will be useful for our reduction.

\subsection{Third Ingredient: Known reduction between $k$-NCS and $k$-depth broadcast networks}

 The final ingredient that we shall need is a known reduction between coverability in $k$-\emph{NCS}, i.e., $k$-NRCS without any reset transitions and coverability in $k$-depth broadcast networks. To this end, let $\mathcal{N} = (Q,\delta_u)$ be a $k$-NCS, (note that it has no resets). In~\cite{Balasubramanian21}, it has been shown, how to start from $\mathcal{N}$
 and construct a $k$-depth broadcast protocol $\mathcal{P}$ that satisfies certain properties.
 We will now describe these properties, giving exactly enough detail that will be sufficient for our purposes.

 \paragraph{States and configurations} $\mathcal{P}$ will have states and configurations
 that satisfy the following properties:
\begin{itemize}
    \item For every state $q$ of $\mathcal{N}$, $\mathcal{P}$ has states
    of the form $(q,i)$ for every $0 \le i \le k$. Furthermore, it also has 
    states $(start,i)$ and $(finish,i)$ for every $0 \le i \le k$. Such states
    are called the good states of $\mathcal{P}$. Note that each one of these two states
    is a pair, where the first element (called the base) belongs to $Q \cup \{start,finish\}$ and the second element (called the grade) belongs to $\{0,1,\cdots,k\}$.
    ($\mathcal{P}$ will also have other auxiliary states, but they are not relevant for our discussion here).
    \item A good configuration of $\mathcal{P}$ is any configuration $C$ such that 1) $C$ is a tree of height at most $k$, 2)  $C$ is labelled only by good states, 3) the label of the root of $C$ has grade 0 and, 4) every node
    at distance $i$ from the root has a label of grade $i$.
    \item Every good configuration $C$ of $\mathcal{P}$ can be uniquely mapped
    to a configuration $\mathbb{E}(C)$ of $\mathcal{N}$ as follows: First, 
    remove all the grades in all the labels of $C$ and just keep the bases.
    Next, remove all the nodes from $C$ containing $start$ or $finish$ as their bases
    and from the resulting forest, pick the tree containing the root. 
    In this way, we get a unique configuration $\mathbb{E}(C)$ of $\mathcal{N}$.
\end{itemize}

Having described the states and configurations of $\mathcal{P}$ that are of interest to us,
we now move on to describing the transitions of $\mathcal{P}$, which are defined relative
to transitions of $\mathcal{N}$.

\paragraph{Transitions} Let $t = (p_0,\dots,p_i) \xdashrightarrow{u} (q_0,\dots,q_j)$ be a transition of $\mathcal{N}$ with $w = \max(i,j)$. For the sake of uniformity, if $i < j$, we let $p_\ell = start$ for every $i < \ell \le j$. Similarly, if $i > j$, we let $q_\ell = finish$ for every $j < \ell \le i$.
For every transition $t$ of $\mathcal{N}$, $\mathcal{P}$ has a subset of transitions $S_t$ over an alphabet
$\{begin^t_\ell, end^t_\ell : 0 \le \ell \le w\}$. Any transition in any subset $S_t$ will always ensure the following properties.
\begin{itemize}
    \item Any transition of $S_t$ will never change the grade of a node, i.e., if an agent was labelled
    with some grade $g$ before firing a transition, then its grade will remain $g$ after firing.
    \item No agent at a state with base $finish$ can ever broadcast a message.
\end{itemize}

For the transition $t$, we say that a run between two good configurations $\gamma$ and $\gamma'$ of $\mathcal{P}$ is $S_t$-good iff there exists a path $v_0,\dots,v_w$ from the root in $\gamma$ labelled by $(p_0,0),\cdots,(p_w,w)$ such that the run between $\gamma$ and $\gamma'$ is of the form  
$$\gamma \act{v_0, begin^t_0} \gamma_1 \act{v_1, begin^t_1} \cdots \gamma_w \act{v_w, begin^t_w} \beta$$
and 
$$\beta \act{v_w, end^t_w} \beta_w \act{v_{w-1},end^t_{w-1}} \beta_{w-1} \cdots \beta_1 \act{v_0, end^t_0} \gamma'$$

Having described the transitions of $\mathcal{P}$, we now move on to stating the following two 
\emph{simulation properties}, which were proven in~\cite{Balasubramanian21}.

\begin{itemize}
    \item Suppose $C \act{t} C'$ is a step in $\mathcal{N}$. Let $\gamma$ be any good configuration of $\mathcal{P}$ such that $\mathbb{E}(\gamma) = C$ and $\gamma$ has a path $v_0,\dots,v_w$ from the root in $\gamma$ labelled by $(p_0,0),\cdots,(p_w,w)$. Then, there is a $S_t$-good run from $\gamma$ to $\gamma'$ such that $\mathbb{E}(\gamma') = C'$ and the only difference between $\gamma$ and $\gamma'$
    is that the nodes $v_0,\dots,v_w$ are now labelled by $(q_0,0),\cdots,(q_w,w)$.
    \item Suppose $\gamma$ is a good configuration and $\gamma \act{*} \gamma'$ is a run such that
    $\gamma'$ is the first configuration after $\gamma$ along the run in which the root is labelled by 
    a state with base $q \in Q$. Then, $\gamma'$ is a good configuration and the run between $\gamma$
    and $\gamma'$ is a $S_t$-good run for some transition $t$. Further, we also have
    that there is a lossy step between $\mathbb{E}(\gamma)$ and $\mathbb{E}(\gamma')$ in $\mathcal{N}$.
\end{itemize}

\subsection{Our reduction from $k$-NRCS to $k$-depth broadcast networks}

Using all the three ingredients from above, we now present our reduction from coverability in $k$-NRCS to coverability in $k$-depth broadcast networks.
By the first two ingredients it suffices to reduce from the simple coverability problem for 
$k$-NRCS over the lossy semantics.

Let $\mathcal{N} = (Q,\delta_u,\delta_r)$ be a $k$-NRCS. We convert $\mathcal{N}$ into a $k$-NCS $\mathcal{N'}$ by taking every reset transition $t := (p_0,\dots,p_i) \xdashrightarrow{r, p} (q_0,\dots,q_i)$
and changing it into $t_{rename} := (p_0,\dots,p_i) \xdashrightarrow{u} (q_0,\dots,q_i)$.
We now use the reduction mentioned in the third ingredient to get
a protocol $\mathcal{P'}$ for $\mathcal{N'}$ that satisfies the simulation properties.

Recall that, for every transition $t'$ of $\mathcal{N'}$, $\mathcal{P'}$ has a subset of transitions $S_{t'}$ over the alphabet $\{begin^t_\ell, end^t_\ell : 0 \le \ell \le w\}$. We now modify $\mathcal{P'}$
as follows: For every transition of the form $t_{rename}:= (p_0,\dots,p_i) \xdashrightarrow{u} (q_0,\dots,q_i)$ (obtained form the reset transition $t := (p_0,\dots,p_i) \xdashrightarrow{r, p} (q_0,\dots,q_i)$), we consider the subset of transitions $S_{t_{rename}}$ and we modify it by adding the following reception transition to it $(p,i+1) \act{?end^{t_{rename}}_{i}} (finish,i+1)$.
We call this modified subset as $S_t$ and call this new protocol $\mathcal{P}$. 

Similar to the previous subsection, for any transition $t$ of $\mathcal{N}$, we can define
the notion of an $S_t$-good run between good configurations $\gamma$ and $\gamma'$ of $\mathcal{P}$.
More precisely, let $t$ be any transition of $\mathcal{N}$. 
If $t = (p_0,\dots,p_i) \xdashrightarrow{u} (q_0,\dots,q_j)$ is an update transition let $w = \max(i,j)$.
If $t = (p_0,\dots,p_i) \xdashrightarrow{r,p} (q_0,\dots,q_i)$ is a reset transition let $w = i$.
Recall that we use $p_{i+1},\dots,p_w$ to denote $start$ and $q_{j+1},\dots,q_w$ to denote $finish$.
With this in mind,  we say that a run between two good configurations $\gamma$ and $\gamma'$ of $\mathcal{P}$ is $S_t$-good iff there exists a path $v_0,\dots,v_w$ from the root in $\gamma$ labelled by $(p_0,0),\cdots,(p_w,w)$ such that the run between $\gamma$ and $\gamma'$ is of the form  
$$\gamma \act{v_0, begin^t_0} \gamma_1 \act{v_1, begin^t_1} \cdots \gamma_w \act{v_w, begin^t_w} \beta$$
and 
$$\beta \act{v_w, end^t_w} \beta_w \act{v_{w-1},end^t_{w-1}} \beta_{w-1} \cdots \beta_1 \act{v_0, end^t_0} \gamma'$$

We now prove various properties of $\mathcal{P}$.

\subsubsection{Completeness properties of $\mathcal{P}$}

Using the first simulation property of $\mathcal{P'}$ we can prove the following two \emph{completeness} properties for $\mathcal{P}$.
\begin{itemize}
    \item Suppose $C \act{t} C'$ is a step in $\mathcal{N}$ for some update transition $t = (p_0,\dots,p_i) \xdashrightarrow{u} (q_0,\dots,q_j)$ with $w = \max(i,j)$. 
    Let $\gamma$ be any good configuration of $\mathcal{P}$ such that $\mathbb{E}(\gamma) = C$ and $\gamma$ has a path $v_0,\dots,v_w$ from the root in $\gamma$ labelled by $(p_0,0),\cdots,(p_w,w)$. Then, there is an $S_t$-good run from $\gamma$ to $\gamma'$ such that $\mathbb{E}(\gamma') = C'$ and the only difference between $\gamma$ and $\gamma'$ is that the nodes $v_0,\dots,v_w$ are now labelled by $(q_0,0),\cdots,(q_w,w)$.
    \item Suppose $C \act{t} C'$ is a step in $\mathcal{N}$ for some reset transition $t = (p_0,\dots,p_i) \xdashrightarrow{r,p} (q_0,\dots,q_i)$. Let $\gamma$ be any good configuration of $\mathcal{P}$ such that $\mathbb{E}(\gamma) = C$ and $\gamma$ has a path $v_0,\dots,v_i$ from the root in $\gamma$ labelled by $(p_0,0),\cdots,(p_i,i)$. Then, there is an $S_{t}$-good run from $\gamma$ to $\gamma'$ such that $\mathbb{E}(\gamma') = C'$ and the only difference between $\gamma$ and $\gamma'$
    is that the nodes $v_0,\dots,v_i$ are now labelled by $(q_0,0),\cdots,(q_i,i)$ and all the children
    of $v_i$ that were labelled by $(p,i+1)$ are now labelled by $(finish,i+1)$.
\end{itemize}

The first one is immediate: If $C \act{t} C'$ is a step in $\mathcal{N}$ for some update transition $t$,
then $C \act{t} C'$ is also a step in $\mathcal{N'}$. By the first simulation property of $\mathcal{P'}$,
we get the desired $S_t$-good run between $\gamma$ and $\gamma'$. Since the subset $S_t$ is unchanged between $\mathcal{P'}$ and $\mathcal{P}$, it follows that the same run is also an $S_t$-good run in $\mathcal{P}$.

The second one is proven as follows: Suppose $C \act{t} C'$ is a step in $\mathcal{N}$ for some reset transition $t = (p_0,\dots,p_i) \xdashrightarrow{r,p} (q_0,\dots,q_i)$. Let $v_0,\dots,v_i$ be the path in $C$ where $t$ was applied.
Notice that $C \act{t_{rename}} C''$ is a step in $\mathcal{N'}$ where $C''$ is the same as $C$, except there could be children of $v_i$ labelled by $p$ in $C''$. By the first simulation property of $\mathcal{P'}$, we get an $S_{t_{rename}}$-good run between $\gamma$ and some configuration $\gamma'$
as follows:
$$\gamma \act{v_0, begin^t_0} \gamma_1 \act{v_1, begin^t_1} \cdots \gamma_i \act{v_i, begin^t_i} \beta$$
and 
$$\beta \act{v_i, end^t_i} \beta_i \act{v_{i-1},end^t_{i-1}} \beta_{i-1} \cdots \beta_1 \act{v_0, end^t_0} \beta_0 := \gamma'$$

If we fire the exact same steps from $\gamma$ in $\mathcal{P}$, then we get the following $S_t$-good run
$$\gamma \act{v_0, begin^t_0} \gamma_1 \act{v_1, begin^t_1} \cdots \gamma_i \act{v_i, begin^t_i} \beta$$
and 
$$\beta \act{v_i, end^t_i} \beta_i' \act{v_{i-1},end^t_{i-1}} \beta_{i-1}' \cdots \beta_1' \act{v_0, end^t_0} \beta_0' := \gamma''$$

where each $\beta_\ell'$ is the same as $\beta_\ell$ except that all the children of $v_i$ that were labelled by $(p,i+1)$ in $\beta_\ell$  are now labelled by $(finish,i+1)$. It follows then that
we get a $S_t$-good run between $\gamma$ and $\gamma'$ where $\mathbb{E}(\gamma') = C'$, which completes the proof of the second case.

\subsubsection{Soundness properties of $\mathcal{P}$}

By a similar argument as the previous part, using the second simulation property of $\mathcal{P'}$, we can prove the following \emph{soundness} property for $\mathcal{P}$.
\begin{quote}
     Suppose $\gamma$ is a good configuration and $\gamma \act{*} \gamma'$ is a run such that
    $\gamma'$ is the first configuration after $\gamma$ along the run in which the root is labelled by 
    a state with base $q \in Q$. Then, $\gamma'$ is a good configuration and the run between $\gamma$
    and $\gamma'$ is a $S_{t}$-good run for some transition $t$ of $\mathcal{N}$. Furthermore, 
    there is a lossy step between $\mathbb{E}(\gamma)$ and $\mathbb{E}(\gamma')$ in $\mathcal{N}$ using $t$.
\end{quote}

Indeed, let $\gamma \act{v_0, a_0} \gamma_1 \act{v_1, a_1} \dots \act{v_n, a_n} \gamma_n := \gamma'$ be such a run in $\mathcal{P}$. Compared to $\mathcal{P'}$, in $\mathcal{P}$ we have only added
reception transitions, which could lead an agent to a state with base $finish$, from which no broadcasts
can ever be made. Hence, it must also be possible to get a run
$\gamma \act{v_0, a_0} \gamma_1' \act{v_1, a_1} \dots \act{v_n, a_n} \gamma_n' := \gamma''$
in $\mathcal{P'}$. By the second simulation property of $\mathcal{P'}$, this is a $S_{t'}$-good run
for some transition $t'$ of $\mathcal{N'}$ such that there is a lossy step between $\mathbb{E}(\gamma)$ and $\mathbb{E}(\gamma'')$ in $\mathcal{N'}$ using $t'$. We consider two cases:
\begin{itemize}
    \item Suppose $t' = t$ for some $t \in \mathcal{N}$. Hence, we have changed nothing in the sets
    $S_{t'}$ between $\mathcal{P'}$ and $\mathcal{P}$. Therefore, each $\gamma_\ell' = \gamma_\ell$
    and so the original run was also an $S_t$-good run of $\mathcal{P}$. Furthermore, this also implies that there is a lossy step between $\mathbb{E}(\gamma)$ and $\mathbb{E}(\gamma')$ in $\mathcal{N}$ using $t$.
    \item Suppose $t' = t_{rename}$ for some reset transition $$t = (p_0,\dots,p_i) \xdashrightarrow{r, p} (q_0,\dots,q_i) \in \mathcal{N}$$ Notice then that the only difference between $\gamma_\ell'$ and $\gamma_\ell$ might be that some states that were labelled by $(p,i+1)$ in $\gamma'_\ell$ are now labelled by $(finish,i+1)$. Hence, this implies that the original run was an $S_t$-good run of $\mathcal{P}$. Furthermore, this also implies that there is a lossy step between $\mathbb{E}(\gamma)$ and $\mathbb{E}(\gamma')$ in $\mathcal{N}$ using $t$.
\end{itemize}

This completes the proof of the soundness property.

\subsubsection{Wrapping up}

Using these two properties, the following immediately follows~\cite[Theorem 12]{Balasubramanian21}.
\begin{proposition}\label{prop:iff-runs}
    A state $q_{init}$ can cover a state $q_f$ in $\mathcal{N}$ (in the sense of simple coverability) iff the state $(q_f,0)$
    can be covered in the broadcast network $\mathcal{P}$ from some good initial configuration, i.e., a good configuration where the root is labelled by $(q_{init},0)$ and all children at level $i$ are labelled by $(start,i)$.
\end{proposition}

\begin{proof}
    The right-to-left implication follows by taking a run in $\mathcal{P}$ covering $(q_f,0)$, splitting this run into parts where the root has a label with base $q \in Q$
    and then applying the completeness properties of $\mathcal{P}$ to each of those parts. 
    (Recall that lossy coverability is the same as coverability).

    The left-to-right implication is true by the following argument: Take a run
    $C_0 \act{t_1} C_1 \act{t_2} \dots \act{t_N} C_N$
    covering $q_f$ from $q_{init}$ in $\mathcal{N}$. 
    Take a good initial configuration $\gamma_0$ such that each node has $N$ children. By an induction on $i$ from 0 to $N$ and using the soundness property of $\mathcal{P}$, we can construct
    a sequence of configurations $\gamma_0,\gamma_1,\dots,\gamma_N$ satisfying the following properties:
    \begin{itemize}
        \item $\mathbb{E}(\gamma_i) = C_i$
        \item For each maximal branch $B = (p_0,\dots,p_j)$ of $C_i$, there
        is a unique branch $((p_0,0),\dots,(p_j,j))$ of $\gamma_i$ and from
        the node labelled with $(p_j,j)$, there are at least $(N-\ell)$ different paths
        labelled by $(start,j+1),\dots,(start,k)$.
        \item Each $\gamma_i \act{*} \gamma_{i+1}$.
    \end{itemize}
    This concludes the proof of the left-to-right implication and therefore of the proposition as well.
\end{proof}

Hence, we have given the desired reduction from simple coverability of $k$-NRCS over the lossy semantics to coverability of $k$-depth broadcast networks. This completes the proof of the desird lower bound for $k$-depth broadcast networks.

\section{Proofs for Subsection~\ref{subsec:freeze-LTL}}
\label{sec:appendix-freeze-LTL}
\freezeLTL*

In order to prove this theorem, we first set up the formal syntax and semantics
of freeze LTL with $k$-ordered attributes. 

\subsection{Syntax and semantics of freeze LTL}
\label{subsec:appendix-syntax-freeze-LTL}

Fix some $k \in \mathbb{N}$, a finite set of atomic propositions $\Sigma$ 
and an infinite set $\mathbb{D}$. A $k$-attributed data letter (or simply data letter) is a tuple $(A,d_1,\dots,d_k)$ where $A \subseteq \Sigma$ and each $d_i \in \mathbb{D}$. A data word is a sequence of data letters. 

The syntax of freeze LTL with $k$-ordered attributes is all formulas generated according to the following grammar
$$\varphi ::= a  \ | \ \lnot \varphi \ | \ \varphi \land \varphi \ | \ X \ \varphi \ | \ \varphi \ U \ \varphi \ | \ G \varphi \ | \ F \varphi \ | \ \downarrow^i \varphi \ | \ \uparrow^i $$
\newcommand{\data}{\mathbf{d}}
where $a \in \Sigma$ and $1 \le i \le k$. The semantics of such formulas is defined 
in terms of data words as follows. 

We say that $\mathbf{d}$ is a \emph{partial data valuation} if $\mathbf{d}$ is a function from $\{1,\dots,k\}$ to $\mathbb{D} \cup \{\bot\}$
where $\bot$ is a special element not in $\mathbb{D}$. Intuitively a partial data valuation is like a $k$-tuple of values from $\mathbb{D}$, except that some of the values
can now be undefined ($\bot$). A complete data valuation
is a partial data valuation which does not use $\bot$ in its range. Given any complete
data valuation $\data$ and any $i \le k$, we can get the partial data valuation
$\data_{|i}$ as the mapping such that $\data_{|i}(j) = \data(j)$ if $j \le i$ 
and $\bot$ otherwise.

Now, let $w = (A_1,\mathbf{d_1}), (A_2,\mathbf{d_2}), \dots, (A_n,\mathbf{d_n})$ be a data word where each $\mathbf{d_j} = (d_j^1,\dots,d_j^k)$. Let $1 \le i \le n$ and $\mathbf{d}$ be a partial data valuation. We define the semantics of $(w,i,\mathbf{d})$ satisfying a formula $\varphi$ as follows

\begin{align*}
    (w,i,\data) & \models a   && \iff \qquad   a \in A_i\\ 
    (w,i,\data) & \models \lnot \varphi && \iff \qquad   (w,i,\data) \not\models \varphi\\ 
    (w,i,\data) & \models \varphi_1 \land \varphi_2 && \iff \qquad   (w,i,\data) \models \varphi_1 \text{ and } (w,i,\data) \models \varphi_2\\
    (w,i,\data) & \models X \ \varphi && \iff \qquad   (w,i+1,\data) \models \varphi\\
    (w,i,\data) & \models \varphi_1 \ U \ \varphi_2 && \iff \qquad   \exists  j \in [i,n] \text{ s.t. } (w,j,\data) \models \varphi_2 \text{ and } \\ &&& \qquad \qquad \qquad \forall \ell \in [i,j-1], (w,\ell,\data) \models \varphi_1\\
    (w,i,\data) & \models G \ \varphi && \iff \qquad  \forall \ell \in [i,n],  (w,\ell,\data) \models \varphi\\
    (w,i,\data) & \models F \ \varphi && \iff \qquad  \exists  \ell \in [i,n],  (w,\ell,\data) \models \varphi\\
    (w,i,\data) & \models \downarrow^x \varphi && \iff \qquad (w,i,\data_{i|x}) \models \varphi\\
    (w,i,\data) & \models \uparrow^x && \iff \qquad \data_{|x} = \data_{i|x}\\
\end{align*}

Given a freeze LTL formula with $k$-ordered attibutes, we say that it is satisfiable if there is some data
word $w$ such that $(w,0,\mathbf{\bot})$ satisfies it. Here $\mathbf{\bot}$ is the
partial data valuation that maps everything in $\{1,\dots,k\}$ to $\bot$.
The satisfiability problem is to decide if a given freeze LTL formula is satisfiable.

\subsection{Upper Bound for Satisfiability of Freeze LTL}
\label{subsec:appendix-upper-bound-freeze-LTL}

We now show an $\fastf{\Omega_k}$ upper bound for the satisfiability problem. To this end, we use a result of~\cite[Theorem 12]{DeckerT16} which showed that satisfiability for any formula in freeze LTL with $k$-ordered attributes can be converted to coverability for a $(k+1)$-NCS, i.e., a $(k+1)$-NRCS without any reset transitions.
The particular $(k+1)$-NCS that they construct has the property that
any configuration $C$ reachable from the initial configuration $C_{init}$ has 
at most 2 children of the root~\cite[Sections C.2 and C.3]{DeckerT15-fullversion}.  We now show that any such $(k+1)$-NCS can be simulated by a $k$-NRCS as follows. The intuition is that in the $k$-NRCS that we construct, by means of resets, the root will simply simulate the 2 children without explicitly keeping them. This will allow us to save a level, thereby giving us a $k$-NRCS. We now explain this in detail.

Let $\mathcal{N} = (Q,\delta_u)$ be a $(k+1)$-NCS satisfying the property that in 
any configuration $C$ reachable from a configuration $C_{init}$, there are at most 2 children of the root. We will now construct a $k$-NRCS $\mathcal{N'} = (Q',\delta_u',\delta_r')$ 
as follows: Let $\bot, \checkmark_1, \checkmark_2$ be elements not in $Q$ and let $Q_\bot = Q \cup \{\bot\}, M = \{\checkmark_1,\checkmark_2\}$.
We set $Q' = (Q \times Q_\bot \times Q_\bot) \cup (Q \times M) \cup Q$.  
Before we define the transitions $\delta'$ of $\mathcal{N'}$, we set up some notation.

\subsubsection{Good configurations}

A configuration $C'$ of $\mathcal{N'}$ is called good if the following conditions are satisfied:
\begin{itemize}
    \item The label of the root is of the form $(q,a,b) \in Q \times Q_\bot \times Q_\bot$.
    Intuitively, the $a,b$ keep track of the role of the 2 children of the root in $\mathcal{N}$.
    \item The labels of the children of the root are of the form $(q,m) \in Q \times M$. Intuitively, if $m = \checkmark_i$, then this node would have been a child
    of the $i^{th}$ original child of the root in $\mathcal{N}$.
    \item The labels of all of the other nodes belong to $Q$.
    \item Suppose the label of the root is $(q,a,b)$
    \begin{itemize}
        \item If $a = b = \bot$ then the root has no children.
        \item If $a \neq \bot, b = \bot$, then the label of each child of the root is of the form $(q',\checkmark_1)$ for some $q' \in Q$. Similarly if $a = \bot, b \neq \bot$,
        then the label of each child of the root is of the form $(q',\checkmark_2)$.
    \end{itemize}
\end{itemize}

For each good configuration $C'$ of $\mathcal{N'}$,
we can assign a configuration $C := \mathbb{E}(C')$ of $\mathcal{N}$ as follows.
Let $(q,a,b)$ be the label of the root $r$.
\begin{itemize}
    \item If $a = b = \bot$, then $C = C'$, except the label of $r$ is now just $q$.
    \item Suppose exactly one of $a, b$ is $\bot$ (Without loss of generality, we can assume $b$ is $\bot$. The construction is the same for the other case, except $\checkmark_1$ below has to be replaced by $\checkmark_2$). Let $v_1,\dots,v_n$ be the children of the root $r$
    with labels $(q_1,\checkmark_1),\dots,(q_n,\checkmark_1)$. 
    Then we construct $C$ as follows:
    We first add a new child $v$ to $r$ and replace every edge $(r,v_i)$
    with $(v,v_i)$. 
    Then, we label the root as $q$, label the new child $v$ as $a$,
    label each node $v_i$ as $q_i$ and
    keep the other nodes as it is.
    \item Suppose $a, b \in Q$.
    Let $u_1,\dots,u_m$ be the children of the root $r$ with labels
    $(p_1,\checkmark_1),\dots,(p_m,\checkmark_1)$ and $v_1,\dots,v_n$ be the children of the root $r$
    with labels $(q_1,\checkmark_2),\dots,(q_n,\checkmark_2)$. Then we construct $C$ as follows:
    We first add two new children $u,v$ to $r$ and replace edges $(r,u_i)$
    with $(u,u_i)$ and edges $(r,v_i)$ with $(v,v_i)$ respectively. 
    Then, we label the root as $q$, label the new children $u,v$ as $a,b$ respectively,
    label each node $u_i$ as $p_i$, label each node $v_i$ as $q_i$ and
    keep the other nodes as it is.
\end{itemize}

\subsubsection{Simulation properties}

We will now specify the transitions of $\mathcal{N'}$. Before that,
we first define a new type of transition called \emph{generalized reset transition}.
A generalized reset transition is a transition of the form
$$t := (p_0,\dots,p_i) \xdashrightarrow{gr, S} (q_0,\dots,q_i)$$
where $S$ is any subset of states. The semantics of such a transition $t$ is given as follows:
We say that $C \act{t} C'$ if there is a path $v_0,\dots,v_i$ in $C$ starting at the root
such that for each $0 \le \ell \le i$, $v_\ell$ is labelled by $p_\ell$ and $C'$ is obtained
from $C$ by 1) for every $0 \le \ell \le i$, changing the label of $v_\ell$ to $q_\ell$
and 2) removing every subtree rooted at a child $v$ of $v_i$ such that $v$ has a label from $S$.
Intuitively, a generalized reset transition allows us to perform a reset on a subset $S$ of labels simultaneously. 

From the semantics, it is clear that any generalized reset transition can be simulated by a series of reset transitions; for instance if $t = (p_0,\dots,p_i) \xdashrightarrow{gr, S} (q_0,\dots,q_i)$ with $S = \{s_1,\dots,s_m\}$, then this can be replaced by 
$m$ transitions $t_1,\dots,t_m$ with 
$$t_1 = (p_0,\dots,p_i) \xdashrightarrow{r, s_1} (t^1,\dots,t^1)$$
$$t_{j+1} = (t^j,\dots,t^{j}) \xdashrightarrow{r, s_{j+1}} (t^{j+1},\dots,t^{j+1})$$
for every $1 \le j < m-1$ and then finally 
$$t_m = (t^{m-1},\dots,t^{m-1}) \xdashrightarrow{r, s_{m}} (q_0,\dots,q_i)$$

For the sake of conciseness, in our definition of $\mathcal{N'}$, we will use generalized reset transitions. For each transition $t$ of $\mathcal{N}$, we will construct a subset of transitions $S_t$ of $\mathcal{N'}$ such that the following simulation properties hold:
\begin{itemize}
    \item Suppose $C \act{t} D$ is a step in $\mathcal{N}$. Then 
    for any $C'$ such that $\mathbb{E}(C') = C$, there exists $D'$
    such that $\mathbb{E}(D') = D$ and $C' \act{t'} D'$ for some $t' \in S_t$
    \item Suppose $C' \act{t'} D'$ is a step in $\mathcal{N'}$ 
    for some $t' \in S_t$ such that $C'$
    is a good configuration. Then $D'$ is a good configuration and
    $\mathbb{E}(C') \act{t} \mathbb{E}(D')$ is a step in $\mathcal{N}$.
\end{itemize}

These two properties ensure that a configuration $C_f$ can be covered in $\mathcal{N}$ from $C_{init}$ iff 
some configuration $C'_f$ such that $\mathbb{E}(C'_f) = C_{f}$ can be covered from some configuration $C'_{init}$ such that $\mathbb{E}(C'_{init}) = C_{init}$ in $\mathcal{N'}$.
For any configuration $C$, note that there are only a fixed number of configurations
$C'$ such that $\mathbb{E}(C') = C$. Hence, these two properties would then show that coverability in $\mathcal{N}$ can be solved by a constant number of coverability queries in $\mathcal{N'}$, which we know is in $\fastf{\Omega_k}$, thereby giving the required bound
Hence, it only suffices to construct transitions $\delta'$ that satisfy the above simulation properties.

\subsubsection{Transitions of $\mathcal{N'}$}

Let $t = (p_0,p_1,\dots,p_i) \xdashrightarrow{u} (q_0,q_1,\dots,q_j)$ be a transition of the $(k+1)$-NRCS $\mathcal{N}$. Without loss of generality we can assume that if $j = 0$ then $i = 1$. Indeed, if $j = 0$ and $i > 1$ then $t$ can be replaced by $t' = (p_0,p_1,\dots,p_i) \xdashrightarrow{u} (p_0',p_1',p_2',\dots,p_i')$ and $t'' = (p_0',p_1') \xdashrightarrow{u} (p_0')$. 

Recall the intuition that, in $\mathcal{N'}$, the root will also simulate the role of its children (which in this case corresponds to $p_1, q_1$).
Corresponding to this intuition, $\mathcal{N'}$ will have the following transitions $S_t$, which we divide into multiple cases.

\paragraph{Case 1: $i = 0$} In this case, we would have the following transitions $S_t$:
\begin{equation}
    ((p_0,a,\bot)) \xdashrightarrow{u} ((q_0,a,q_1), (q_2,\checkmark_2), q_3, \dots, q_j)
\end{equation}
\begin{equation}
    ((p_0,\bot,b)) \xdashrightarrow{u} ((q_0,q_1,b), (q_2,\checkmark_1), q_3, \dots, q_j)
\end{equation}
Further if $j = 0$ we would also have
\begin{equation}
    ((p_0,a,b)) \xdashrightarrow{u} ((q_0,a,b))
\end{equation}
It can be easily seen that $S_t$ satisfies the required simulation properties in this case.

\paragraph{Case 2: $i \ge 1, j \ge 1$} In this case, we would have the following transitions $S_t$:
\begin{equation}
    ((p_0,p_1,b),(p_2,\checkmark_1),p_3,\dots,p_i) \xdashrightarrow{u} ((q_0,q_1,b), (q_2,\checkmark_1), q_3, \dots, q_j)
\end{equation}
\begin{equation}
    ((p_0,a,p_1),(p_2,\checkmark_2),p_3,\dots,p_i) \xdashrightarrow{u} ((q_0,b,q_1), (q_2,\checkmark_2), q_3, \dots, q_j)
\end{equation}

It can be easily seen that $S_t$ satisfies the required simulation properties in this case.

\paragraph{Case 3: $i \ge 1, j = 0$} In this case, as remarked before, we can assume that $i = 1$. In this case, we would have the following transitions 
$S_t$:
\begin{equation}
    ((p_0,p_1,b),(p_2,\checkmark_1)) \xdashrightarrow{gr, (Q \times \{\checkmark_1\})} ((q_0,\bot,b))
\end{equation}
\begin{equation}
    ((p_0,a,p_1),(p_2,\checkmark_2)) \xdashrightarrow{gr, (Q \times \{\checkmark_2\})} ((q_0,a,\bot))
\end{equation}

It can be easily seen that $S_t$ satisfies the required simulation properties in this case.

Hence, this completes the construction of the transitions of $\mathcal{N'}$ satisfying the simulation properties. As discussed before, this is sufficient to reduce coverability in $\mathcal{N}$ to a constant number of instances of coverability in $\mathcal{N'}$.
This shows the required upper bound.

\subsection{Lower Bound for Satisfiability of Freeze LTL}
\label{subsec:appendix-lower-bound-freeze-LTL}

We now show an $\fastf{\Omega_k}$ lower bound for the satisfiability problem.
To this end, we use a result of~\cite[Theorem 13]{DeckerT16} which gives a reduction
from coverability in $k$-NCS to satisfiability of freeze LTL with $k$-ordered attributes.
This result uses an encoding of runs of a $k$-NCS by means of data words.
We recall that enocding here and show that it can be adapted to also handle
runs of $k$-NRCS.

Let $\mathcal{N} = (Q,\delta_u,\delta_r)$ be a $k$-NRCS. As stated in Subsubsection~\ref{subsubsec:lossy-semantics}, coverability over the lossy semantics in $\mathcal{N}$ is the same as coverability over the usual semantics in $\mathcal{N}$ and so it suffices to only concentrate on lossy runs of $\mathcal{N}$. 
Note that any lossy run of $\mathcal{N}$ can be written down as a sequence of configurations $C_0, C_0', C_1, C_1', \dots, C_n, C_n'$ where each $C_i' \le_{is} C_i$ and each $C_i' \act{t_i} C_{i+1}$ for some transition $t_i \in \delta_u \cup \delta_r$.
We now build an encoding of such a run as a data word in a step-by-step manner as follows.

\subsubsection{Encoding a single configuration.} We say that a data word $w$ is an encoding of a configuration $C$ if it satisfies the following conditions:
\begin{itemize}
    \item $w$ contains exactly $2b$ many letters where $b$ is the number of maximal branches of $C$.
    \item Every odd letter of $w$ contains the atomic proposition $odd$
    and every even letter of $w$ contains the atomic proposition $even$.
    \item For every maximal branch $B = (q_0,\dots,q_i)$ of $C$ there are exactly two letters of $w$ given by
    $odd_B = (X,d_1,\dots,d_k)$ and $even_B = (Y,d_1',\dots,d_k')$, where $odd_B$ appears in an odd position and $even_B$ appears next to $odd_B$.
    \item The sets $X$ and $Y$ in $odd_B$ and $even_B$ contain the atomic propositions $(q_0,0),(q_1,1),\dots,(q_i,i),(-,i+1),\dots,(-,k)$. Furthermore, no other atomic
    propositions of the form $(q,j)$ are present for any $0 \le j \le k$.
    \item The data values of $odd_B$ and $even_B$ are mutually exclusive, i.e., no $d_j$
    is equal to any $d_\ell'$.
    \item For any two maximal branches $B, B'$ and any $1 \le j \le k$, 
    the $j^{th}$ data value
    of $odd_B$ and $odd_{B'}$ are the same iff they share a common prefix
    of length $j$ from the root. Similarly, for $even_B$ and $even_{B'}$.
\end{itemize}

Now, note that these conditions stay the same irrespective of whether $\mathcal{N}$ contains reset transitions or not. Therefore, if we let $\mathcal{N}_{no-r}$ be the same as $\mathcal{N}$
except it has no reset transitions, then a data word encodes a configuration of $\mathcal{N}_{no-r}$
iff it encodes the same configuration of $\mathcal{N}$. Since we know that
there is a formula of freeze LTL which is satisfied by exactly words encoding configurations~\cite[Section D.1]{DeckerT15-fullversion}, it follows that
\begin{proposition}
    There is a formula $\Phi_{Conf}$ of freeze LTL with $k$-ordered attributes such that
    $\Phi_{Conf}$ is satisfied by a word $w$ iff it is an encoding of some configuration
    of $\mathcal{N}$.
\end{proposition}

\subsubsection{Encoding a lossy step} Using the above encoding, we will now encode a lossy step of the machine $\mathcal{N}$, i.e., we will encode a tuple $(C,t,C')$
which satisfies the conditions that there exists $C'' \le_{is} C'$ such that $C'' \act{t} C$. 
(Note that $C$ is the configuration we move \textbf{to} after firing t).

For any transition $t \in \delta_u$ (i.e., a non-reset transition of $\mathcal{N}$), we say that a word $w$ encodes a tuple $(C,t,C')$ if it satisfies the following conditions:
\begin{itemize}
    \item $w$ contains at least $3$ occurrences of a letter of the form $(\{\$\},d_1,\dots,d_k)$. Such letters are called special letters.
    \item The subword of $w$ in between the first and second occurrence
    of a special letter encodes $C$. The subword of $w$ in between
    the second and third occurrence of a special letter encodes $C'$.
    \item There exists $C'' \le_{is} C'$ such that $C'' \act{t} C$. 
    Furthermore, suppose $v_0,\dots,v_j$ is the path along which $t$ was applied to in $C$.
    Then 
    \begin{itemize}
        \item If a maximal branch $B$ of $C$ contains a prefix of this path up till depth $1 \le \ell \le j$, then the letter $even_B$ in the encoding of $C$ in $w$ also contains the atomic propositions $t,\checkmark_1,\dots,\checkmark_\ell$.
        \item No other letter contains any of the atomic propositions $t,\checkmark_1,\dots,\checkmark_k$.
    \end{itemize}
\end{itemize}

By~\cite[Sections D.2 and D.3]{DeckerT15-fullversion}, we have that
\begin{proposition}
    For any transition $t \in \delta_u$, there is a formula $\Phi_{t-step}$ which is satisfied by a word $w$ iff $w$ encodes a tuple $(C,t,C')$.
\end{proposition}

Now, we show that such a formula can also be obtained for any reset transition $t = (p_0,\dots,p_i) \xdashrightarrow{r,p} (q_0,\dots,q_i)$. Corresponding to the reset transition $t$, we define a non-reset transition $t_{rename} = (p_0,\dots,p_i) \xdashrightarrow{u} (q_0,\dots,q_i)$. Note that $t_{rename}$ need not be in $\mathcal{N}$.
However, $t_{rename}$ satisfies the following useful property.

\begin{proposition}\label{prop:t-rename}
    For any $C, C'$, there exists $C''$ such that $C'' \le_{is} C'$ and $C'' \act{t} C$
    iff there exists $D''$ such that $D'' \le_{is} C'$ and $D'' \act{t_{rename}} C$
    and along the path $v_0,\dots,v_i$ in $C$ where $t_{rename}$ was applied to, there is no child of $v_i$ labelled by $p$.
\end{proposition}

\begin{proof}
    Suppose there exists $C''$ such that $C'' \le_{is} C'$ and $C'' \act{t} C$.
    Let $v_0,v_1,\dots,v_i$ be the path along which $t$ is fired in $C''$.
    We construct $D''$ by taking $C''$ and removing all the sub-trees rooted at the children of $v_i$ labelled by $p$. Then from $D''$ we fire $t_{rename}$ along the
    path $v_0,v_1,\dots,v_i$. By construction we have $D'' \act{t_{rename}} C$
    and along the path $v_0,\dots,v_i$ in $C$ where $t_{rename}$ was applied to, there is no child of $v_i$ labelled by $p$.

    Suppose there exists $D'' \le_{is} C'$ and $D'' \act{t_{rename}} C$
    and along the path $v_0,\dots,v_i$ in $C$ where $t_{rename}$ was applied to, there is no child of $v_i$ labelled by $p$. We simply take $C'' = D''$ and apply $t$ along the
    same path. In this way, it is clear that $C'' \act{t} C$ and so we are done.
\end{proof}

Motivated by this proposition, we make the following definition: For any reset transition $t$, we say that a word $w$ encodes a tuple $(C,t,C')$ if it satisfies the following conditions:
\begin{itemize}
    \item $w$ contains exactly $3$ occurrences of a letter of the form $(\{\$\},d_1,\dots,d_k)$. Such letters are called special letters.
    \item The subword of $w$ in between the first and second occurrence
    of a special letter encodes $C$. The subword of $w$ in between
    the second and third occurrence of a special letter encodes $C'$.
    \item There exists $C'' \le_{is} C'$ such that $C'' \act{t_{rename}} C$. 
    Furthermore, suppose $v_0,\dots,v_i$ is the path along which $t_{rename}$ was applied to in $C$.
    Then 
    \begin{itemize}
        \item If a maximal branch $B$ of $C$ contains a prefix of this path up till depth $1 \le \ell \le i$, then the letter $even_B$ in the encoding of $C$ in $w$ contains the atomic propositions $t_{rename},\checkmark_1,\dots,\checkmark_\ell$.
        \item If a maximal branch $B$ of $C$ contains a prefix of this path up till depth $i$, then the letter $even_B$ in the encoding of $C$ in $w$ does not contain the atomic proposition $(p,i+1)$.
        \item No other letter contains any of the atomic propositions $t_{rename},\checkmark_1,\dots,\checkmark_k$.
    \end{itemize}
\end{itemize}

From the above definition, it is clear that the following property called \emph{Property A}
holds:
\begin{quote}
    Suppose $t = (p_0,\dots,p_i) \xdashrightarrow{r,p} (q_0,\dots,q_i)$ is a reset transition. A word $w$ encodes a tuple $(C,t,C')$ iff $w$ encodes the tuple $(C,t_{rename},C')$ and for every letter $even_B$ appearing in the encoding of $C$ in $w$ containing the 
    atomic propositions $t_{rename},\checkmark_1,\dots,\checkmark_i$,
    $even_B$ must not have $(p,i+1)$ as a proposition.
\end{quote}

Therefore, using Property A and Proposition~\ref{prop:t-rename} we get,
\begin{proposition}
    For ant reset transition $t$, a word $w$ encodes a tuple $(C,t,C')$ iff there exists
    $C''$ such that $C'' \le_{is} C'$ and $C'' \act{t} C$.
\end{proposition}

Furthermore, because of Property A we can easily show that
\begin{proposition}
    For any reset transition $t$, there is a formula $\Phi_{t-step}$ which is satisfied by a word $w$ iff $w$ encodes a tuple $(C,t,C')$.
\end{proposition}

\begin{proof}
    Let $t = (p_0,\dots,p_i) \xdashrightarrow{r, p} (q_0,\dots,q_i)$.
    By Property A, $\Phi_{t-step}$ can simply be defined as the conjunction of $\Phi_{t_{rename}-step}$
    and $G((even \land t \land \bigwedge_{1 \le \ell \le i} \checkmark_\ell  )\implies \lnot (p,i+1))$.
\end{proof}

\subsubsection{Encoding lossy runs} We can now use the encoding a lossy step to encode lossy runs. Indeed, we say that a word $w$ encodes a sequence $C_n, t_n, C_{n-1}, t_{n-1}, C_{n-2}, \dots, t_1, C_0$ if
\begin{itemize}
    \item $w$ has at least $n+2$ special letters.
    \item The word between any $(i+1)^{th}$ special letter and $(i+3)^{th}$ special 
    letter encodes the tuple $(C_{n-i},t_{n-i},C_{n-i-1})$.
\end{itemize}

Note that a word $w$ can encode a sequence $C_n, t_n, C_{n-1}, \dots, t_0, C_0$ iff
there is a lossy run of $\mathcal{N}$ of the form $C_n \ge_{is} C_n' \act{t_n} C_{n-1} \ge_{is} C_{n-1}' \act{t_{n-1}} \dots \act{t_0} C_0$. 

It is easy to see that we can build a formula that recognizes only words $w$ encoding such sequences. Indeed, we just need to specify that any subword of $w$ starting at $\$$ must encode a lossy step, except for the last two $\$$, which can be done using the formula

$$\Phi_{seq} := G(\$ \implies \bigvee_{r \in \delta} \Phi_{r-step} \lor \lnot X F \$ \lor \lnot (X(F(\$ \land \lnot XF\$)))$$

\subsubsection{Reducing simple coverability of $k$-NRCS to satisfiability}

We now reduce simple coverability of $k$-NRCS
to satisfiability of freeze LTL formulas. Recall that in the simple coverability problem,
we are given a $k$-NRCS $\net$ and two states $q_{init}$ and $q_f$ and we want
to checke if starting from a configuration with a single node labelled by $q_{init}$
we can reach a configuration covering a single node labelled by $q_f$. As shown in Subsection~\ref{subsec:simple-coverability}, this is $\fastf{\Omega_k}$-hard.

Note that the answer to the simple coverability problem is yes iff there is a data word
$w$ encoding a sequence $C_n, t_n, C_{n-1} , \cdots, t_0, C_0$ with $C_n$ being a configuration where the root is labelled by $q_f$ and $C_0$ being a configuration with a single node labelled by $q_{init}$. Therefore, we just need to take the formula $\Phi_{seq}$
from the previous part and attach to it the following two constraints:
\begin{itemize}
    \item The second letter of the word encoding the sequence must represent a branch of a configuration whose root node is labelled by $q_f$ (Remember that the first letter of the word must be a special letter with $\$$)
    $$X((q_f,0))$$
    \item The penultimate letter of the word encoding the sequence must represent the initial configuration (Remember that the last letter of the word must be a special letter with $\$$)
    $$(X(\$) \land \lnot XF(\$) )\implies ((q_{init,0}) \land \bigwedge_{q \in Q, 1 \le \ell \le k} \lnot (q,\ell))$$
\end{itemize}

By the property of the formula $\Phi_{seq}$, it follows that $\Phi_{seq}$ along
with these constraints are satisfiable by a word iff the answer to the simple coverability problem is yes. This completes the required reduction
and proves the $\fastf{\Omega_k}$-hardness of satisfiability of freeze LTL with $k$-ordered attributes.

\end{document}